\documentclass[12pt,draftclsnofoot,onecolumn]{IEEEtran}
\usepackage{amsthm}
\newtheorem{definition}{Definition}[section]
\newtheorem{proposition}{Proposition}[section]
\newtheorem{theorem}{Theorem}[section]
\newtheorem{lemma}[theorem]{Lemma}
\newtheorem{corollary}[theorem]{Corollary}
\newtheorem{example}[theorem]{Example}

\usepackage{amssymb}
\usepackage{amsmath}
\usepackage{cancel}
\usepackage{algorithm}
\usepackage{algpseudocode}
\usepackage{graphicx}
\usepackage{subcaption}
\usepackage{float} 
\usepackage{enumitem}
\usepackage{tikz-cd}
\usepackage{booktabs}
\usepackage{comment}

\theoremstyle{remark}
\newtheorem{remark}{Remark}

\ifCLASSINFOpdf
\else
\fi
\begin{document}
%
\title{Dual Riemannian Newton Method\\ on Statistical Manifolds}


\author{%
\IEEEauthorblockN{Derun Zhou\IEEEauthorrefmark{1}\IEEEauthorrefmark{3},\;
Keisuke Yano\IEEEauthorrefmark{2}\IEEEauthorrefmark{3},\;
Mahito Sugiyama\IEEEauthorrefmark{1}\IEEEauthorrefmark{3},\;
~\IEEEmembership{Member,~IEEE}}\\
\IEEEauthorblockA{\IEEEauthorrefmark{1}National Institute of Informatics, Tokyo, 101-8430, Japan}\\
\IEEEauthorblockA{\IEEEauthorrefmark{2}The Institute of Statistical Mathematics, Tokyo, 190-8562, Japan}\\
\IEEEauthorblockA{\IEEEauthorrefmark{3}The Graduate University for Advanced Studies, SOKENDAI, Kanagawa, 240-0193, Japan}%
\thanks{Corresponding author: Derun Zhou (email: zhouderun@nii.ac.jp).}%
%
}




%



\IEEEtitleabstractindextext{%
\begin{abstract}
In probabilistic modeling, parameter estimation is commonly formulated 
as a minimization problem on a parameter manifold. Optimization in such spaces requires \emph{geometry-aware methods} 
that respect the underlying information structure. While the \emph{natural gradient} leverages the \emph{Fisher information metric} as a form of Riemannian gradient descent, it remains a \emph{first-order method} and often exhibits slow convergence near optimal solutions. Existing second-order manifold algorithms typically rely on the Levi-Civita connection, thus overlooking the dual-connection structure that is central to information geometry. We propose the \emph{dual Riemannian Newton method}, a Newton-type optimization algorithm on manifolds endowed with a metric and a pair of dual affine connections. The dual Riemannian Newton method explicates how duality shapes second-order updates: when the retraction (a local surrogate of the exponential map) is defined by one connection, the associated Newton equation is posed with its dual. 
We clarify its relationship with affine-connection Newton methods, establish its equivalence to natural gradient descent and mirror descent for potential-based objectives, and prove local quadratic convergence, with cubic convergence under specific $\alpha$-connections. We validate these theoretical results through experiments on representative statistical models. Thus, the dual Riemannian Newton method delivers second-order efficiency while preserving the dualistic structures underlying information-geometric learning and inference.
\end{abstract}


\begin{IEEEkeywords}
Differential geometry; information geometry; projection problem; Riemannian optimization; second-order methods; statistical inference
\end{IEEEkeywords}}

\maketitle

\IEEEdisplaynontitleabstractindextext

%
\IEEEpeerreviewmaketitle

\section{Introduction}
In classical minimization over Euclidean spaces, the negative gradient provides the steepest descent direction, which means that first-order methods such as gradient descent are ubiquitous. To accelerate convergence, higher-order schemes with superlinear or even quadratic rates, most notably Newton's method and trust-region methods, leverage second-order information to refine both directions and step sizes~\cite{sun2006}.

However, in many applications, the ambient space of the minimization is not Euclidean, but instead a manifold endowed with a nontrivial geometry. Typical examples include orthogonality/subspace constraints such as the Stiefel/Grassmann manifolds~\cite{Edelman1998}, and the cone of symmetric positive definite (SPD) matrices~\cite{Pennec2006}. In these settings, the steepest descent direction is not the Euclidean negative gradient; it is the \emph{Riemannian negative gradient}, defined by the underlying Riemannian metric. This viewpoint leads to Riemannian gradient descent~\cite{Boumal2023}. In statistical manifolds, choosing the Fisher information matrix as the metric specializes Riemannian gradient descent to the well-known \emph{natural gradient} method~\cite{Amari1998Natural} also known as Fisher's scoring \cite{Longford1987}. While Riemannian gradient descent respects the geometry, it remains a first-order scheme and may converge slowly on ill-conditioned problems (c.f.,~\cite{Absil2008}). To obtain superlinear convergence on manifolds, higher-order methods such as Riemannian Newton and Riemannian trust region methods have been developed, primarily under the Levi--Civita connection~\cite{Boumal2023}, and have seen success in minimization tasks like fixed-rank matrix completion~\cite{Vandereycken2013} and independent component analysis~\cite{sato2017riemannian}. More generally, Newton methods on manifolds can be formulated using affine connections without requiring a Riemannian metric. The Newton equation is defined through the covariant Hessian, and the iterates are updated using a suitable retraction~\cite{Absil2008,MahonyManton2002}.

Parameter estimation in probabilistic modeling is typically formulated as a minimization problem over a probabilistic model; for instance, minimizing the negative log-likelihood or a divergence between the model and data distributions. Higher-order manifold optimization methods have been widely used for parameter estimation of probabilistic models, 
although they have mainly been restricted to distributions with specific structural properties, such as Gaussian (mixture) models with special covariance structures, to reduce the number of interaction steps during optimization. These problems can be reformulated as Riemannian optimization tasks on the SPD~\cite{Dryden2009}, Stiefel, or Grassmann manifolds equipped with the Levi--Civita connection~\cite{Zhou2021,sembach2022riemannian}.
Beyond these matrix-manifold approaches, Malagò and Pistone~\cite{malago2014combinatorial} developed a Newton method for optimization over exponential families on finite sample spaces, using the Riemannian Hessian associated with the Levi--Civita connection of the Fisher information metric. Their subsequent work~\cite{malago2015second} extended the second-order geometric analysis by considering Hessians associated with the exponential and mixture connections and deriving corresponding Taylor approximations, with applications to multivariate Gaussian distributions.



From the view of information geometry, in the minimization problem over a probabilistic model, the \emph{natural gradient} method arises naturally as a first-order optimization method. However, when extending to second-order methods, it is necessary to compute the Hessian of the objective function. Unlike Riemannian optimization on manifolds such as the Stiefel or Grassmann manifold, where the Hessian is typically defined using the Levi--Civita connection, statistical manifolds admit dual affine connections that induce distinct Hessian operators. This motivates the study of Newton-type optimization on statistical manifolds, focusing on the role of dual affine connections in the formulation of Newton updates and their local convergence properties.

 In this paper, we further generalize this line of research by introducing a Newton-type optimization framework on Riemannian manifolds endowed with \emph{arbitrary dual affine connections}. Unlike the Levi--Civita or the exponential and mixture connections
in specific exponential-family models, the proposed method, called the \emph{dual Riemannian Newton method}, is applicable to minimization problems involving \emph{general probability distributions}. Specifically, we consider the quadruple
\((\mathcal{M},\, g,\, \nabla,\, \nabla^{*})\),
where \( \mathcal{M} \) is the parameter manifold, \( g \) is the chosen Riemannian metric (for example, the Fisher information metric for statistical models), and \( \nabla \) and \( \nabla^* \) are an arbitrary pair of dual affine connections (for example, $\alpha$-connections). This formulation provides a unified geometric framework that allows Newton-type methods to be applied beyond the exponential family, enabling efficient second-order optimization directly on statistical manifolds.

We summarize our contribution in the following:

\begin{enumerate}
\item We formulate the dual Riemannian Newton method
on statistical manifolds equipped with
\((\mathcal{M},\, g,\, \nabla,\, \nabla^{*})\).
By deriving a second-order Taylor expansion
along a $\nabla$-based second-order retraction,
we obtain a Newton equation involving the
$\nabla^*$-Hessian and establish its equivalence
to the affine-connection Newton equation
associated with $\nabla$ in Section~\ref{dual_newtown}.

\item We establish the following theoretical properties:
(i) On dually flat manifolds, the proposed method
reduces to the Euclidean Newton method in affine
coordinates (Theorem~\ref{euc_hession});
(ii) For potential-based objectives defined in
\eqref{eq:potential-based-objective}, the proposed method coincides
with natural gradient descent in primal affine
coordinates and, through Legendre duality, with
mirror descent in dual affine coordinates
(Theorem~\ref{natural_equal_newton});
(iii) The method achieves local quadratic convergence
under suitable regularity conditions
(Theorem~\ref{convergence}), and local cubic convergence for specific
potential-based objectives under the
$\alpha$-connections with $\alpha=1/3$
and $\alpha=-1/3$
(Corollary~\ref{cor:alpha_one_third}).

\item We numerically investigate the convergence
behavior of the proposed method on representative
statistical models, including empirical convergence
orders and comparisons of iteration counts and
computational costs with natural gradient,
mirror descent, and Adam~\cite{Kingma2014}
in Section~\ref{chap6}.
\end{enumerate}

\section{Preliminaries}
\label{chap2}
We introduce the notations and provide a summary of the mathematical background relevant to this paper.
We begin by introducing the dual affine connections, along with the Riemannian gradient and retraction.  
An affine connection is formally defined in Definition~\ref{def:affine-connection} in the Appendix. 

\begin{definition}[Dual affine Connection and dualistic Structure]\label{dual_affine}
Let \( (\mathcal{M}, g) \) be a Riemannian manifold, 
\( \langle \cdot, \cdot \rangle \) the Riemannian inner product induced by a Riemannian metric \( g \), 
and \( \nabla \) an affine connection on \( \mathcal{M} \).
The \emph{dual affine connection} \( \nabla^* \) with respect to \( g \) is the unique affine connection that satisfies
\begin{equation}
X \langle Y, Z \rangle
= \langle \nabla_X Y,\; Z \rangle
+ \langle Y,\; \nabla^*_X Z \rangle
\label{eq:dual-connection}
\end{equation}
for all vector fields \( X, Y, Z \in \mathfrak{X}(\mathcal{M}) \).
If \( \nabla \) and \( \nabla^* \) satisfy the above duality condition \eqref{eq:dual-connection}, then 
\((g, \nabla, \nabla^*)\) is called a \emph{dualistic structure} (or dual affine structure) on \( \mathcal{M} \).
Moreover, the dual connection determined by this condition satisfies
\begin{equation}
(\nabla^*)^* = \nabla.
\label{eq:dual-of-dual}
\end{equation}
\end{definition}

Throughout this paper, we assume that both \( \nabla \) and \( \nabla^* \) are \emph{torsion-free}, an assumption that is usually satisfied in information geometry \cite{AmariNagaoka2000}: that is,
\begin{equation}
\nabla_X Y - \nabla_Y X = [X, Y],
\qquad
\nabla^*_X Y - \nabla^*_Y X = [X, Y],
\end{equation}
for any vector fields \( X, Y \in \mathfrak{X}(\mathcal{M}) \).

A concrete example illustrating the dual affine connection on a probabilistic model is $\alpha$-connection presented in the following.
Other examples include $q$-divergence \cite{amariohara2011} and the logarithmic divergence \cite{wong2018}.
\begin{example}[$\alpha$-geometry on a statistical manifold]
\label{alpha_geometry}
Let \( p(x; \xi) \) be the \emph{probability density function (PDF)} of a random variable \( x \), parameterized by \( \xi =(\xi_{1},\ldots,\xi_{n})^{\top}\). We denote the log-likelihood function by \( \ell(\xi) = \ln p(x; \xi) \). The Riemannian metric on the parameter space, induced by the Fisher information, is given by
\[
g_{ij}(\xi) = \mathbb{E}\bigl[\partial_i \ell(\xi)\,\partial_j \ell(\xi)\bigr],\quad 1\le i,j\le n,
\]
where \( \partial_i = \frac{\partial}{\partial \xi_i} \).

Within the framework of $\alpha$-geometry, the inner product of vector fields satisfies the following duality relation:
\[
X\langle Y, Z\rangle
= \langle \nabla^{(\alpha)}_X Y, Z \rangle 
+ \langle Y, \nabla^{(-\alpha)}_X Z \rangle,
\]
where \( \nabla^{(\alpha)} \) and \( \nabla^{(-\alpha)} \) form a dual pair of affine connections.
The Christoffel symbols of the $\alpha$-connection in local coordinates are given by~\cite{AmariNagaoka2000}
\[
\langle \nabla^{(\alpha)}_{\partial_i} \partial_j, \partial_k \rangle 
= \Gamma^{(\alpha)}_{ij,k}(\xi)
:= \mathbb{E}\Bigl[\bigl( \partial_i \partial_j \ell 
+ \tfrac{1-\alpha}{2}\,\partial_i \ell\, \partial_j \ell \bigr)\, \partial_k \ell \Bigr].
\]
These coefficients are symmetric in the lower indices,
\(
\Gamma^{(\alpha)}_{ij,k}(\xi) = \Gamma^{(\alpha)}_{ji,k}(\xi),
\)
which implies that the $\alpha$-connection is \emph{torsion-free}.
Moreover, raising an index using \( g^{ks}(\xi) \) yields
\[
\Gamma^{(\alpha)s}_{ij}(\xi) 
= \sum_{k=1}^{n} \Gamma^{(\alpha)}_{ij,k}(\xi)\, g^{ks}(\xi).
\]

The pair \( (\nabla^{(\alpha)}, \nabla^{(-\alpha)}) \) is \emph{dual} with respect to the Riemannian metric \( g \). Notably, when \( \alpha = 0 \), the $\alpha$-connection coincides with the \emph{Levi-Civita connection}, the unique torsion-free and metric-compatible connection in Riemannian geometry. In information geometry, the special cases \( \alpha = 1 \) and \( \alpha = -1 \) correspond to the \emph{e-connection} and \emph{m-connection}, respectively, which play central roles in the geometry of exponential and mixture families.
\end{example}


Subsequently, we introduce the definition of the Riemannian gradient.

\begin{definition}[Riemannian Gradient]\label{def:riemannian-gradient}
Let $(\mathcal{M}, g)$ be a Riemannian manifold and let $f: \mathcal{M} \to \mathbb{R}$ be a smooth real-valued function.  
For any point $p \in \mathcal{M}$, the \emph{Riemannian gradient} of $f$ at $p$, denoted by $\mathrm{grad}\,f(p)$, is the unique vector in the tangent space $T_p\mathcal{M}$ at $p$ that satisfies
\[
Df(p)[X_p] = \left\langle \mathrm{grad}\,f(p),\; X_p \right\rangle
\quad \text{for all } X_p \in T_p\mathcal{M},
\]
where \( Df(p) \) is the differential of \( f \) at \( p \).
\end{definition}

Here, the differential \( Df(p) \) of a smooth function \( f: \mathcal{M} \to \mathbb{R} \) at a point \( p \in \mathcal{M} \) is defined as
\[
Df(p)[X_p] = X_p(f)
\]
for all \( X_p \in T_p\mathcal{M} \), where \( X_p(f) \) denotes the directional derivative of \( f \) at \( p \) in the direction \( X_p \).

We next define the notion of a retraction on a manifold.
\begin{definition}[Retraction on a manifold]\label{def:retraction}
Let \( \mathcal{M} \) be a smooth manifold.  
A smooth mapping \( R: T\mathcal{M} \to \mathcal{M} \), where \( T\mathcal{M}:=\bigsqcup_{p\in\mathcal{M}}T_{p}\mathcal{M} \) denotes the tangent bundle of \( \mathcal{M} \), is called a \emph{retraction} if for any \( p \in \mathcal{M} \), the restriction  
\[
R_p := R|_{T_p\mathcal{M}} : T_p\mathcal{M} \to \mathcal{M}
\]
satisfies the following two conditions:
\begin{enumerate}
    \item \( R_p(0_p) = p \) with the zero vector \( 0_p \in T_p\mathcal{M} \) at \( p \), and 
    \item the differential of \( R_p \) at \( 0_p \), denoted by \( DR_p(0_p) \), satisfies
\[
DR_p(0_p)[X_p] = X_p \quad \text{for all } X_p \in T_p\mathcal{M},
\]
that is, \( DR_p(0_p) = \mathrm{id}_{T_p\mathcal{M}} \).
\end{enumerate}
\end{definition}
For the definition of the differential of a smooth map between manifolds, please refer to Definition~\ref{def:manifold-map-differential} in the Appendix.

\section{Dual Riemannian Newton method}
\label{dual_newtown}
In this section,
we introduce our proposed method, called the \emph{dual Riemannian Newton method}, which is a second-order optimization technique defined on a manifold equipped with dual affine connections. 

\subsection{Dual Riemannian Hessian}
We begin by introducing an important ingredient of the proposed method,
 the \emph{dual Riemannian Hessian} defined on a dualistic structure $(g, \nabla, \nabla^*)$ as well as its symmetry. 

\begin{definition}[Dual Riemannian Hessian]\label{def:hessian}
Let $f: \mathcal{M} \to \mathbb{R}$ be a smooth function on a Riemannian manifold $\mathcal{M}$, and let $\nabla$ and $\nabla^*$ be a pair of dual affine connections with respect to the metric $g$.
The \emph{Hessian} of $f$ with respect to $\nabla$ at a point $p \in \mathcal{M}$, denoted by $\mathrm{Hess}\,f(p)$, is a linear map $\mathrm{Hess}\,f(p): T_p\mathcal{M} \to T_p\mathcal{M}$ given as
\[
\mathrm{Hess}\, f(p)[X_p] = \nabla_{X_p} \mathrm{grad}\, f
\]
for any $X_p \in T_p\mathcal{M}$, where $\mathrm{grad}\, f$ is the Riemannian gradient of $f$ with respect to the metric $g$.
Similarly, the \emph{dual Hessian} $\mathrm{Hess}^* f(p): T_p\mathcal{M} \to T_p\mathcal{M}$ of $f$ with respect to $\nabla^*$ at the point $p$ is given as
\[
\mathrm{Hess}^* f(p)[X_p] = \nabla^*_{X_p} \mathrm{grad}\, f.
\]
\end{definition}

\begin{proposition}[Symmetry of the Hessian]\label{prop:hessian-symmetry}
Let \( f : \mathcal{M} \to \mathbb{R} \) be a smooth  function defined on a Riemannian manifold 
\(\bigl( \mathcal{M}, g, \nabla, \nabla^* \bigr)\), 
where \( \nabla \) and \( \nabla^* \) are dual affine connections. Then both linear maps $\mathrm{Hess}\,f(p)$ and $\mathrm{Hess}^* f(p)$ are symmetric at any point $p \in \mathcal{M}$; that is, for any $X_p, Y_p \in T_p \mathcal{M}$, the following equations hold.
\begin{align*}
\langle \mathrm{Hess}\,f(p)[X_p],\; Y_p \rangle = \langle X_p,\; \mathrm{Hess}\,f(p)[Y_p] \rangle, \quad
\langle \mathrm{Hess}^* f(p)[X_p],\; Y_p \rangle = \langle X_p,\; \mathrm{Hess}^* f(p)[Y_p] \rangle.
\end{align*}
\end{proposition}

\begin{proof}
For smooth vector fields \( X, Y \in \mathfrak{X}(\mathcal{M}) \),
since both \( \nabla \) and \( \nabla^* \) are \emph{torsion-free}; that is,
\begin{equation}
\nabla_X Y - \nabla_Y X = [X, Y],
\qquad
\nabla^*_X Y - \nabla^*_Y X = [X, Y],
\end{equation}
these vector fields satisfy $X(p) = X_p, Y(p) = Y_p$
for all $p \in \mathcal{M}$.
According to Definition~\ref {dual_affine} in the preliminaries, we obtain
\begin{align*}
&\left\langle \nabla_X \operatorname{grad} f,\; Y \right\rangle 
- \left\langle \nabla_Y \operatorname{grad} f,\; X \right\rangle \\
=\; &X \left\langle \operatorname{grad} f,\; Y \right\rangle 
- \left\langle \operatorname{grad} f,\; \nabla^*_X Y \right\rangle 
- Y \left\langle \operatorname{grad} f,\; X \right\rangle 
+ \left\langle \operatorname{grad} f,\; \nabla^*_Y X \right\rangle \\
=\; &X(Yf) - Y(Xf) 
- \left\langle \operatorname{grad} f,\; \nabla^*_X Y - \nabla^*_Y X \right\rangle \\
=\; &[X, Y]f - [X, Y]f = 0.
\end{align*}
Similarly, for the dual connection $\nabla^*$ we have
\[
\left\langle \nabla^*_X \operatorname{grad} f,\; Y \right\rangle 
- \left\langle \nabla^*_Y \operatorname{grad} f,\; X \right\rangle = 0,
\]
which completes the proof.
\end{proof}

\subsection{Dual Riemannian Newton method}
As mentioned in the introduction, parameter estimation in probabilistic modeling is often formulated as an optimization problem.  
In particular, the parameter space can be modeled as a Riemannian manifold, and the optimization problem takes the form:
\begin{equation}
    \min_{p \in \mathcal{M}} f(p),
    \label{optimization_problem}
\end{equation}
where \( f: \mathcal{M} \to \mathbb{R} \) is a smooth function defined on a smooth Riemannian manifold \( \mathcal{M} \).  
Typical examples include the negative log-likelihood function in maximum likelihood estimation (MLE) or divergence-based loss functions such as the Kullback--Leibler divergence, when estimating parameters of a probabilistic model. In the following, we introduce the dual Riemannian Newton method for solving this optimization problem.

First, we recall the following lemma concerning the covariant derivative; please refer to Definition~\ref{def:covariant-derivative} in the Appendix for its definition.
\begin{lemma}[\cite{AmariNagaoka2000}]\label{prop:covariant-inner}
Consider a manifold $\mathcal{M}$ with $(g, \nabla, \nabla^*)$.
For any smooth curve $\gamma(t)$ and vector fields $Y(t), Z(t)$ along the curve, the following identity holds:
\[
\frac{d}{dt}\left\langle Y,\; Z \right\rangle 
= \left\langle \frac{D}{dt} Y,\; Z \right\rangle 
+ \left\langle Y,\; \frac{D^{*}}{dt} Z \right\rangle.
\]
Here, $\frac{D}{dt} Y$ (resp.\ $\frac{D^*}{dt} Y$) denotes the covariant derivative of $Y(t)$ induced by the affine connection $\nabla$ (resp.\ $\nabla^*$) along the curve $\gamma(t)$.
\end{lemma}
Based on the above lemma, we now derive our proposed dual Riemannian Newton method on a manifold \( (\mathcal{M}, g, \nabla, \nabla^*) \).
We recall a second-order retraction~\cite{Boumal2023} associated with the affine connection \( \nabla \), defined as $R_p(t X_p),$
where \( p \in \mathcal{M} \) and \( X_p \in T_p \mathcal{M} \).  
This retraction satisfies the standard conditions:
\[R_p(0_p) = p \quad\text{and}\quad \mathrm{D}R_p(0_p)[X_p] = X_p.\]
By denoting \( R_p(t X_p) \) := \( \gamma_{p,X_p}(t) \), 
these conditions are written as
\begin{equation}
\gamma_{p,X_p}(0)=p \quad\text{and}\quad \dot{\gamma}_{p,X_p}(t)\Big|_{t=0} =X_p .
\label{conditionforretraction1}
\end{equation}    
Since \( \gamma_{p,X_p}(t) \) is a second-order retraction, it also satisfies the condition:
\begin{equation}
\nabla_{\dot{\gamma}_{p,X_p}(t)} \dot{\gamma}_{p,X_p}(t)\Big|_{t=0} = 0.
\label{conditionforretraction2}
\end{equation}

As the first step, we derive the second-order approximation of $f(\gamma_p(t X_p))$.
By applying Taylor's expansion to approximate the function $f(\gamma_p(t X_p))$, we have
\begin{align}
& f(\gamma_p(t X_p)) \nonumber\\ 
&\approx f(\gamma_p(0_p))
+ t \left. \frac{d}{dt} f(\gamma_p(t X_p)) \right|_{t=0}
+ \frac{1}{2} t^2 \left. \frac{d^2}{dt^2} f(\gamma_p(t X_p)) \right|_{t=0} \notag \\
&= f(p)
+ t \left. \langle \operatorname{grad} f(\gamma_{p,X_p}(t)), \dot{\gamma}_{p,X_p}(t) \rangle \right|_{t=0}
+ \frac{1}{2} t^2 \left. \frac{d}{dt} \langle \dot{\gamma}_{p,X_p}(t), \operatorname{grad} f(\gamma_{p,X_p}(t)) \rangle \right|_{t=0} \notag ,
\end{align}
where the last equality follows from (\ref{conditionforretraction1}).
Together with (\ref{conditionforretraction1}) and (\ref{conditionforretraction2}), we get
\begin{align}
&f(p)
+ t \left. \langle \operatorname{grad} f(\gamma_{p,X_p}(t)), \dot{\gamma}_{p,X_p}(t) \rangle \right|_{t=0}
+ \frac{1}{2} t^2 \left. \frac{d}{dt} \langle \dot{\gamma}_{p,X_p}(t), \operatorname{grad} f(\gamma_{p,X_p}(t)) \rangle \right|_{t=0}\nonumber\\ 
&= f(p)
+ t \langle \operatorname{grad} f(p), X_p \rangle \notag \\
&\quad
+ \frac{1}{2} t^2 \left. \left(
\left\langle \frac{\mathrm{D}}{dt} \dot{\gamma}_{p,X_p}(t), \operatorname{grad} f(\gamma_{p,X_p}(t)) \right\rangle
+ \left\langle \dot{\gamma}_{p,X_p}(t), \frac{\mathrm{D}^*}{dt} \left( \operatorname{grad} f \left(\gamma_{p,X_p}(t)\right) \right) \right\rangle
\right) \right|_{t=0} \notag \\
&= f(p)
+ t \langle \operatorname{grad} f(p), X_p \rangle \notag \\
&\quad
+ \frac{1}{2} t^2 \left. \left(
\left\langle \nabla_{\dot{\gamma}_{p,X_p}(t)} \dot{\gamma}_{p,X_p}(t), \operatorname{grad} f(\gamma_{p,X_p}(t)) \right\rangle
+ \left\langle \dot{\gamma}_{p,X_p}(t), \frac{\mathrm{D}^*}{dt} \left( \operatorname{grad} f\left(\gamma_{p,X_p}(t)\right) \right) \right\rangle
\right) \right|_{t=0} \notag \\
&= f(p)
+ t \langle \operatorname{grad} f(p), X_p \rangle
+ \frac{1}{2} t^2 \langle \mathrm{Hess}^* f(p)[X_p], X_p \rangle.
\label{eq:taylor-expansion}
\end{align}
Therefore, the function \( f(\gamma_p(X_p)) \) can be locally approximated as
\begin{equation}
f(\gamma_p(X_p))
\approx f(p)
+ \left\langle \operatorname{grad} f(p),\, X_p \right\rangle
+ \frac{1}{2} 
  \left\langle \mathrm{Hess}^* f(p)[X_p],\, X_p \right\rangle.
\label{eq:local_approx}
\end{equation}
Subsequently, we have
\[
f(\gamma_p(X_p + tY_p)) \approx  f(p)
+ \left\langle \operatorname{grad} f(p),\, X_p + tY_p \right\rangle
+ \frac{1}{2}
  \left\langle 
      \mathrm{Hess}^* f(p)[X_p + tY_p],\,
      X_p + tY_p
  \right\rangle.
\]

As the second step, we derive Newton's equation with a second-order retraction related to $\nabla$ from the above approximation.
For simplicity, we define
\begin{equation}
\bar{f}(t)
:= f(p)
+ \left\langle \operatorname{grad} f(p),\, X_p + tY_p \right\rangle
+ \frac{1}{2}
  \left\langle 
      \mathrm{Hess}^* f(p)[X_p + tY_p],\,
      X_p + tY_p
  \right\rangle.
\label{eq:fbar_def}
\end{equation}
By applying the self-adjoint property of the dual Hessian stated in Proposition~\ref{prop:hessian-symmetry}
\[
\left\langle \mathrm{Hess}^* f(p)[X_p],\, Y_p \right\rangle 
= 
\left\langle X_p,\, \mathrm{Hess}^* f(p)[Y_p] \right\rangle,
\]
we obtain
\begin{equation}
\begin{aligned}
\left.\frac{d\bar{f}(t)}{dt}\right|_{t=0}
&= 
\left\langle \operatorname{grad} f(p),\, Y_p \right\rangle
+ \left\langle \mathrm{Hess}^* f(p)[X_p],\, Y_p \right\rangle
+ t \left\langle \mathrm{Hess}^* f(p)[Y_p],\, Y_p \right\rangle \Big|_{t=0} \\[4pt]
&=
\left\langle \operatorname{grad} f(p),\, Y_p \right\rangle
+ \left\langle \mathrm{Hess}^* f(p)[X_p],\, Y_p \right\rangle.
\end{aligned}
\end{equation}
By considering the stationary condition
\[\left.\frac{d\bar{f}(t)}{dt}\right|_{t=0} = 0 \quad\text{for all}\quad Y_p \in T_p\mathcal{M},
\]
we arrive at
\begin{equation}
\mathrm{Hess}^* f(p)[X_p] = -\operatorname{grad} f(p),
\label{duality show}
\end{equation}
which corresponds to Newton's equation with a second-order retraction $ R_p(X_p)=\gamma_{p, X_p}(1)$ related to \( \nabla \).
In other words, Equation~\eqref{duality show} implies that, for each point \( p \in \mathcal{M} \),  
the tangent vector \( X_p \in T_p\mathcal{M} \) yielding the next iterate  
\( p \gets R_p(X_p) \)  
can be obtained by solving this equation.

Equation~\eqref{duality show} motivates the Dual Riemannian Newton method presented in Algorithm~\ref{alg: dual riemannian newton}.
Surprisingly, this illustrates the \emph{duality} of Newton's method formulated on \( (\mathcal{M}, g, \nabla, \nabla^*) \).  

Specifically, if we choose the second-order retraction \( \gamma_{p, X_p}(t) \) satisfying 
\(\nabla_{\dot{\gamma}_{p,X_p}(t)} \dot{\gamma}_{p,X_p}(t)\big|_{t=0} = 0,\) the corresponding quadratic model is expressedin terms of the dual Hessian associated with
$\nabla^*$. Consequently, the Newton step should be computed using its dual connection \( \nabla^* \) through Equation~\eqref{duality show}. 

\begin{algorithm}
\caption{Dual Riemannian Newton Method}
\label{alg: dual riemannian newton}
\begin{algorithmic}[1]
\Require \( (\mathcal{M}, g, \nabla, \nabla^*) \), initial point $p \in \mathcal{M}$ and retraction $R_{p}(\cdot)$.
\Ensure Approximate local minimizer of $f : \mathcal{M} \to \mathbb{R}$


\Repeat
    \State Compute the Newton direction $X_{p} \in T_{p}\mathcal{M}$ by solving:
    \[
    \mathrm{Hess^*} f(p)[X_{p}] = -\operatorname{grad} f(p)
    \]
    \State Update: $p \gets R_{p}(X_{p})$
\Until{convergence}

\end{algorithmic}
\end{algorithm}

Absil et al.~\cite[Proposition~5.5.2]{Absil2008}
established that the Riemannian Hessian associated
with the Levi-Civita connection satisfies
\[
\left\langle
\mathrm{Hess}^{\mathrm{LC}}f(p)[X_p],Y_p
\right\rangle
=
(\nabla^{\mathrm{LC}}\mathrm{D}f)_p(X_p,Y_p).
\]
In contrast, our formulation considers a general pair
of dual affine connections $(\nabla,\nabla^*)$,
which need not coincide with the Levi-Civita connection.
In this setting, the dual Hessian associated with
$\nabla^*$ is related to the covariant Hessian
associated with $\nabla$ through the following
proposition.

\begin{proposition}[Relationship between dual Hessian
and covariant Hessian]
\label{prop:dual-covariant-hessian}
Let $(\mathcal{M},g,\nabla,\nabla^*)$ be a Riemannian
manifold equipped with a pair of dual affine connections,
and let $f:\mathcal{M}\to\mathbb{R}$ be a smooth function.

Then, for any $p\in\mathcal{M}$ and
$X_p,Y_p\in T_p\mathcal{M}$, we have
\begin{equation}
\left\langle
\mathrm{Hess}^*f(p)[X_p],Y_p
\right\rangle
=
(\nabla\mathrm{D}f)_p(X_p,Y_p),
\label{eq:dual-covariant-identity}
\end{equation}
where the covariant Hessian associated with $\nabla$
is defined by
\begin{equation}
(\nabla\mathrm{D}f)_p(X_p,Y_p)
=
X\bigl[\mathrm{D}f(Y)\bigr]\big|_p
-
\mathrm{D}f(p)\bigl[(\nabla_XY)_p\bigr],
\label{eq:covariant-hessian-definition}
\end{equation}
with $X$ and $Y$ being smooth vector fields extending
$X_p$ and $Y_p$, respectively.
\end{proposition}

\begin{proof}
Let $X$ and $Y$ be smooth vector fields extending
$X_p$ and $Y_p$, respectively.

By the definition of the Riemannian gradient
given in the preliminaries and the duality relation
in Definition~\ref{dual_affine}, we obtain
\begin{align}
X\bigl[\mathrm{D}f(Y)\bigr]
&=
X\left\langle Y,\operatorname{grad}f\right\rangle
\nonumber\\
&=
\left\langle
\nabla_XY,\operatorname{grad}f
\right\rangle
+
\left\langle
Y,\nabla_X^*\operatorname{grad}f
\right\rangle
\nonumber\\
&=
\mathrm{D}f(\nabla_XY)
+
\left\langle
\nabla_X^*\operatorname{grad}f,Y
\right\rangle.
\end{align}

Rearranging the above identity yields
\begin{equation}
(\nabla\mathrm{D}f)(X,Y)
= X\bigl[\mathrm{D}f(Y)\bigr]
-\mathrm{D}f(\nabla_XY)
= \left\langle\nabla_X^*\operatorname{grad}f,Y\right\rangle.
\end{equation}

By Definition~\ref{def:hessian}, we have
$
\nabla_X^*\operatorname{grad}f
=
\mathrm{Hess}^*f[X].
$

Evaluating the identity at $p$, we obtain
\[
\left\langle
\mathrm{Hess}^*f(p)[X_p],Y_p
\right\rangle
=
(\nabla\mathrm{D}f)_p(X_p,Y_p),
\]
which completes the proof.
\end{proof}

 Based on this identity,
we further show that the corresponding Newton
equations are equivalent.

\begin{proposition}[Equivalence of Newton equations]
\label{prop:newton-equivalence}
Under the assumptions of
Proposition~\ref{prop:dual-covariant-hessian},
the dual Riemannian Newton equation
\begin{equation}
\mathrm{Hess}^*f(p)[X_p]
=
-\operatorname{grad}f(p)
\label{eq:dual-newton-equivalence}
\end{equation}
is equivalent to the affine-connection Newton equation
\begin{equation}
(\nabla\mathrm{D}f)_p(X_p,Y_p)
=
-\mathrm{D}f(p)[Y_p],
\qquad
\forall Y_p\in T_p\mathcal{M}.
\label{eq:affine-newton-equivalence}
\end{equation}
\end{proposition}

\begin{proof}
Suppose that $X_p\in T_p\mathcal{M}$ satisfies
Equation~\eqref{eq:dual-newton-equivalence}.

Taking the Riemannian inner product with an arbitrary
$Y_p\in T_p\mathcal{M}$ gives
\begin{equation}
\left\langle
\mathrm{Hess}^*f(p)[X_p],Y_p
\right\rangle
=
-\left\langle
\operatorname{grad}f(p),Y_p
\right\rangle=-\mathrm{D}f(p)[Y_p].
\end{equation}

By Proposition~\ref{prop:dual-covariant-hessian}, we obtain
\begin{equation}
(\nabla\mathrm{D}f)_p(X_p,Y_p)
=
-\mathrm{D}f(p)[Y_p],
\qquad
\forall Y_p\in T_p\mathcal{M}.
\end{equation}

Conversely, suppose that
Equation~\eqref{eq:affine-newton-equivalence}
holds for all $Y_p\in T_p\mathcal{M}$.

Using Proposition~\ref{prop:dual-covariant-hessian}
again, we have
\begin{align}
0
&=
(\nabla\mathrm{D}f)_p(X_p,Y_p)
+\mathrm{D}f(p)[Y_p]
\nonumber\\
&=
\left\langle
\mathrm{Hess}^*f(p)[X_p],Y_p
\right\rangle
+
\left\langle
\operatorname{grad}f(p),Y_p
\right\rangle
\nonumber\\
&=
\left\langle
\mathrm{Hess}^*f(p)[X_p]
+\operatorname{grad}f(p),Y_p
\right\rangle.
\end{align}

Since this identity holds for every
$Y_p\in T_p\mathcal{M}$ and the Riemannian
metric $g$ is nondegenerate, it follows that
$
\mathrm{Hess}^*f(p)[X_p]
+\operatorname{grad}f(p)=0.
$

Therefore,
\[
\mathrm{Hess}^*f(p)[X_p]
=
-\operatorname{grad}f(p),
\]
which completes the proof.
\end{proof}

The above propositions show that the dual Riemannian
Newton equation is a tangent-vector representation
of the covariant-Hessian Newton equation associated
with the affine connection $\nabla$. In particular, although Algorithm~\ref{alg: dual riemannian newton}
is formulated using the dual Hessian associated with
$\nabla^*$, the Newton direction can be computed
directly from the covariant Hessian associated with
$\nabla$, without explicitly computing the dual
connection $\nabla^*$. Thus, the implementation of Algorithm~\ref{alg: dual riemannian newton}
requires only the affine connection $\nabla$,
together with the derivatives of $f$ and the
corresponding retraction $R_p$.
We discuss the coordinate representation and
computational details in
Section~\ref{subsec: Computation of Dual Riemannian Newton method}.

\(
\)

For the \emph{Riemannian gradient descent method} defined on \( (\mathcal{M}, g, \nabla, \nabla^*) \), it is not necessary to solve the Newton equation~\eqref{duality show}. Instead, Step~2 in Algorithm~\ref{alg: dual riemannian newton} is replaced with
\[
X_{p} = -s\, \operatorname{grad} f(p),
\]
where \( s \) is a learning rate.
Moreover, the \emph{natural gradient method} can be viewed as a special case of the Riemannian gradient descent method when the Riemannian metric is chosen as the Fisher information matrix, while the \emph{mirror descent method}~\cite{RaskuttiMukherjee2015} can be interpreted as the natural gradient method expressed in affine coordinates on a dually flat manifold.
We will elaborate on this relationship in Section \ref{subsec: relationship with projection problem}.

\subsection{Computation of Dual Riemannian Newton method}
\label{subsec: Computation of Dual Riemannian Newton method}
We provide an explicit computation of the dual Riemannian Newton method under a local coordinate chart:
\[
\begin{tikzcd}
\mathcal{M} \arrow[r, "f"] \arrow[d, "\varphi"'] & \mathbb{R} \\
\xi=\left( \xi_1, \dots, \xi_n \right) \arrow[ur, "f \circ \varphi^{-1}"'] &
\end{tikzcd}
\]
Here, 
\( \mathcal{M} \) is a smooth manifold of dimension \( n \). Moreover, \( \varphi: U \subset \mathcal{M} \to \mathbb{R}^n \) is a coordinate chart, and for any point \( p \in U \), we denote the local coordinates by 
$\xi^p = (\xi_1^p, \dots, \xi_n^p)^\top := \varphi(p).$
Given a smooth function \( f : \mathcal{M} \to \mathbb{R} \), its coordinate expression under \( \varphi \) is 
$f \circ \varphi^{-1} : \varphi(U) \subset \mathbb{R}^n \to \mathbb{R}.$

Firstly, we describe the Newton equation in~\eqref{duality show} using the coordinate coefficients of a tangent vector.
Observe that the Riemannian gradient $\operatorname{grad} f(p)$ in a local coordinate chart $\varphi$ is expressed as:  
\[
\operatorname{grad} f(p) = \sum_{i,j=1}^n g^{ij}(\xi^p) \, \frac{\partial \left(f \circ \varphi^{-1}\right)}{\partial \xi_j}\biggr|_{\xi^p}  \frac{\partial}{\partial \xi_i} \biggr|_{\xi^p},
\]
where \( g^{ij}(\xi^p) \) is the $(i,j)$-component of the inverse Riemannian metric \( g^{-1} \) at the point \( p \in \mathcal{M} \) under the local coordinates $\xi$, in detail, 
\[
[g^{ij}(\xi^p)] = [g_{ij}(\xi^p)]^{-1}
\quad \text{with} \quad
g_{ij}(\xi^p) := \left\langle \frac{\partial}{\partial \xi_i}, \frac{\partial}{\partial \xi_j} \right\rangle\biggr|_{\xi^p}.
\]
Here, \( \left\{ \frac{\partial}{\partial \xi_1}\big|_{\xi^p}, \dots, \frac{\partial}{\partial \xi_n}\big|_{\xi^p} \right\} \) is the coordinate basis of the tangent space \( T_p \mathcal{M} \).  
For simplicity, we write
\[
\operatorname{grad} f(p) = \sum_{k=1}^{n} a_k(\xi^p) \, \frac{\partial}{\partial \xi_k}\biggr|_{\xi^p}\quad \text{with} \quad a_k(\xi^p) := \sum_{j=1}^{n} g^{kj}(\xi^p) \, \frac{\partial \left(f \circ \varphi^{-1}\right)}{\partial \xi_j}\biggr|_{\xi^p}.
\]
Take an arbitrary tangent vector
$X_p=\sum_{k=1}^{n}\beta_k(\xi^p)\frac{\partial}{\partial \xi_k}\bigg|_{\xi^p}
$
under local coordinate $\xi$.
Consequently, $\mathrm{Hess}^{*} f(p)[X_p]$ can be calculated as:
\begin{equation}
\begin{aligned}
&\mathrm{Hess}^{*} f(p)[X_p]\\
=\ &\nabla^{*}_{X_p} \operatorname{grad} f \\
=\ & 
\begin{pmatrix}
\beta_1 & \cdots & \beta_n
\end{pmatrix}
\left[
\begin{pmatrix}
\frac{\partial a_1}{\partial \xi_1} & \cdots & \frac{\partial a_n}{\partial \xi_1} \\
\vdots & \ddots & \vdots \\
\frac{\partial a_1}{\partial \xi_n} & \cdots & \frac{\partial a_n}{\partial \xi_n}
\end{pmatrix}
+
\begin{pmatrix}
\sum_{j=1}^n a_j \Gamma^{*}_{1 j}{}^1 & \cdots & \sum_{j=1}^n a_j \Gamma^{*}_{1 j}{}^n \\
\vdots & \ddots & \vdots \\
\sum_{j=1}^n a_j \Gamma^{*}_{n j}{}^1 & \cdots & \sum_{j=1}^n a_j \Gamma^{*}_{n j}{}^n
\end{pmatrix}
\right]
\left.
\begin{pmatrix}
\frac{\partial}{\partial \xi_1} \\
\vdots \\
\frac{\partial}{\partial \xi_n}
\end{pmatrix}
\right\rvert_{\xi^p},
\end{aligned}
\end{equation}
where \( \Gamma^{*}_{ij}{}^k \) denotes the Christoffel symbol of the dual affine connection \( \nabla^* \) in the coordinate chart \(\xi= (\xi_1, \dots, \xi_n) \), defined by
$
\nabla^*_{\frac{\partial}{\partial \xi_i}} \left( \frac{\partial}{\partial \xi_j} \right)\biggr|_{\xi^p}
= \sum_{k=1}^n \Gamma^{*}_{ij}{}^k(\xi^p) \frac{\partial}{\partial \xi_k}\biggr|_{\xi^p}.
$
For convenience, we define the matrix
$\mathbf{H}^{\nabla^{*}}(\xi)\in\mathbb{R}^{n\times n}$
whose $(i,j)$-element satisfies
\begin{equation}
\mathbf{H}^{\nabla^{*}}_{ij}(\xi^p)
=
\frac{\partial a_i(\xi)}{\partial\xi_j}
\biggr|_{\xi^p}
+\sum_{k=1}^n
a_k(\xi^p)\Gamma^{*}_{jk}{}^i(\xi^p).
\label{hession_matrix}
\end{equation}
Thus, the Newton equation
in \eqref{duality show} can be written as
\begin{equation}
\mathbf{H}^{\nabla^{*}}(\xi^p)
\boldsymbol{\beta}(\xi^p)
=
-\mathbf{G}^{-1}(\xi^p)
\boldsymbol{\nabla}f(\xi^p).
\label{eq:newton-eq-coord-compact}
\end{equation}
where
$\boldsymbol{\beta}(\xi^p)
=(\beta_1(\xi^p),\ldots,\beta_n(\xi^p))^\top$
is the Newton direction in local coordinates,
$
\boldsymbol{\nabla}f(\xi^p)
=
\left(
\frac{\partial(f\circ\varphi^{-1})}
{\partial\xi_1},
\ldots,
\frac{\partial(f\circ\varphi^{-1})}
{\partial\xi_n}
\right)^\top
\biggr|_{\xi^p},
$
and
$\mathbf{G}(\xi^p)=[g_{ij}(\xi^p)]$
denotes the Riemannian metric matrix,
with inverse
$\mathbf{G}^{-1}(\xi^p)=[g^{ij}(\xi^p)]$.

Alternatively, by
Proposition~\ref{prop:newton-equivalence},
the Newton equation in Equation~\eqref{duality show} 
can be computed directly from the $\nabla$-affine connection Newton equation
\begin{equation}
(\nabla\mathrm{D}f)_p(X_p,Y_p)
=
-\mathrm{D}f(p)[Y_p],
\qquad
\forall Y_p\in T_p\mathcal{M}.
\label{eq:affine-newton-local}
\end{equation}

In the local coordinate chart $\xi$, the
covariant Hessian associated with $\nabla$
is given by
\begin{equation}
\begin{aligned}
&(\nabla\mathrm{D}f)_p
\left(
\left.\frac{\partial}{\partial\xi_i}\right|_{\xi^p},
\left.\frac{\partial}{\partial\xi_j}\right|_{\xi^p}
\right)
=
\frac{\partial^2(f\circ\varphi^{-1})}
{\partial\xi_i\partial\xi_j}
\biggr|_{\xi^p}
-
\sum_{k=1}^n
\Gamma_{ij}{}^k(\xi^p)
\frac{\partial(f\circ\varphi^{-1})}
{\partial\xi_k}
\biggr|_{\xi^p},
\end{aligned}
\label{eq:covariant-hessian-coordinate}
\end{equation}
where $\Gamma_{ij}{}^k$ denotes the Christoffel symbol
of the affine connection $\nabla$ in the local
coordinate chart $\xi=(\xi_1,\ldots,\xi_n)^\top$,
defined by
$
\nabla_{\frac{\partial}{\partial \xi_i}}
\left(\frac{\partial}{\partial \xi_j}\right)
\Big|_{\xi^p}
=
\sum_{k=1}^{n}
\Gamma_{ij}{}^k(\xi^p)
\frac{\partial}{\partial \xi_k}
\Big|_{\xi^p}.
$

Recall that the Newton direction is expressed as
$X_p
=
\sum_{j=1}^n
\beta_j(\xi^p)
\left.\frac{\partial}{\partial\xi_j}\right|_{\xi^p}.$
Since Equation~\eqref{eq:affine-newton-local}
holds for every $Y_p\in T_p\mathcal{M}$,
we may choose
$
Y_p=
\left.\frac{\partial}{\partial\xi_i}\right|_{\xi^p}, i=1,\ldots,n.
$

Substituting the coordinate expressions into
Equation~\eqref{eq:affine-newton-local},
and using the torsion-free property of $\nabla$,
we obtain
\begin{equation}
\begin{aligned}
\sum_{j=1}^n
\Bigg[
&\frac{\partial^2(f\circ\varphi^{-1})}
{\partial\xi_i\partial\xi_j}
\biggr|_{\xi^p}-
\sum_{k=1}^n
\Gamma_{ij}{}^k(\xi^p)
\frac{\partial(f\circ\varphi^{-1})}
{\partial\xi_k}
\biggr|_{\xi^p}
\Bigg]
\beta_j(\xi^p)=
-\frac{\partial(f\circ\varphi^{-1})}
{\partial\xi_i}
\biggr|_{\xi^p},
\qquad i=1,\ldots,n.
\end{aligned}
\label{eq:affine-newton-coordinate}
\end{equation}

Therefore, the Newton direction in
Algorithm~\ref{alg: dual riemannian newton}
can be obtained by solving
Equation~\eqref{eq:affine-newton-coordinate},
which involves only the Christoffel symbols
of $\nabla$ and the first- and second-order
partial derivatives of $f$. In particular, the dual affine connection $\nabla^*$ does not need to be computed explicitly.

For convenience, we define the matrix
$\mathbf{M}^{\nabla}(\xi^p)\in\mathbb{R}^{n\times n}$
whose $(i,j)$-element is given by
\begin{equation}
\begin{aligned}
\mathbf{M}^{\nabla}_{ij}(\xi^p)
={}&
\frac{\partial^2(f\circ\varphi^{-1})}
{\partial\xi_i\partial\xi_j}
\biggr|_{\xi^p}-
\sum_{k=1}^n
\Gamma_{ij}{}^k(\xi^p)
\frac{\partial(f\circ\varphi^{-1})}
{\partial\xi_k}
\biggr|_{\xi^p}.
\end{aligned}
\label{eq:affine-newton-matrix-definition}
\end{equation}
By Proposition~\ref{prop:dual-covariant-hessian}, the matrix
representations of the covariant Hessian and the dual Hessian
are related by
$
    \mathbf{M}^{\nabla}(\xi^p)
    =
    \mathbf{G}(\xi^p)\mathbf{H}^{\nabla^{*}}(\xi^p).
$
Since $\nabla$ is torsion-free,
$\mathbf{M}^{\nabla}(\xi^p)$ is symmetric.
Thus, Equation~\eqref{eq:affine-newton-coordinate}
can be equivalently written in matrix form as
\begin{equation}
\mathbf{M}^{\nabla}(\xi^p)
\boldsymbol{\beta}(\xi^p)
=
-\boldsymbol{\nabla}f(\xi^p),
\label{eq:affine-newton-matrix}
\end{equation}
where $\boldsymbol{\beta}(\xi^p)$ and
$\boldsymbol{\nabla}f(\xi^p)$ are defined above.
Consequently, if $\mathbf{M}^{\nabla}(\xi^p)$
is nonsingular, the Newton direction is given by
\begin{equation}
\boldsymbol{\beta}(\xi^p)
=
-\left(\mathbf{M}^{\nabla}(\xi^p)\right)^{-1}
\boldsymbol{\nabla}f(\xi^p).
\label{eq:affine-newton-direction}
\end{equation}

Therefore, the Newton direction in
Algorithm~\ref{alg: dual riemannian newton}
can be computed directly using the Christoffel
symbols of $\nabla$ and the first- and
second-order partial derivatives of $f$.

Next, we will approximate the update based on a second-order retraction under a local coordinate chart.
Let $\gamma_{p,X_p}(t):=R_{p}\left(t X_p\right)$ be a second-order retraction on the manifold.
From its local chart $\left(\varphi,\xi \right)$, 
we have \(\varphi(\gamma_{p,X_p}(t))=\left(\xi_1(t),\ldots,\xi_n(t) \right)^{\top}.\)
Moreover, we obtain \( \dot{\gamma}_{p,X_p}(0) = X_p \), and
\(
\varphi\!\left(\gamma_{p,X_p}(0)\right)
= (\xi_1(0), \ldots, \xi_n(0))^{\top}
= \varphi(p)
= (\xi_1^p, \ldots, \xi_n^p)^{\top}.
\)
The tangent vector \( X_p \) can be expressed as
\(
X_p = \sum_{k=1}^{n} \beta_k(\xi^p)\,\frac{\partial}{\partial \xi_k}\Big|_{\xi^p}.
\)
Furthermore, the curve $\bar{\gamma}_{p,X_p}$ is defined as the geodesic~\cite{lee2018introduction} emanating from $p$ in the direction $X_p$:
\[
\bar{\gamma}_{p,X_p}(0)=p,\qquad \dot{\bar{\gamma}}_{p,X_p}(0)=X_p,
\qquad \text{and}\qquad
\nabla_{\dot{\bar{\gamma}}_{p,X_p}(t)} \dot{\bar{\gamma}}_{p,X_p}(t)=0, \;\; \forall t \in [0,\epsilon].
\]
Let $\varphi\!\left(\bar{\gamma}_{p,X_p}(t)\right)=\big(\bar{\xi}_1(t),\ldots,\bar{\xi}_n(t)\big)^{\top}$. 
The local coordinate form of the geodesic equation is 
\[
\ddot{\bar{\xi}}_{\,i}(t)
+\sum_{j,k=1}^{n}\Gamma^{i}_{jk}\!\big(\bar{\xi}(t)\big)\,
\dot{\bar{\xi}}_{j}(t)\,\dot{\bar{\xi}}_{k}(t)=0,
\qquad
\bar{\xi}_{i}(0)=\xi_i^p,\quad
\dot{\bar{\xi}}_{i}(0)=\beta_i(\xi^p);
\]
therefore at $t=0$, we have
$
\ddot{\bar{\xi}}_{i}(0)
= -\sum_{j,k=1}^{n}\Gamma^{i}_{jk}\!\left(\xi^p\right)\,
\beta_j(\xi^p)\,\beta_k(\xi^p).
$
Since a second-order retraction agrees with the geodesic up to second order (in other words, we require $\nabla_{\dot{\gamma}_{p,X_p}(t)} \dot{\gamma}_{p,X_p}(t)\big|_{t=0}=0$), we can take a specific second order retraction,
a local-coordinate
second-order retraction $\varphi\!\left(\gamma_{p,X_p}(t)\right)$, of the form
\begin{align}
\xi_i(t)
= \xi^p_i
+ t\,\beta_i(\xi^p)
- \frac{t^{2}}{2}\sum_{j,k=1}^{n}
\Gamma^{i}_{jk}\!\left(\xi^p\right)\,
\beta_j(\xi^p)\,\beta_k(\xi^p),
\qquad i=1,\ldots,n,
\label{approximate retraction}
\end{align}


From (\ref{eq:affine-newton-matrix}) and (\ref{approximate retraction}),
we arrive at the explicit computational procedure of Algorithm~\ref{alg: dual riemannian newton} described in Algorithm~\ref{alg:Computation dual riemannian newton}.

\begin{remark}
The full-step dual Riemannian Newton method
is locally well-defined near a nondegenerate
stationary point $p^\star$ 
satisfying
$\operatorname{grad}f(p^\star)=0$ and
$(\nabla Df)_{p^\star}\succ0$.
Indeed, since
$\mathbf{M}^{\nabla}(\xi^{p^\star})
=\boldsymbol{\nabla}^2f(\xi^{p^\star})\succ0$,
the symmetry and continuity of
$\mathbf{M}^{\nabla}$ imply that it remains
positive definite in a sufficiently small
neighborhood of $p^\star$.
Consequently, the Newton system is uniquely
solvable, and the Newton direction satisfies
$
\begin{aligned}
\mathrm{D}f(p)[X_p]
&=-\boldsymbol{\nabla}f(\xi^p)^\top
[\mathbf{M}^{\nabla}(\xi^p)]^{-1}
\boldsymbol{\nabla}f(\xi^p)<0
\end{aligned}
$
whenever $\operatorname{grad}f(p)\neq0$.
However, outside this neighborhood,
$\mathbf{M}^{\nabla}$ may become singular
or indefinite.
In such cases, the Newton direction may
not be well-defined or may fail to be
a descent direction. Globalization strategies such as regularization or trust-region methods may be required. In this paper, the convergence guarantees established in
Theorem~\ref{convergence} concern the
local full-step iteration.
\end{remark}

At each iteration, Algorithm~\ref{alg:Computation dual riemannian newton}
requires evaluating the gradient and Hessian of $f$,
constructing the $n\times n$ Newton coefficient matrix
$\mathbf M^\nabla$, and computing the second-order
retraction correction. Solving the Newton system
$\mathbf M^\nabla(\xi^p)\boldsymbol{\beta}(\xi^p)
=-\boldsymbol{\nabla}f(\xi^p)$
requires $O(n^3)$ operations using a dense matrix
factorization.

The geometric computation can be accelerated by
avoiding the explicit evaluation of all connection
coefficients. In particular, the algorithm only
requires the sums
\begin{equation}
\begin{aligned}
&\sum_{k=1}^n
\Gamma_{ij}{}^k(\xi^p)
\frac{\partial(f\circ\varphi^{-1})}
{\partial\xi_k}\biggr|_{\xi^p},
\qquad
\sum_{j,k=1}^n
\Gamma_{jk}{}^i(\xi^p)
\beta_j(\xi^p)\beta_k(\xi^p).
\end{aligned}
\label{eq:connection_contractions}
\end{equation}
which can be computed directly using automatic
differentiation and Hessian-vector products
when the connection admits an appropriate
computational representation, following the
connection-vector product approach of
Song et al.~\cite{song2018accelerating}.

For example, on a Fisher--Amari geometry described in Example~\ref{alpha_geometry},
the metric coefficients $g_{ij}$ and connection
coefficients $\Gamma_{ij}{}^k$ can be computed
from expectations with respect to the underlying
probability distribution. When these expectations
are approximated numerically using $Q$ integration
points, evaluating all connection coefficients
requires $O(Qn^3)$ operations. 
By directly evaluating the sums in
Equation~\eqref{eq:connection_contractions}
during numerical integration, this cost can
be reduced to $O(Qn^2)$.
 If the objective derivatives
can also be evaluated within $O(Qn^2)$ operations,
the overall per-iteration complexity is
$O(Qn^2+n^3)$, with $O(n^2)$ auxiliary memory
when the computations are performed without
storing all connection coefficients.


  

\begin{algorithm}
\caption{Coordinate-based Dual Riemannian Newton Method}
\label{alg:Computation dual riemannian newton}
\begin{algorithmic}[1]
\Require \( (\mathcal{M}, g, \nabla, \nabla^*) \), local chart \((\varphi, \xi)\), initial coordinates \( \xi^{p} = \varphi(p) \in \mathbb{R}^n \)
\Ensure Approximate local minimizer of \( f : \mathcal{M} \to \mathbb{R} \)

\Repeat
    \State Compute the Newton direction \( \boldsymbol{\beta}(\xi^{p}) = (\beta_1(\xi^{p}), \dots, \beta_n(\xi^{p}))^{\top} \) by solving:
    \[
       \mathbf{M}^{\nabla}(\xi^p)\,
        \boldsymbol{\beta}(\xi^p)
        =-\boldsymbol{\nabla}f(\xi^p).
\]
    \State Update coordinates using the second-order retraction approximation:
    \[
    \xi^p_i \gets \xi^p_i + \beta_i(\xi^p) - \frac{1}{2} \sum_{j,k=1}^n \Gamma^i_{jk}(\xi^p) \, \beta_j(\xi^p) \, \beta_k(\xi^p), \quad i = 1, \ldots, n
    \]
\Until{convergence}
\end{algorithmic}

\end{algorithm}

For comparison, we present the corresponding computation of the natural gradient descent method under a local coordinate chart.
The natural gradient descent method corresponds to a Riemannian gradient descent using a first-order retraction, which is defined by
$
\xi_i(t) = \xi^p_i + t\, \beta_i(\xi^p),
$
where \( \beta_i(\xi^p) \) denotes the \( i \)-th component of the
$-s\, \mathbf{G}^{-1}(\xi^p)\, \boldsymbol{\nabla} f(\xi^p)$ at the point \( \xi^p \).
We write the computation procedure of the natural gradient descent method in Algorithm~\ref{alg:Computation Natural gradient descent Method} for comparison with the dual Riemannian Newton method.
\begin{algorithm}
\caption{Coordinate-based Natural Gradient Descent Method}
\label{alg:coordinate-natural-gradient-descent}
\begin{algorithmic}[1]
\Require Riemannian manifold \( (\mathcal{M}, g, \nabla, \nabla^*) \), local chart \( (\varphi, \xi) \), initial point \( p \in \mathcal{M} \) with coordinates \( \xi^p = \varphi(p) \in \mathbb{R}^n \); step size \( s > 0 \)
\Ensure Approximate local minimizer of \( f : \mathcal{M} \to \mathbb{R} \)

\Repeat
    \State Compute the natural gradient direction in coordinates:
    \[
    \boldsymbol{\beta}(\xi^p) = -s\, \mathbf{G}^{-1}(\xi^p)\, \boldsymbol{\nabla} f(\xi^p)
    \]
    \State Update local coordinates:
    \[
    \xi^p_i \gets \xi^p_i + \beta_i(\xi^p), \quad i = 1, \dots, n
    \]
\Until{convergence}
\end{algorithmic}
\label{alg:Computation Natural gradient descent Method}
\end{algorithm}
Algorithm~\ref{alg:Computation Natural gradient descent Method} further illustrates that, contrast to the dual Riemannian Newton Method, the
natural gradient descent method does not require computing the Christoffel symbols
\( \Gamma_{ij}{}^{k} \) and \( \Gamma^{*}_{ij}{}^{k} \) associated with the affine connections
\( \nabla \) and \( \nabla^{*} \).

\section{Theoretical properties}
\label{chap4}
Although the introduction of the dual Riemannian Newton method does not rely on any specific structure of \( (\mathcal{M}, g, \nabla, \nabla^*) \),
for a dually flat manifold in the context of information geometry the dual Riemannian Newton method enjoys additional theoretical properties in specific coordinates, which will be introduced in this section.

\subsection{Dual Riemannian Newton Method on dually flat manifold}
First, we prove that each step of the dual Riemannian Newton method reduces to the Euclidean Newton step when applied to a dually flat manifold under affine coordinates.

Let \( \nabla \) be an affine connection on a  manifold \( \mathcal{M} \), and let \( [\cdot,\cdot] \) denote the Lie bracket of vector fields defined by
\(
[X,Y] := XY - YX
\)
for all \( X, Y \in \mathfrak{X}(\mathcal{M}) \).
The \emph{torsion tensor} \( T \) is the \( \mathcal{C}^\infty(\mathcal{M}) \)-bilinear map defined by
\begin{equation}
T(X, Y) = \nabla_X Y - \nabla_Y X - [X, Y],
\label{eq:torsion}
\end{equation}
and the \emph{Riemannian curvature tensor} \( R \) is the \( \mathcal{C}^\infty(\mathcal{M}) \)-trilinear map defined by
\begin{equation}
R(X, Y)Z = \nabla_X \nabla_Y Z - \nabla_Y \nabla_X Z - \nabla_{[X, Y]} Z,
\label{eq:curvature}
\end{equation}
for all vector fields \( X, Y, Z \in \mathfrak{X}(\mathcal{M}) \).
Using such tensors, we can introduce the dually flat manifold as follows:
\begin{definition}[Dually Flat Manifold]\label{def:dually-flat}
Let \( (\mathcal{M}, g, \nabla, \nabla^*) \) be a Riemannian manifold endowed with dual affine connections \( \nabla \) and \( \nabla^* \).  
We say that \( \mathcal{M} \) is a \emph{dually flat manifold} if both the torsion and curvature tensors of \( \nabla \) and \( \nabla^* \) vanish:
\begin{equation}
T = T^* = 0,
\qquad
R = R^* = 0.
\label{eq:dually-flat}
\end{equation}
\end{definition}
In the dually flat manifold, the following two affine coordinates naturally exist~\cite{AmariNagaoka2000}:
Let \( \mathcal{M} \) be a dually flat manifold equipped with $(g, \nabla, \nabla^*)$. A local chart $\left(\varphi,\xi \right)$ of \( U \subset\mathcal{M} \) along which all the connection coefficients \(\{ \Gamma_{ij}^k \}\) vanish identically is called a \(\nabla\)\emph{-affine coordinate neighborhood}, and the associated local coordinate system \( (\xi_i) \) is called a \(\nabla\)\emph{-affine coordinate system}. Similarly, we can also define
\(\nabla^*\)\emph{-affine coordinate system}.
For example, a probability distribution belongs to the exponential family if its density can be written as
$
p(x; \theta) = h(x) \exp\left( \sum_{i=1}^n \theta_i F_i(x) - \psi(\theta) \right),
$
where \( \theta = (\theta_1, \dots, \theta_n)^{\top} \in \Theta \subset \mathbb{R}^n \) is the natural parameter, \( F(x) = (F_1(x), \dots, F_n(x))^{\top} \) is the vector of sufficient statistics, \( h(x) \) is the base measure, and \( \psi(\theta) \) is the log-partition function that normalizes the distribution:
\begin{equation}
\psi(\theta) = \log \int h(x) \exp\left( \sum_{i=1}^n \theta_i F_i(x) \right) dx.
\label{psi}
\end{equation}
The expectation parameter is defined by the gradient \( \eta=(\eta_1,\dots,\eta_n)^{\top} := \nabla_{\theta} \psi(\theta) \). In the framework of information geometry, \( \theta \) defines an $\mathrm{e}$-affine coordinate system under the exponential connection \( \nabla^{(\mathrm{e})} \), while \( \eta \) defines an $\mathrm{m}$-affine coordinate system under the mixture connection \( \nabla^{(\mathrm{m})} \)~\cite{AmariNagaoka2000}. The manifold \( (\mathcal{M}, g, \nabla^{(\mathrm{e})}, \nabla^{(m)}) \) is dually flat. Moreover, the explicit computation of the exponential and mixture connections \( \nabla^{(\mathrm{e})} \) and \( \nabla^{(\mathrm{m})} \) on a statistical manifold is provided in Example~\ref{alpha_geometry}.

In the subsequent discussion, for notational convenience, we define \(\theta \) as the affine coordinate system with respect to the connection \( \nabla \), and \( \eta \) as that with respect to the dual connection \( \nabla^* \) on a dually flat manifold \( (\mathcal{M}, g, \nabla, \nabla^*) \). 
Recall the $\nabla$-affine-connection Newton equation
in local coordinates $\xi$,
as given in Equation~\eqref{eq:affine-newton-matrix}:
\begin{equation}
\mathbf{M}^{\nabla}(\xi^p)
\boldsymbol{\beta}(\xi^p)
=
-\boldsymbol{\nabla}f(\xi^p).
\label{eq:newton-eq-coord-compact-gleft}
\end{equation}
In an affine coordinate system associated with $\nabla$,
the matrix $\mathbf{M}^{\nabla}$ admits the following
simplified form.

\begin{theorem}
\label{euc_hession}
Let $(\mathcal{M},g,\nabla,\nabla^*)$ be a dually flat
manifold, and let $\theta=(\theta_1,\ldots,\theta_n)^\top$
be an affine coordinate system associated with $\nabla$.
Then, for any $p$ in the coordinate domain,
\begin{equation}
\mathbf{M}^{\nabla}_{ij}(\theta^p)
=
\bigl[
    \mathbf{G}(\theta^p)
    \mathbf{H}^{\nabla^{*}}(\theta^p)
\bigr]_{ij}
=
\frac{\partial^2(f\circ\varphi_\theta^{-1})}
{\partial\theta_i\partial\theta_j}
\biggr|_{\theta^p}.
\end{equation}
\end{theorem}

\begin{proof}
Since $\theta$ is an affine coordinate system
associated with $\nabla$, we have
\[
\Gamma_{ij}{}^k(\theta)=0,
\qquad i,j,k=1,\ldots,n.
\]
By the definition of $\mathbf{M}^{\nabla}$
in Equation~\eqref{eq:affine-newton-matrix-definition},
\[
\begin{aligned}
\mathbf{M}^{\nabla}_{ij}(\theta^p)
&=
\frac{\partial^2(f\circ\varphi_\theta^{-1})}
{\partial\theta_i\partial\theta_j}
\biggr|_{\theta^p}
-
\sum_{k=1}^n
\Gamma_{ij}{}^k(\theta^p)
\frac{\partial(f\circ\varphi_\theta^{-1})}
{\partial\theta_k}
\biggr|_{\theta^p}\\
&=
\frac{\partial^2(f\circ\varphi_\theta^{-1})}
{\partial\theta_i\partial\theta_j}
\biggr|_{\theta^p},
\end{aligned}
\]
which completes the proof.
\end{proof}

From this simplified form,
we reduce the dual Riemannian Newton method 
 in affine coordinates to the classical Newton method in Euclidean space as follows:
Since \( \theta = (\theta_1,\ldots,\theta_n)^{\top} \) are affine coordinates associated with the affine connection \( \nabla \), we have \( \Gamma_{ij}{}^{k}(\theta) = 0 \).  
Consequently, \textbf{Line 2} of Algorithm~\ref{alg:Computation dual riemannian newton} simplifies to
\[
\boldsymbol{\beta}(\theta^{p}) = -\, [\boldsymbol{\nabla}^{2}f(\theta^{p})]^{-1} \, \boldsymbol{\nabla} f(\theta^{p}),
\]
where 
$
(\boldsymbol{\nabla}^{2}f(\theta^{p}))_{ij} = \frac{\partial^2 (f \circ \varphi_\theta^{-1})}{\partial \theta_i \, \partial \theta_j} \Big|_{\theta^{p}}.
$
Furthermore, \textbf{Line 3} becomes
\[
\theta_i^{p} \gets \theta_i^{p} + \beta_i(\theta^{p}), \quad i = 1, \dots, n,
\]
which implies that the dual Riemannian Newton method in $\theta$-coordinates matches the classical Newton method in Euclidean space. 
Similarly, let $R_p^*$ be a second-order retraction
associated with $\nabla^*$, satisfying
$
\left.
\nabla^*_{\dot{\gamma}^*_{p,X_p}(t)}
\dot{\gamma}^*_{p,X_p}(t)
\right|_{t=0}=0,
$
where $\gamma^*_{p,X_p}(t):=R_p^*(tX_p)$.
Since $\eta$ is a $\nabla^*$-affine coordinate
system, we have $\Gamma^{*k}_{ij}(\eta)=0$.
 Consequently,
\[
\mathbf{M}^{\nabla^*}_{ij}(\eta^p)
=
\frac{\partial^2(f\circ\varphi_\eta^{-1})}
{\partial\eta_i\partial\eta_j}
\biggr|_{\eta^p}.
\]
By interchanging the roles of $\nabla$ and $\nabla^*$
in Algorithm~\ref{alg:Computation dual riemannian newton},
\textbf{Line 2} simplifies to
\[
\boldsymbol{\beta}(\eta^p)
=
-[\boldsymbol{\nabla}^2 f(\eta^p)]^{-1}
\boldsymbol{\nabla}f(\eta^p),
\]
and \textbf{Line 3} becomes
$
\eta_i^p\gets\eta_i^p+\beta_i(\eta^p),
\quad i=1,\ldots,n.
$
Thus, the corresponding dual Riemannian Newton method
also reduces to the classical Euclidean Newton method
in the $\eta$-coordinates.

\subsection{Relationship with Natural Gradient and Mirror Descent for Potential-Based Objectives}
\label{subsec: relationship with projection problem}
We investigate a class of objective functions for which the dual Riemannian Newton method coincides with natural gradient descent with unit step size in the primal affine coordinates $\theta$, where the Riemannian metric is given by $G(\theta)=\nabla^2\psi(\theta)$. Through the Legendre duality between $\psi$ and $\varphi$, we further establish the equivalence between natural gradient descent in $\theta$ and mirror descent in the dual affine coordinates $\eta$, using the Bregman divergence induced by $\varphi$.
The results apply to maximum likelihood estimation
in exponential families and divergence projection
problems on auto-parallel submanifolds.
We first recall the dual potential functions
on a dually flat manifold.

\begin{lemma}[\cite{AmariNagaoka2000}]
\label{dual_potiential}
Let $(\mathcal M,g,\nabla,\nabla^*)$ be a
dually flat manifold, and let $\theta$ and $\eta$
be affine coordinate systems associated with
$\nabla$ and $\nabla^*$, respectively.
Then there exist smooth, strictly convex
potential functions $\psi(\theta)$ and
$\varphi(\eta)$ satisfying
\begin{equation}
\begin{aligned}
\eta &= \boldsymbol{\nabla}_\theta\psi(\theta),
\qquad
\theta = \boldsymbol{\nabla}_\eta\varphi(\eta),\\
\mathbf G(\theta)
&=\boldsymbol{\nabla}_\theta^2\psi(\theta),
\qquad
\mathbf G(\eta)
=\boldsymbol{\nabla}_\eta^2\varphi(\eta).
\end{aligned}
\label{eq:legendre-coord-metric}
\end{equation}
Moreover, the potentials satisfy
$
\psi(\theta)+\varphi(\eta)=\sum_{i=1}^{n}\theta_i\eta_i.
$
\end{lemma}
The corresponding $\nabla$-divergence is
defined as follows.
\begin{definition}
For $p,q\in\mathcal M$, the
$\nabla$-divergence is
\begin{align*}
D(p\|q)
&:=\psi(\theta^p)+\varphi(\eta^q)
-\sum_{i=1}^{n}\theta_i^p\eta_i^q\\
&=\psi(\theta^p)-\psi(\theta^q)
-\left\langle
\boldsymbol{\nabla}_\theta\psi(\theta^q),
\theta^p-\theta^q
\right\rangle_E.
\end{align*}
\end{definition}
where $\left\{\left(\theta_i\right),\left(\eta_i\right)\right\}$ is dual affine coordinates on $\mathcal{M}$, and $\langle u, v\rangle_E$ denotes the standard Euclidean inner product, i.e., $\langle u, v\rangle_E=\sum_{i=1}^n u_i v_i$.

For example, consider an exponential family, $\psi(\theta)$ is the
log-partition function defined in \eqref{psi},
and its dual potential $\varphi(\eta)$ is
given by the Legendre transform~\cite{amari2016information}:
\begin{equation}
\varphi(\eta)
:=\sup_{\theta\in\Theta}
\left\{
\sum_{i=1}^n\theta_i\eta_i-\psi(\theta)
\right\}.
\end{equation}

The corresponding $\nabla$-divergence satisfies
$D(p\|q)=D_{\mathrm{KL}}(q\|p)$,
where $D_{\mathrm{KL}}$ denotes the
Kullback--Leibler (KL) divergence.

When the projection problem is restricted to an \emph{auto-parallel} submanifold, the resulting projection enjoys desirable properties such as \emph{uniqueness} and \emph{orthogonality} \cite{amari2016information}, 
which ensures a well-defined geometric correspondence between the point and the submanifold. In the subsequent analysis, we restrict the projection problem to auto-parallel submanifolds.

A formal definition of the auto-parallel submanifold is provided below:
Let \( \mathcal{S} \subset \mathcal{M} \) be a submanifold equipped with a local coordinate system \( (u_a)_{1 \leq a \leq m} \).  
For any vector fields \( \frac{\partial}{\partial u_a}, \frac{\partial}{\partial u_b} \in \mathcal{X}(\mathcal{S}) \),  
where \( \mathcal{X}(\mathcal{S}) \) denotes the set of all smooth vector fields on \( \mathcal{S} \),  
and for any point \( p \in \mathcal{S} \), the covariant derivative satisfies
$
 \left. {\nabla}_{\frac{\partial}{\partial u_a}} \frac{\partial}{\partial u_b} \right|_{u^p} \in T_{p} \mathcal{S},
$
then \( \mathcal{S} \) is called an \emph{auto-parallel submanifold} of \( \mathcal{M} \) with respect to the connection \({\nabla} \) defined on \( \mathcal{M} \). Intuitively, the affine connection-induced covariant derivative of any tangent vector on the submanifold remains in its tangent space.

In the context of information geometry, an auto-parallel submanifold with respect to the \( \mathrm{e} \)-connection is called an \emph{\( \mathrm{e} \)-flat submanifold} (typical for exponential families), whereas one with respect to the \( \mathrm{m} \)-connection is called an \emph{\( \mathrm{m} \)-flat submanifold} (typical for mixture families). 

We now provide a lemma of auto-parallel submanifolds on dually flat manifolds.

\begin{lemma}[\cite{AmariNagaoka2000}]
\label{autoparallel}
Let \( (\mathcal{M}, g, \nabla, \nabla^*) \) be a dually flat manifold, and let \( (\theta_i)_{1 \leq i \leq n} \) be an affine coordinate system associated with the connection \( \nabla \).  
A necessary and sufficient condition for an \( m \)-dimensional submanifold \( \mathcal{S} \subset \mathcal{M} \) to be auto-parallel with respect to \( \nabla \) is the existence of a constant matrix \( \mathbf{A} \in \mathbb{R}^{n \times m} \) of rank \( m \), a constant vector \( \mathbf{b} \in \mathbb{R}^n \), and a local coordinate system \( (u_a)_{1 \leq a \leq m} \) on \( \mathcal{S} \) that is affine with respect to \( \nabla \), such that the embedding of \( \mathcal{S} \) into \( \mathcal{M} \) is locally given by
\[
\begin{bmatrix}
\theta_1 \\
\theta_2 \\
\vdots \\
\theta_n
\end{bmatrix}
=
\mathbf{A}
\begin{bmatrix}
u_1 \\
u_2 \\
\vdots \\
u_m
\end{bmatrix}
+ \mathbf{b}.
\]
\end{lemma}

We now establish the equivalence for
a general class of potential-based
objective functions.

\begin{theorem}
\label{natural_equal_newton}
Let $(\mathcal M,g,\nabla,\nabla^*)$ be a dually flat
manifold with dual affine coordinates $(\theta,\eta)$
and corresponding potentials $(\psi,\varphi)$.
Suppose that
\begin{equation}
f(\theta)=\psi(\theta)-\mathbf c^\top\theta+C_0,
\label{eq:potential-based-objective}
\end{equation}
where $\mathbf c\in\mathbb R^n$ and $C_0\in\mathbb R$.

Then the $\nabla$-based dual Riemannian Newton method coincides with natural gradient descent with unit step size in the $\theta$-coordinates. Moreover, under the Legendre transformation, both methods are equivalent to mirror descent in the $\eta$-coordinates with unit step size, using the Bregman divergence induced by the Legendre dual potential $\varphi$.

The dual statement holds for
$f(\eta)=\varphi(\eta)-\mathbf c^\top\eta+C_0$
upon interchanging $\nabla$ and $\nabla^*$.
\end{theorem}

\begin{proof}

Since $\theta$ is $\nabla$-affine, Theorem~\ref{euc_hession}
and Lemma~\ref{dual_potiential} yield
$
\mathbf M^\nabla(\theta)
=\boldsymbol{\nabla}_\theta^2 f(\theta)
=\mathbf G(\theta).
$
Consequently, the dual Riemannian Newton iteration
with the $\nabla$-based affine-coordinate retraction is
\begin{equation}
\theta^{p_{k+1}}
=\theta^{p_k}
-\mathbf G(\theta^{p_k})^{-1}
\boldsymbol{\nabla}_\theta f(\theta^{p_k}),
\label{eq:DN-NG-equivalence}
\end{equation}
which is precisely natural gradient descent with
unit step size.

Let $\widetilde f(\eta)=f(\theta(\eta))$.
By Lemma~\ref{dual_potiential}, we have
$
\boldsymbol{\nabla}_\eta\varphi(\eta)=\theta,
\boldsymbol{\nabla}_\eta\widetilde f(\eta)
=\mathbf G(\theta)^{-1}
\boldsymbol{\nabla}_\theta f(\theta).
$

Mirror descent in the $\eta$-coordinates with
potential $\varphi$ and unit step size is
given by~\cite{RaskuttiMukherjee2015}
\begin{equation}
\eta^{p_{k+1}}
=\arg\min_z
\left\{
\left\langle
\boldsymbol{\nabla}_\eta\widetilde f(\eta^{p_k}),z
\right\rangle_E
+D(\eta^{p_k}\|z)
\right\},
\label{eq:dual-mirror-descent}
\end{equation}

where the $\nabla$-divergence in the
$\eta$-coordinates is given by
\[
D(\eta^{p_k}\|z)
=\varphi(z)-\varphi(\eta^{p_k})
-\left\langle
\boldsymbol{\nabla}_\eta\varphi(\eta^{p_k}),
z-\eta^{p_k}
\right\rangle_E.
\]

Taking the derivative with respect to $z$ in
\eqref{eq:dual-mirror-descent}, we obtain
\[
\boldsymbol{\nabla}_z
\left(
\left\langle
\boldsymbol{\nabla}_\eta\widetilde f(\eta^{p_k}),z
\right\rangle_E
+D(\eta^{p_k}\|z)
\right)
=
\boldsymbol{\nabla}_\eta\widetilde f(\eta^{p_k})
+\boldsymbol{\nabla}_\eta\varphi(z)
-\boldsymbol{\nabla}_\eta\varphi(\eta^{p_k}).
\]

Setting $z=\eta^{p_{k+1}}$, the first-order optimality
condition of \eqref{eq:dual-mirror-descent} yields
\[
\begin{aligned}
0
&=\boldsymbol{\nabla}_\eta\widetilde f(\eta^{p_k})
+\boldsymbol{\nabla}_\eta\varphi(\eta^{p_{k+1}})
-\boldsymbol{\nabla}_\eta\varphi(\eta^{p_k}),\\
\boldsymbol{\nabla}_\eta\varphi(\eta^{p_{k+1}})
&=\boldsymbol{\nabla}_\eta\varphi(\eta^{p_k})
-\boldsymbol{\nabla}_\eta\widetilde f(\eta^{p_k})\\
&=\theta^{p_k}
-\mathbf G(\theta^{p_k})^{-1}
\boldsymbol{\nabla}_\theta f(\theta^{p_k})
=\theta^{p_{k+1}},
\end{aligned}
\]
where the second equality follows from
\eqref{eq:legendre-coord-metric} and the chain rule,
and the last equality follows from
\eqref{eq:DN-NG-equivalence}. The dual statement follows analogously.

\end{proof}

Here, the Riemannian metric
$\mathbf G(\theta)=\boldsymbol{\nabla}_\theta^2\psi(\theta)$
is induced by $\psi$, whereas mirror descent in the
$\eta$-coordinates uses the dual potential $\varphi$.
The following examples illustrate the applicability
of Theorem~\ref{natural_equal_newton}.

\noindent\textbf{Example 1: Maximum likelihood estimation
in exponential families.}

For an exponential family
$p(x;\theta)=h(x)\exp(\sum_{i=1}^n
\theta_iF_i(x)-\psi(\theta))$, the normalized
negative log-likelihood of observations
$x_1,\ldots,x_N$ is
$
f(\theta)=\psi(\theta)-\mathbf c^\top\theta+C_0,
$
where
$c_i=N^{-1}\sum_{a=1}^N F_i(x_a)$ and
$C_0=-N^{-1}\sum_{a=1}^N\log h(x_a)$.
Since $\mathbf G(\theta)=\boldsymbol{\nabla}_\theta^2
\psi(\theta)$, Theorem~\ref{natural_equal_newton}
applies directly.

\noindent\textbf{Example 2: Divergence projection
onto an auto-parallel submanifold.}

Let $\mathcal S\subset\mathcal M$ be
$\nabla$-auto-parallel and
$q\in\mathcal M\setminus\mathcal S$.
Consider $\nabla^*$-projection
$p^\star=\arg\min_{p\in\mathcal S}D(p\|q)$.
By Lemma~\ref{autoparallel}, the
$\nabla$-affine coordinates $u$ on $\mathcal S$
satisfy $\theta(u)=\mathbf A u+\mathbf b$.
The projection objective is
$
f(u)=D(p(u)\|q)
=\psi_{\mathcal S}(u)-\mathbf c^\top u+C_0,
$
where
\[
\psi_{\mathcal S}(u)=\psi(\mathbf A u+\mathbf b),\qquad
\mathbf c=\mathbf A^\top\eta^q,\qquad
C_0=\varphi(\eta^q)-\mathbf b^\top\eta^q.
\]
Moreover, the induced metric satisfies
$
\mathbf G(u)
=\mathbf A^\top\mathbf G(\mathbf A u+\mathbf b)\mathbf A
=\boldsymbol{\nabla}_u^2\psi_{\mathcal S}(u).
$
Thus, Theorem~\ref{natural_equal_newton} applies
in the $u$-coordinates, with mirror descent
expressed in the induced dual coordinates
$\boldsymbol{\nabla}_u\psi_{\mathcal S}(u)$.
The dual result applies to the
$\nabla$-projection
$q^\star=\arg\min_{q\in\mathcal S}D(p\|q)$
on a $\nabla^*$-auto-parallel submanifold.

The equivalence generally fails under regularization.
For $\bar f(u)=f(u)+\lambda\|u\|_2^2$ with $\lambda>0$,
the Newton coefficient matrix becomes
\begin{equation}
\mathbf M^\nabla(u)
=\boldsymbol{\nabla}_u^2\bar f(u)
=\mathbf G(u)+2\lambda\mathbf I.
\label{add_2I}
\end{equation}
Hence,
$
\boldsymbol{\beta}_{\mathrm{DN}}(u)
=-(\mathbf G(u)+2\lambda\mathbf I)^{-1}
\boldsymbol{\nabla}_u\bar f(u), \quad
\boldsymbol{\beta}_{\mathrm{NG}}(u)
=-\mathbf G(u)^{-1}
\boldsymbol{\nabla}_u\bar f(u).
$
Thus, dual Riemannian Newton generally differs from
natural gradient and its dual-coordinate mirror descent.

\section{Convergence analysis}
\label{chap5}
In this section, we investigate the local
convergence properties of the dual Riemannian
Newton method on statistical manifolds.
In contrast to natural gradient descent,
which typically exhibits local linear
convergence~\cite{Boumal2023}, we establish
local quadratic convergence of the proposed
method near a nondegenerate minimizer.
We further investigate the local error behavior
of our method on dually flat statistical
manifolds equipped with $\alpha$-connections.
In particular, for potential-based objectives of
the form \eqref{eq:potential-based-objective} and
their dual counterparts, we derive explicit expressions for the leading
quadratic error terms and establish local cubic
convergence when $\alpha=1/3$ in the
$\theta$-affine coordinates and $\alpha=-1/3$
in the dual $\eta$-affine coordinates.

We begin by briefly introducing the properties of the \emph{consistent matrix norm} \( \| \cdot \| \) on \( \mathbb{R}^{n \times n} \), which satisfies \( \| \mathbf{I} \| = 1 \) and the submultiplicative property \( \|AB\| \leq \|A\|\,\|B\| \).  
A matrix norm \( \|\cdot\| \) is said to be consistent with a vector norm \( \|\cdot\|_v \) if
\[
\|A x\|_v \;\le\; \|A\|\,\|x\|_v
\quad \text{for all } A \in \mathbb{R}^{n \times n},\ x \in \mathbb{R}^n.
\]

For all three matrices $\mathbf{E}$, $\mathbf{A}$, and $\mathbf{B}$, we have the following two properties~\cite{Absil2008}:
\begin{itemize}
  \item[(i)] If \( \| \mathbf{E} \| < 1 \), then \( (\mathbf{I} - \mathbf{E})^{-1} \) exists and it satisfies
  \[
  \| (\mathbf{I} - \mathbf{E})^{-1} \| \leq \frac{1}{1 - \| \mathbf{E} \|}.
  \]

\item[(ii)] If $\mathbf{A}$ is nonsingular and
$\|\mathbf{A}^{-1}(\mathbf{B}-\mathbf{A})\|<1$,
then $\mathbf{B}$ is also nonsingular and satisfies
\begin{equation}
\|\mathbf{B}^{-1}\|
\leq
\frac{\|\mathbf{A}^{-1}\|}
{1-\|\mathbf{A}^{-1}(\mathbf{B}-\mathbf{A})\|}.
\label{eq:matrix-inverse-bound}
\end{equation}

\end{itemize}

Moreover, we also recall the following standard consequence
of local smoothness. Let $F$ be a continuously
differentiable map between finite-dimensional
normed spaces. Since $DF$ is continuous and
locally bounded, there exist $r,L_F>0$ such that

\begin{equation}
\begin{aligned}
\|F(\xi^p)-F(\xi^q)\|
&=\left\|\int_0^1
DF\bigl(\xi^q+t(\xi^p-\xi^q)\bigr)
[\xi^p-\xi^q]\,dt\right\|\\
&\leq\int_0^1
\left\|DF\bigl(\xi^q+t(\xi^p-\xi^q)\bigr)\right\|
\|\xi^p-\xi^q\|\,dt\leq L_F\|\xi^p-\xi^q\|.
\end{aligned}
\label{eq:local_lipschitz}
\end{equation}

for all $\xi^p,\xi^q\in B_r(\xi^{p^\star})$.
Moreover, if $DF$ is locally Lipschitz with
constant $L_{DF}$, the integral form of
Taylor's theorem yields

\begin{equation}
\begin{aligned}
&\|F(\xi^p)-F(\xi^q)
-DF(\xi^q)[\xi^p-\xi^q]\|=\left\|\int_0^1
\bigl[DF(\xi^q+t(\xi^p-\xi^q))
-DF(\xi^q)\bigr]
[\xi^p-\xi^q]\,dt\right\|\\
&\leq\int_0^1
L_{DF}t\|\xi^p-\xi^q\|^2\,dt
=\frac{L_{DF}}{2}\|\xi^p-\xi^q\|^2.
\end{aligned}
\label{eq:local_taylor_bound}
\end{equation}

Using the above properties of the consistent matrix norm, we theoretically prove that the dual Riemannian Newton method achieves local quadratic convergence.

\begin{theorem}
Let $(\varphi,\xi)$ be a coordinate chart on
an $n$-dimensional manifold $\mathcal{M}$.
Assume that $f\circ\varphi^{-1}\in C^2$
with locally Lipschitz continuous second
derivatives, and that the Christoffel symbols
of $\nabla$ are continuously differentiable.
Suppose that there exists $p^\star\in\mathcal{M}$
such that
\[
\operatorname{grad}f(p^\star)=0
\quad\text{and}\quad
(\nabla Df)_{p^\star}\succ0,
\]
where $(\nabla Df)_{p^\star}\succ0$ means
\[
(\nabla Df)_{p^\star}
(X_{p^\star},X_{p^\star})>0,
\quad
\forall X_{p^\star}\in
T_{p^\star}\mathcal{M}\setminus\{0\}.
\]

Assume further that
Algorithm~\ref{alg:Computation dual riemannian newton}
employs the coordinate-based second-order retraction
defined in Equation~\eqref{approximate retraction}.

Then there exists a neighborhood
$\mathcal{U}\subset\mathcal{M}$ of $p^\star$
such that, for any initial point
$p_0\in\mathcal{U}$,
Algorithm~\ref{alg:Computation dual riemannian newton}
is well-defined, remains in $\mathcal{U}$,
and converges locally at least quadratically
to $p^\star$.

\label{convergence}
\end{theorem}

\begin{proof}
Fix an iteration number $k$.
First, we introduce the notation
\[
\xi^{p_k} := \varphi(p_k), 
\qquad 
\varphi(R_{p_k}(X_{p_k})) := \hat{R}_{\xi^{p_k}}(\boldsymbol{\beta}(\xi^{p_k})),
\]
where the tangent vector \( X_{p_k} \) is represented in local coordinates as
$
X_{p_k} 
= \sum_{s=1}^{n}\beta_s(\xi^{p_k})\,\frac{\partial}{\partial \xi_s}\bigg|_{\xi^{p_k}},
$
and
\(
\boldsymbol{\beta}(\xi^{p_k}) 
= (\beta_1(\xi^{p_k}), \dots, \beta_n(\xi^{p_k}))^{\top}.
\)

Next, recall that the Newton equation in local
coordinates is given by Equation~\eqref{eq:affine-newton-matrix}:
$$
\mathbf{M}^{\nabla}(\xi^{p_k})
\boldsymbol{\beta}(\xi^{p_k})
=
-\boldsymbol{\nabla}f(\xi^{p_k}).
$$
where
\[
\begin{aligned}
\mathbf{M}^{\nabla}_{ij}(\xi^{p_k})
={}&
\frac{\partial^2(f\circ\varphi^{-1})}
{\partial\xi_i\partial\xi_j}
\biggr|_{\xi^{p_k}}-
\sum_{s=1}^n\Gamma_{ij}{}^s(\xi^{p_k})
\frac{\partial(f\circ\varphi^{-1})}
{\partial\xi_s}
\biggr|_{\xi^{p_k}}.
\end{aligned}
\]

Since $\boldsymbol{\nabla}f(\xi^{p^\star})=0$
at the optimal point $p^\star$, we have
$
\mathbf{M}^{\nabla}(\xi^{p^\star})
=
\boldsymbol{\nabla}^2f(\xi^{p^\star}).
$

Therefore, the Newton direction is given by
\[
\boldsymbol{\beta}(\xi^{p_k})
=
-[\mathbf{M}^{\nabla}(\xi^{p_k})]^{-1}
\boldsymbol{\nabla}f(\xi^{p_k}),
\]
whenever $\mathbf{M}^{\nabla}(\xi^{p_k})$
is nonsingular.

Accordingly, \textbf{Line 3} of
Algorithm~\ref{alg:Computation dual riemannian newton}
can be written as
\begin{equation}
\xi^{p_{k+1}}
=
\hat{R}_{\xi^{p_k}}
\bigl(\boldsymbol{\beta}(\xi^{p_k})\bigr),
\label{eq:newton_update_M}
\end{equation}
where the second-order coordinate retraction is
given by
\[
\bigl(\hat{R}_{\xi}(v)\bigr)_i
=
\xi_i+v_i
-\frac12\sum_{j,k=1}^n
\Gamma_{jk}^i(\xi)v_jv_k,
\quad i=1,\ldots,n.
\]

The proof proceeds through the following steps:

\noindent\textbf{Step 1: Establish an upper bound
for the variation of the coordinate Hessian.}

Since $f\circ\varphi^{-1}\in C^2$ with locally
Lipschitz continuous second derivatives,
there exist $r_H,L_H>0$ such that
\begin{equation}
\|\boldsymbol{\nabla}^2f(\xi^p)
-\boldsymbol{\nabla}^2f(\xi^q)\|
\leq L_H\|\xi^p-\xi^q\|
\label{eq:coordinate_hessian_lipschitz}
\end{equation}
for all $\xi^p,\xi^q\in B_{r_H}(\xi^{p^\star})$, where 
$B_{r_H}(\xi^{p^\star})
=\{\xi\in\mathbb{R}^n:
\|\xi-\xi^{p^\star}\|<r_H\}$.

Since $\boldsymbol{\nabla}f(\xi^{p^\star})=0$,
applying Equation~\eqref{eq:local_taylor_bound}
to $F=\boldsymbol{\nabla}f$, we obtain
\begin{equation}
\begin{aligned}
&\|\boldsymbol{\nabla}f(\xi^{p_k})
-\boldsymbol{\nabla}^2f(\xi^{p^\star})
(\xi^{p_k}-\xi^{p^\star})\|=\left\|\int_0^1
\left[\boldsymbol{\nabla}^2f
\bigl(\xi^{p^\star}+t(\xi^{p_k}-\xi^{p^\star})\bigr)
-\boldsymbol{\nabla}^2f(\xi^{p^\star})\right]
(\xi^{p_k}-\xi^{p^\star})\,dt\right\|\\
&\leq\int_0^1 L_Ht
\|\xi^{p_k}-\xi^{p^\star}\|^2\,dt=\frac{L_H}{2}
\|\xi^{p_k}-\xi^{p^\star}\|^2.
\end{aligned}
\label{eq:gradient_taylor_bound}
\end{equation}

\noindent\textbf{Step 2: Establish an upper bound
for $\|[\mathbf{M}^{\nabla}(\xi^{p_k})]^{-1}\|$.}

Let
\[
\kappa:=
\|[\mathbf{M}^{\nabla}(\xi^{p^\star})]^{-1}\|.
\]
Since $(\nabla Df)_{p^\star}\succ0$
and $\boldsymbol{\nabla}f(\xi^{p^\star})=0$,
the matrix $\mathbf{M}^{\nabla}(\xi^{p^\star})$
is positive definite and hence nonsingular. 
Furthermore, by the definition of
$\mathbf{M}^{\nabla}$ and the triangle
inequality, for each $i,j=1,\ldots,n$,
we have
\begin{equation}
\begin{aligned}
|\mathbf{M}_{ij}^{\nabla}(\xi^p)
-\mathbf{M}_{ij}^{\nabla}(\xi^q)|
&\leq
\biggl|\left.
\frac{\partial^2(f\circ\varphi^{-1})}
{\partial\xi_i\partial\xi_j}
\right|_{\xi^p}
-\left.
\frac{\partial^2(f\circ\varphi^{-1})}
{\partial\xi_i\partial\xi_j}
\right|_{\xi^q}\biggr|
+\sum_{s=1}^n
|\Gamma_{ij}{}^s(\xi^p)
-\Gamma_{ij}{}^s(\xi^q)|
\left|\left.
\frac{\partial(f\circ\varphi^{-1})}
{\partial\xi_s}
\right|_{\xi^p}\right|\\
&\quad+\sum_{s=1}^n
|\Gamma_{ij}{}^s(\xi^q)|
\biggl|\left.
\frac{\partial(f\circ\varphi^{-1})}
{\partial\xi_s}
\right|_{\xi^p}
-\left.
\frac{\partial(f\circ\varphi^{-1})}
{\partial\xi_s}
\right|_{\xi^q}\biggr|.
\end{aligned}
\end{equation}

By Equation~\eqref{eq:coordinate_hessian_lipschitz},
the local boundedness and Lipschitz continuity
of the Christoffel symbols and the coordinate
gradient, there exist $r_M,C_M>0$ such that
\[
|\mathbf{M}_{ij}^{\nabla}(\xi^p)
-\mathbf{M}_{ij}^{\nabla}(\xi^q)|
\leq C_M\|\xi^p-\xi^q\|
\]
for all $i,j=1,\ldots,n$ and
$\xi^p,\xi^q\in B_{r_M}(\xi^{p^\star})$.

Since all matrix norms are equivalent in
finite dimensions, there exists $c>0$ such that
$\|A\|\leq c\max_{i,j}|A_{ij}|$.
Consequently,
\begin{equation}
\begin{aligned}
&\|\mathbf{M}^{\nabla}(\xi^p)
-\mathbf{M}^{\nabla}(\xi^q)\|\leq c\max_{i,j}
|\mathbf{M}_{ij}^{\nabla}(\xi^p)
-\mathbf{M}_{ij}^{\nabla}(\xi^q)|\leq cC_M\|\xi^p-\xi^q\|
=\gamma_M\|\xi^p-\xi^q\|.
\end{aligned}
\label{eq:gammaM}
\end{equation}

Choose $\epsilon_M>0$ such that
$
\epsilon_M\leq r_M,
\kappa\gamma_M\epsilon_M\leq\frac12.
$
For any
$\xi^{p_k}\in B_{\epsilon_M}(\xi^{p^\star})$,
we obtain
\[
\|[\mathbf{M}^{\nabla}(\xi^{p^\star})]^{-1}
(\mathbf{M}^{\nabla}(\xi^{p_k})-\mathbf{M}^{\nabla}(\xi^{p^\star}))\|
\leq \kappa\gamma_M\|\xi^{p_k}-\xi^{p^\star}\|
<\kappa\gamma_M\epsilon_M\leq\frac12.
\]

By Equation~\eqref{eq:matrix-inverse-bound},
$\mathbf{M}^{\nabla}(\xi^{p_k})$ is nonsingular and
\[
\|[\mathbf{M}^{\nabla}(\xi^{p_k})]^{-1}\|
\leq
\frac{\|[\mathbf{M}^{\nabla}(\xi^{p^\star})]^{-1}\|}
{1-\|[\mathbf{M}^{\nabla}(\xi^{p^\star})]^{-1}
(\mathbf{M}^{\nabla}(\xi^{p_k})-\mathbf{M}^{\nabla}(\xi^{p^\star}))\|}
\leq 2\kappa.
\]

Therefore, for all
$\xi^{p_k}\in B_{\epsilon_M}(\xi^{p^\star})$,
the Newton direction
\[
\boldsymbol{\beta}(\xi^{p_k})
=
-[\mathbf{M}^{\nabla}(\xi^{p_k})]^{-1}
\boldsymbol{\nabla}f(\xi^{p_k})
\]
is well-defined and satisfies
\begin{equation}
\|[\mathbf{M}^{\nabla}(\xi^{p_k})]^{-1}\|
\leq 2\kappa.
\label{eq:inverse_M_bound}
\end{equation}

\noindent\textbf{Step 3: Establish an upper bound
for the retraction error.}

Recall that the second-order coordinate retraction
in Algorithm~\ref{alg:Computation dual riemannian newton}
is given by
\[
(\hat R_\xi(v))_i
=\xi_i+v_i-\frac12\sum_{j,k=1}^n
\Gamma_{jk}^i(\xi)v_jv_k.
\]
Define the retraction remainder by
\[
(\mathbf r_R(\xi))_i
=-\frac12\sum_{j,k=1}^n
\Gamma_{jk}^i(\xi)\beta_j(\xi)\beta_k(\xi).
\]
Since $\boldsymbol{\nabla}f(\xi^{p^\star})=0$
and $\mathbf M^\nabla(\xi^{p^\star})$ is nonsingular,
we have $\boldsymbol\beta(\xi^{p^\star})=0$
and hence $\mathbf r_R(\xi^{p^\star})=0$.

Since $f\circ\varphi^{-1}\in C^2$,
the coordinate gradient is continuously
differentiable. Applying
Equation~\eqref{eq:local_lipschitz}
to $F=\boldsymbol{\nabla}f$ and using
$\boldsymbol{\nabla}f(\xi^{p^\star})=0$,
there exist $r_f,L_f>0$ such that
\begin{equation}
\begin{aligned}
\|\boldsymbol{\nabla}f(\xi)\|
&=\|\boldsymbol{\nabla}f(\xi)
-\boldsymbol{\nabla}f(\xi^{p^\star})\|\leq L_f\|\xi-\xi^{p^\star}\|
\end{aligned}
\label{eq:gradient_local_bound}
\end{equation}
for all $\xi\in B_{r_f}(\xi^{p^\star})$.

Together with the inverse matrix bound in Step 2,
this yields
\[
\begin{aligned}
\|\boldsymbol\beta(\xi^{p_k})\|
&\le
\|[\mathbf M^\nabla(\xi^{p_k})]^{-1}\|
\|\boldsymbol{\nabla}f(\xi^{p_k})\|\le 2\kappa L_f\|\xi^{p_k}-\xi^{p^\star}\|.
\end{aligned}
\]

Moreover, by the continuity of the Christoffel
symbols, there exist $r_R,C_\Gamma>0$ such that
$
\sum_{i,j,k=1}^n|\Gamma_{jk}^i(\xi)|^2
\le C_\Gamma^2
$
for all $\xi\in B_{r_R}(\xi^{p^\star})$.
By the Cauchy--Schwarz inequality,

\[
\begin{aligned}
\|\mathbf r_R(\xi)\|^2
&=\frac14\sum_{i=1}^n
\left|\sum_{j,k=1}^n\Gamma_{jk}^i(\xi)\beta_j(\xi)\beta_k(\xi)\right|^2
\leq\frac14\sum_{i=1}^n
\left(\sum_{j,k=1}^n|\Gamma_{jk}^i(\xi)|^2\right)
\left(\sum_{j,k=1}^n|\beta_j(\xi)\beta_k(\xi)|^2\right)\\
&=\frac14\left(\sum_{i,j,k=1}^n|\Gamma_{jk}^i(\xi)|^2\right)
\|\boldsymbol\beta(\xi)\|^4
\leq\frac{C_\Gamma^2}{4}\|\boldsymbol\beta(\xi)\|^4.
\end{aligned}
\]

Choose $\epsilon_R>0$ sufficiently small such that
$\epsilon_R\le\min\{\epsilon_M,r_f,r_R\}$
and the coordinate updates are well-defined.
Then, for all
 $\xi^{p_k}\in B_{\epsilon_R}(\xi^{p^\star})$,
\begin{equation}
\begin{aligned}
&\|\hat R_{\xi^{p_k}}(\boldsymbol\beta(\xi^{p_k}))
-(\xi^{p_k}+\boldsymbol\beta(\xi^{p_k}))\|=\|\mathbf r_R(\xi^{p_k})\|
\leq\frac{C_\Gamma}{2}
\|\boldsymbol\beta(\xi^{p_k})\|^2\leq\gamma_R
\|\xi^{p_k}-\xi^{p^\star}\|^2,
\end{aligned}
\label{eq:retraction_error_bound}
\end{equation}
where $\gamma_R=2C_\Gamma\kappa^2L_f^2$.

\noindent\textbf{Step 4: Establish an upper bound
for the Newton direction error.}

By the definition of the Newton direction
and the identity
$\mathbf{M}^{\nabla}(\xi^{p^\star})
=\boldsymbol{\nabla}^2f(\xi^{p^\star})$,
we have
\[
\begin{aligned}
&\mathbf{M}^{\nabla}(\xi^{p_k})
\bigl(\xi^{p_k}+\boldsymbol{\beta}(\xi^{p_k})
-\xi^{p^\star}\bigr)=\bigl(\mathbf{M}^{\nabla}(\xi^{p_k})
-\mathbf{M}^{\nabla}(\xi^{p^\star})\bigr)
(\xi^{p_k}-\xi^{p^\star})+\boldsymbol{\nabla}^2f(\xi^{p^\star})
(\xi^{p_k}-\xi^{p^\star})
-\boldsymbol{\nabla}f(\xi^{p_k}).
\end{aligned}
\]

Consequently, by the triangle inequality
and the consistency of the matrix norm, we obtain

\begin{equation}
\begin{aligned}
&\|\xi^{p_k}+\boldsymbol{\beta}(\xi^{p_k})
-\xi^{p^\star}\|
\leq\|[\mathbf{M}^{\nabla}(\xi^{p_k})]^{-1}\|\\
&\times\Bigl(
\|\mathbf{M}^{\nabla}(\xi^{p_k})
-\mathbf{M}^{\nabla}(\xi^{p^\star})\|
\|\xi^{p_k}-\xi^{p^\star}\|
+\|\boldsymbol{\nabla}^2f(\xi^{p^\star})
(\xi^{p_k}-\xi^{p^\star})
-\boldsymbol{\nabla}f(\xi^{p_k})\|
\Bigr)\\
&\leq\|[\mathbf{M}^{\nabla}(\xi^{p_k})]^{-1}\|
\left(\gamma_M+\frac{L_H}{2}\right)
\|\xi^{p_k}-\xi^{p^\star}\|^2\\
&\leq2\kappa
\left(\gamma_M+\frac{L_H}{2}\right)
\|\xi^{p_k}-\xi^{p^\star}\|^2.
\end{aligned}
\label{eq:newton_direction_error}
\end{equation}

where the last two inequalities follow from
Equations~\eqref{eq:gammaM},
\eqref{eq:gradient_taylor_bound},
and~\eqref{eq:inverse_M_bound}.

\noindent\textbf{Step5: Local quadratic
convergence and neighborhood invariance.}

By the coordinate update,
$
\xi^{p_{k+1}}
=
\hat R_{\xi^{p_k}}
\bigl(\boldsymbol{\beta}(\xi^{p_k})\bigr).
$
Therefore,
\[
\begin{aligned}
&\xi^{p_{k+1}}-\xi^{p^\star}=
\xi^{p_k}+\boldsymbol{\beta}(\xi^{p_k})
-\xi^{p^\star}+
\hat R_{\xi^{p_k}}
\bigl(\boldsymbol{\beta}(\xi^{p_k})\bigr)
-\bigl(\xi^{p_k}+\boldsymbol{\beta}(\xi^{p_k})\bigr).
\end{aligned}
\]

Combining
Equations~\eqref{eq:newton_direction_error}
and~\eqref{eq:retraction_error_bound},
we obtain
\begin{equation}
\begin{aligned}
&\|\xi^{p_{k+1}}-\xi^{p^\star}\|\leq
\|\xi^{p_k}+\boldsymbol{\beta}(\xi^{p_k})
-\xi^{p^\star}\|+
\|\hat R_{\xi^{p_k}}
(\boldsymbol{\beta}(\xi^{p_k}))
-(\xi^{p_k}+\boldsymbol{\beta}(\xi^{p_k}))\|\leq
C\|\xi^{p_k}-\xi^{p^\star}\|^2,
\end{aligned}
\label{eq:quadratic_convergence_bound}
\end{equation}
where
$
C=
2\kappa\left(\gamma_M+\frac{L_H}{2}\right)
+\gamma_R.
$

Finally, choose $\epsilon^\star>0$
sufficiently small such that
$
\epsilon^\star
\leq\min\{\epsilon_M,\epsilon_R,r_H\}, C\epsilon^\star<1,
$
and all coordinate updates are well-defined
in $B_{\epsilon^\star}(\xi^{p^\star})$.

For any
$\xi^{p_k}\in B_{\epsilon^\star}(\xi^{p^\star})$,
Equation~\eqref{eq:quadratic_convergence_bound}
gives
\[
\begin{aligned}
\|\xi^{p_{k+1}}-\xi^{p^\star}\|
&\leq C\|\xi^{p_k}-\xi^{p^\star}\|^2\leq C\epsilon^\star
\|\xi^{p_k}-\xi^{p^\star}\|<\epsilon^\star.
\end{aligned}
\]

Hence, $B_{\epsilon^\star}(\xi^{p^\star})$
is invariant under the iteration.
By induction, all subsequent iterates
remain in this neighborhood, where the
Newton system is nonsingular and all
the preceding estimates hold.
Therefore, the iterates converge locally
at least quadratically to $p^\star$,
which completes the proof.

\end{proof}

Existing Riemannian methods with cubic convergence,
such as the Chebyshev--Halley method, involve
geodesic exponential maps and parallel transport
~\cite{Castro2018Chebyshev}.
In contrast, for the potential-based objectives
in \eqref{eq:potential-based-objective}, we show
that the proposed dual Riemannian Newton method
achieves local cubic convergence by choosing the specific affine connection, using a second-order coordinate
retraction without explicitly computing parallel
transport or exact geodesics.

\begin{corollary}
\label{cor:alpha_one_third}
Under the assumptions of Theorem~\ref{convergence},
suppose further that $(\mathcal M,g,\nabla,\nabla^*)$
is a dually flat manifold equipped with the
interpolated connections
$
\nabla^{(\alpha)}
=
\frac{1+\alpha}{2}\nabla
+
\frac{1-\alpha}{2}\nabla^*,
\nabla^{(-\alpha)}
=
\frac{1-\alpha}{2}\nabla
+
\frac{1+\alpha}{2}\nabla^*.
$
Let $\theta$ be a $\nabla^{(1)}$-affine coordinate
system, and suppose that
\[
f(\theta)=\psi(\theta)-\mathbf{c}^{\top}\theta+C_0,
\]
where $\mathbf{c}\in\mathbb{R}^n$ is a constant
vector, $C_0\in\mathbb{R}$ is a constant, and $\psi\in C^4$ is a strictly convex potential
of the Hessian metric $g$ with respect to $\nabla$,
satisfying
$
\mathbf G(\theta)=\boldsymbol{\nabla}^2\psi(\theta)
$
in the $\nabla$-affine coordinates $\theta$.
Consider Algorithm~\ref{alg:Computation dual riemannian newton}
with the Newton coefficient matrix
$\mathbf M^{\nabla^{(\alpha)}}$ associated with
$\nabla^{(\alpha)}$ and the second-order
coordinate retraction associated with
$\nabla^{(\alpha)}$.

Then, for all sufficiently close iterates,
\begin{equation}
\theta_i^{p_{k+1}}-\theta_i^{p^\star}
=\frac{3\alpha-1}{4}\sum_{\ell,j,r=1}^n g^{i\ell}(\theta^{p^\star})
\frac{\partial^3\psi}{\partial\theta_\ell\partial\theta_j\partial\theta_r}\biggr|_{\theta^{p^\star}}
(\theta_j^{p_k}-\theta_j^{p^\star})(\theta_r^{p_k}-\theta_r^{p^\star})
+O(\|\theta^{p_k}-\theta^{p^\star}\|^3).
\label{eq:alpha_local_error}
\end{equation}
In particular, when $\alpha=1/3$, the method
converges locally at least cubically, i.e.,
there exists $K>0$ such that
\begin{equation}
\|\theta^{p_{k+1}}-\theta^{p^\star}\|
\leq K\|\theta^{p_k}-\theta^{p^\star}\|^3
\label{eq:alpha_cubic_rate}
\end{equation}
for all sufficiently close initial points.
\end{corollary}

\begin{proof}

\noindent\textbf{Step 1: Establish the local Taylor
expansions of the Newton coefficient matrix.}

Since $f(\theta)=\psi(\theta)-\mathbf{c}^{\top}\theta+C_0$, we have
\[
\frac{\partial^2 f}{\partial\theta_i\partial\theta_j}=\frac{\partial^2 \psi}{\partial\theta_i\partial\theta_j}=g_{ij}(\theta),\qquad \frac{\partial^3 f}{\partial\theta_i\partial\theta_j\partial\theta_r}=\frac{\partial^3\psi}{\partial\theta_i\partial\theta_j\partial\theta_r}.
\]
In the $\nabla^{(1)}$-affine coordinates, the Christoffel symbols of $\nabla^{(\alpha)}$ satisfy
\begin{equation}
\Gamma^{(\alpha)k}_{ij}(\theta)=\frac{1-\alpha}{2}\sum_{\ell=1}^n g^{k\ell}(\theta)\frac{\partial^3\psi(\theta)}{\partial\theta_i\partial\theta_j\partial\theta_\ell}.
\label{eq:alpha_connection_coefficients}
\end{equation}

By the definition of $\mathbf{M}^{\nabla}$,
the Newton coefficient matrix is
\begin{equation}
\mathbf{M}^{\nabla^{(\alpha)}}_{ij}(\theta)
=g_{ij}(\theta)
-\sum_{k=1}^n
\Gamma^{(\alpha)k}_{ij}(\theta)
\frac{\partial f(\theta)}{\partial\theta_k}.
\label{eq:alpha_newton_matrix}
\end{equation}

Since $\boldsymbol{\nabla}f(\theta^{p^\star})=0$,
Equation~\eqref{eq:alpha_newton_matrix} implies
$
\mathbf{M}_{ij}^{\nabla^{(\alpha)}}(\theta^{p^\star})
=g_{ij}(\theta^{p^\star}).
$
Moreover, since $\psi\in C^4$ and
$\boldsymbol{\nabla}f(\theta^{p^\star})=0$,
there exists $r_f>0$ such that Taylor's theorem yields
\begin{equation}
\frac{\partial f(\theta)}{\partial\theta_i}=\sum_{j=1}^n g_{ij}(\theta^{p^\star})(\theta_j-\theta_j^{p^\star})+\frac12\sum_{j,r=1}^n\frac{\partial^3\psi}{\partial\theta_i\partial\theta_j\partial\theta_r}\biggr|_{\theta^{p^\star}}(\theta_j-\theta_j^{p^\star})(\theta_r-\theta_r^{p^\star})+O(\|\theta-\theta^{p^\star}\|^3).
\label{eq:alpha_gradient_expansion}
\end{equation}
for all $\theta\in B_{r_f}(\theta^{p^\star})$.

Similarly, since $\psi\in C^4$, there exists
$r_g>0$ such that applying Taylor's theorem
to $g_{ij}=\partial_{ij}\psi$
at $\theta^{p^\star}$ yields
\begin{equation}
\begin{aligned}
g_{ij}(\theta)
&=g_{ij}(\theta^{p^\star})
+\sum_{r=1}^n
\left.
\frac{\partial^3\psi}
{\partial\theta_i\partial\theta_j
\partial\theta_r}
\right|_{\theta^{p^\star}}
(\theta_r-\theta_r^{p^\star})+O(\|\theta-\theta^{p^\star}\|^2).
\end{aligned}
\label{eq:alpha_metric_taylor_expansion}
\end{equation}
for all $\theta\in B_{r_g}(\theta^{p^\star})$.

Moreover, by
Equation~\eqref{eq:alpha_gradient_expansion},
\begin{equation}
\begin{aligned}
\frac{\partial f(\theta)}{\partial\theta_s}
&=\sum_{r=1}^n
g_{sr}(\theta^{p^\star})
(\theta_r-\theta_r^{p^\star})
+O(\|\theta-\theta^{p^\star}\|^2).
\end{aligned}
\label{eq:alpha_gradient_first_order_expansion}
\end{equation}
Since the Christoffel symbols are
continuously differentiable, there exists $r_{\Gamma}>0$ such that Taylor's theorem yields
\begin{equation}
\Gamma_{ij}^{(\alpha)s}(\theta)
=\Gamma_{ij}^{(\alpha)s}(\theta^{p^\star})
+O(\|\theta-\theta^{p^\star}\|).
\label{eq:alpha_christoffel_expansion}
\end{equation}
for all $\theta\in B_{r_{\Gamma}}(\theta^{p^\star})$.
By substituting Equation~\eqref{eq:alpha_connection_coefficients}
into the summation involving the Christoffel
symbols and the Riemannian metric, and using the
inverse-metric identity
$
\sum_{s=1}^n
g^{s\ell}(\theta^{p^\star})
g_{sr}(\theta^{p^\star})
=\delta_{\ell r},
$
we obtain
\begin{equation}
\begin{aligned}
&\sum_{s=1}^n
\Gamma_{ij}^{(\alpha)s}(\theta^{p^\star})
g_{sr}(\theta^{p^\star})=\frac{1-\alpha}{2}
\left.
\frac{\partial^3\psi}
{\partial\theta_i\partial\theta_j
\partial\theta_r}
\right|_{\theta^{p^\star}}.
\end{aligned}
\label{eq:alpha_christoffel_metric_identity}
\end{equation}

Choose $r_T>0$ such that
$r_T<\min\{r_f,r_g,r_\Gamma\}$
and $\overline{B}_{r_T}(\theta^{p^\star})$
is contained in the coordinate domain.
Then, all the preceding Taylor expansions
hold with uniform remainder bounds on
$B_{r_T}(\theta^{p^\star})$.

Substituting
\eqref{eq:alpha_christoffel_expansion}
and
\eqref{eq:alpha_gradient_first_order_expansion}
into the $\sum_{s=1}^n
\Gamma_{ij}^{(\alpha)s}(\theta)
\frac{\partial f(\theta)}{\partial\theta_s}$, and using
\eqref{eq:alpha_christoffel_metric_identity},
we obtain
\begin{equation}
\begin{aligned}
\sum_{s=1}^n
\Gamma_{ij}^{(\alpha)s}(\theta)
\frac{\partial f(\theta)}{\partial\theta_s}
&=\sum_{s,r=1}^n
\Gamma_{ij}^{(\alpha)s}(\theta^{p^\star})
g_{sr}(\theta^{p^\star})
(\theta_r-\theta_r^{p^\star})
+O(\|\theta-\theta^{p^\star}\|^2)\\
&=\frac{1-\alpha}{2}
\sum_{r=1}^n
\left.
\frac{\partial^3\psi}
{\partial\theta_i\partial\theta_j\partial\theta_r}
\right|_{\theta^{p^\star}}
(\theta_r-\theta_r^{p^\star})
+O(\|\theta-\theta^{p^\star}\|^2).
\end{aligned}
\label{eq:alpha_christoffel_gradient_preliminary}
\end{equation}

Substituting
\eqref{eq:alpha_metric_taylor_expansion}
and
\eqref{eq:alpha_christoffel_gradient_preliminary}
into the right-hand side of \eqref{eq:alpha_newton_matrix},
we obtain
\begin{equation}
\begin{aligned}
\mathbf{M}_{ij}^{\nabla^{(\alpha)}}(\theta)
&=g_{ij}(\theta^{p^\star})
+\frac{1+\alpha}{2}\sum_{r=1}^n
\left.
\frac{\partial^3\psi}
{\partial\theta_i\partial\theta_j
\partial\theta_r}
\right|_{\theta^{p^\star}}
(\theta_r-\theta_r^{p^\star})
+O(\|\theta-\theta^{p^\star}\|^2).
\end{aligned}
\label{eq:alpha_matrix_expansion}
\end{equation}
for all $\theta\in B_{r_T}(\theta^{p^\star})$,
where the remainder is uniformly bounded
by $C_T\|\theta-\theta^{p^\star}\|^2$
for some constant $C_T>0$.

\noindent\textbf{Step 2: Establish local
invertibility and the inverse matrix expansion.}

By Equation~\eqref{eq:alpha_matrix_expansion},
we have
\begin{equation}
\mathbf{M}^{\nabla^{(\alpha)}}(\theta)
=\mathbf{G}(\theta^{p^\star})
+O(\|\theta-\theta^{p^\star}\|).
\label{eq:alpha_matrix_first_order_bound}
\end{equation}

Since
$\mathbf{M}^{\nabla^{(\alpha)}}(\theta^{p^\star})
=\mathbf{G}(\theta^{p^\star})\succ0$,
let
$\kappa=\|\mathbf{G}^{-1}(\theta^{p^\star})\|$.
By Equation~\eqref{eq:alpha_matrix_first_order_bound}
and the same argument as in Step 2 of the
proof of Theorem~\ref{convergence},
there exists $r_M>0$, with $r_M<r_T$,
such that $\mathbf{M}^{\nabla^{(\alpha)}}(\theta)$
is nonsingular and satisfies
\begin{equation}
\|[\mathbf{M}^{\nabla^{(\alpha)}}(\theta)]^{-1}\|
\leq 2\kappa.
\label{eq:alpha_inverse_matrix_bound}
\end{equation}
for all $\theta\in B_{r_M}(\theta^{p^\star})$.

 Using the matrix inverse identity, we obtain
\[
\begin{aligned}
&[\mathbf{M}^{\nabla^{(\alpha)}}(\theta)]^{-1}
-\mathbf{G}^{-1}(\theta^{p^\star})=[\mathbf{M}^{\nabla^{(\alpha)}}(\theta)]^{-1}
\bigl(\mathbf{G}(\theta^{p^\star})
-\mathbf{M}^{\nabla^{(\alpha)}}(\theta)\bigr)
\mathbf{G}^{-1}(\theta^{p^\star}).
\end{aligned}
\]
Consequently, by the submultiplicative
property of the matrix norm, and
Equations~\eqref{eq:alpha_matrix_first_order_bound},
we obtain
\[
\begin{aligned}
&\|[\mathbf{M}^{\nabla^{(\alpha)}}(\theta)]^{-1}
-\mathbf{G}^{-1}(\theta^{p^\star})\|\leq
\|[\mathbf{M}^{\nabla^{(\alpha)}}(\theta)]^{-1}\|
\|\mathbf{M}^{\nabla^{(\alpha)}}(\theta)
-\mathbf{G}(\theta^{p^\star})\|
\|\mathbf{G}^{-1}(\theta^{p^\star})\|=O(\|\theta-\theta^{p^\star}\|).
\end{aligned}
\]

Therefore,
\begin{equation}
[\mathbf{M}^{\nabla^{(\alpha)}}(\theta)]^{-1}
=\mathbf{G}^{-1}(\theta^{p^\star})
+O(\|\theta-\theta^{p^\star}\|).
\label{eq:alpha_inverse_matrix_expansion}
\end{equation}

\noindent\textbf{Step 3: Establish the local
expansion of the Newton direction.}

Since $r_M<r_T$, the Taylor expansions
\eqref{eq:alpha_gradient_expansion}
and~\eqref{eq:alpha_matrix_expansion},
together with the inverse matrix bound
\eqref{eq:alpha_inverse_matrix_bound},
hold uniformly for all
$\theta\in B_{r_M}(\theta^{p^\star})$.

From the Newton equation
$\mathbf{M}^{\nabla^{(\alpha)}}(\theta)
\boldsymbol{\beta}(\theta)
=-\boldsymbol{\nabla}f(\theta)$,
we have
\begin{equation}
\begin{aligned}
&\sum_{j=1}^n
\mathbf{M}^{\nabla^{(\alpha)}}_{ij}(\theta)
\bigl(\beta_j(\theta)+\theta_j-\theta_j^{p^\star}\bigr)=\sum_{j=1}^n
\mathbf{M}^{\nabla^{(\alpha)}}_{ij}(\theta)
(\theta_j-\theta_j^{p^\star})
-\frac{\partial f(\theta)}{\partial\theta_i}.
\end{aligned}
\label{eq:alpha_newton_direction_identity}
\end{equation}

Substituting
Equations~\eqref{eq:alpha_gradient_expansion}
and~\eqref{eq:alpha_matrix_expansion}
into the right-hand side of
Equation~\eqref{eq:alpha_newton_direction_identity},
the first-order terms cancel, yielding
\[
\begin{aligned}
&\sum_{j=1}^n
\mathbf{M}^{\nabla^{(\alpha)}}_{ij}(\theta)
\bigl(\beta_j(\theta)+\theta_j-\theta_j^{p^\star}\bigr)=\frac{\alpha}{2}\sum_{j,r=1}^n
\left.
\frac{\partial^3\psi}
{\partial\theta_i\partial\theta_j\partial\theta_r}
\right|_{\theta^{p^\star}}
(\theta_j-\theta_j^{p^\star})
(\theta_r-\theta_r^{p^\star})
+O(\|\theta-\theta^{p^\star}\|^3).
\end{aligned}
\]
By Equation~\eqref{eq:alpha_inverse_matrix_bound},
$\mathbf{M}^{\nabla^{(\alpha)}}(\theta)$
is nonsingular on
$B_{r_M}(\theta^{p^\star})$.
Multiplying by its inverse yields
\begin{equation}
\begin{aligned}
&\beta_i(\theta)+\theta_i-\theta_i^{p^\star}=\frac{\alpha}{2}
\sum_{\ell,j,r=1}^n
\bigl[
\mathbf{M}^{\nabla^{(\alpha)}}(\theta)^{-1}
\bigr]_{i\ell}
\left.
\frac{\partial^3\psi}
{\partial\theta_\ell\partial\theta_j\partial\theta_r}
\right|_{\theta^{p^\star}}
(\theta_j-\theta_j^{p^\star})
(\theta_r-\theta_r^{p^\star})
+O(\|\theta-\theta^{p^\star}\|^3).
\end{aligned}
\label{eq:alpha_direction_inverse_expansion}
\end{equation}

Moreover, by
Equation~\eqref{eq:alpha_inverse_matrix_expansion},
the first-order error in the inverse matrix,
when multiplied by the quadratic terms,
contributes only
$O(\|\theta-\theta^{p^\star}\|^3)$.
Consequently,
Equation~\eqref{eq:alpha_direction_inverse_expansion}
becomes
\begin{equation}
\begin{aligned}
\beta_i(\theta)
&=-(\theta_i-\theta_i^{p^\star})
+\frac{\alpha}{2}
\sum_{\ell,j,r=1}^n
g^{i\ell}(\theta^{p^\star})
\left.
\frac{\partial^3\psi}
{\partial\theta_\ell\partial\theta_j\partial\theta_r}
\right|_{\theta^{p^\star}}
(\theta_j-\theta_j^{p^\star})
(\theta_r-\theta_r^{p^\star})
+O(\|\theta-\theta^{p^\star}\|^3).
\end{aligned}
\label{eq:alpha_direction_expansion}
\end{equation}
This expansion holds uniformly for all
$\theta\in B_{r_M}(\theta^{p^\star})$.

\noindent\textbf{Step 4: Establish the retraction
expansion and local cubic convergence.}
Since $\boldsymbol{\beta}(\theta^{p^\star})=0$
and the coordinate update map is continuous,
there exists $r_R>0$, with $r_R<r_M$,
such that the retraction update is
well-defined in the coordinate domain
for all $\theta\in B_{r_R}(\theta^{p^\star})$.
We now consider the second-order coordinate retraction in Algorithm~\ref{alg:Computation dual riemannian newton}:
\begin{equation}
\theta_i^{p_{k+1}}
=\theta_i^{p_k}
+\beta_i(\theta^{p_k})
-\frac12\sum_{j,r=1}^n
\Gamma_{jr}^{(\alpha)i}(\theta^{p_k})
\beta_j(\theta^{p_k})
\beta_r(\theta^{p_k}).
\label{eq:alpha_coordinate_retraction_update}
\end{equation}
By Equation~\eqref{eq:alpha_direction_expansion},
we have
$\beta_i(\theta)=-(\theta_i-\theta_i^{p^\star})
+O(\|\theta-\theta^{p^\star}\|^2)$.
Consequently, using
Equations~\eqref{eq:alpha_connection_coefficients}
and~\eqref{eq:alpha_christoffel_expansion},
we obtain
\begin{equation}
\begin{aligned}
&-\frac12\sum_{j,r=1}^n
\Gamma^{(\alpha)i}_{jr}(\theta)
\beta_j(\theta)\beta_r(\theta)=-\frac{1-\alpha}{4}
\sum_{\ell,j,r=1}^n
g^{i\ell}(\theta^{p^\star})
\left.
\frac{\partial^3\psi}
{\partial\theta_\ell\partial\theta_j\partial\theta_r}
\right|_{\theta^{p^\star}}
(\theta_j-\theta_j^{p^\star})
(\theta_r-\theta_r^{p^\star})
+O(\|\theta-\theta^{p^\star}\|^3).
\end{aligned}
\label{eq:alpha_retraction_expansion}
\end{equation}

Substituting
\eqref{eq:alpha_direction_expansion}
and \eqref{eq:alpha_retraction_expansion}
into \eqref{eq:alpha_coordinate_retraction_update},
we obtain
\[
\theta_i^{p_{k+1}}-\theta_i^{p^\star}=\left(\frac{3\alpha-1}{4}\right)\sum_{\ell,j,r=1}^n g^{i\ell}(\theta^{p^\star})\frac{\partial^3\psi}{\partial\theta_\ell\partial\theta_j\partial\theta_r}\biggr|_{\theta^{p^\star}}(\theta_j^{p_k}-\theta_j^{p^\star})(\theta_r^{p_k}-\theta_r^{p^\star})+O(\|\theta^{p_k}-\theta^{p^\star}\|^3),
\]
for all $\theta^{p_k}\in B_{r_R}(\theta^{p^\star})$. This establishes Equation~\eqref{eq:alpha_local_error}.

In particular, when $\alpha=1/3$,
the quadratic term in
Equation~\eqref{eq:alpha_local_error}
vanishes. Hence, there exists $K>0$ such that
\begin{equation}
\|\theta^{p_{k+1}}-\theta^{p^\star}\|
\leq K\|\theta^{p_k}-\theta^{p^\star}\|^3
\end{equation}
for all
$\theta^{p_k}\in B_{r_R}(\theta^{p^\star})$.

Choose $r>0$ such that
$r<r_R$ and $Kr^2<1$.
Then, for any
$\theta^{p_k}\in B_r(\theta^{p^\star})$,
we have
\[
\begin{aligned}
\|\theta^{p_{k+1}}-\theta^{p^\star}\|
&\leq K\|\theta^{p_k}-\theta^{p^\star}\|^3\leq Kr^2\|\theta^{p_k}-\theta^{p^\star}\|<r.
\end{aligned}
\]
Therefore, $B_r(\theta^{p^\star})$
is invariant under the iteration.
Since $Kr^2<1$, the errors converge
to zero. Together with the cubic error
bound, this proves local convergence
of at least cubic order.
\end{proof}

\begin{remark}
Let $\eta=\nabla_\theta\psi(\theta)$ be the
$\nabla^*$-affine coordinates, and let
$\varphi$ be the Legendre dual potential of $\psi$,
satisfying
$
g_{ij}(\eta)
=
\frac{\partial^2\varphi(\eta)}
{\partial\eta_i\partial\eta_j}.
$
For an objective of the form
$
f(\eta)=\varphi(\eta)-\mathbf c^\top\eta+C_0,
$
the corresponding result follows by applying
the same argument in the $\eta$-coordinates.
Since $\nabla^*$ is flat in $\eta$ coordinates,
the Christoffel symbols of the interpolated
connection $\nabla^{(\alpha)}$ satisfy
$
\Gamma^{(\alpha)k}_{ij}(\eta)
=
\frac{1+\alpha}{2}
\sum_{\ell=1}^n
g^{k\ell}(\eta)
\frac{\partial^3\varphi(\eta)}
{\partial\eta_i\partial\eta_j\partial\eta_\ell}.
$
Following the same Taylor expansions as in the
proof of Corollary~\ref{cor:alpha_one_third},
the quadratic error coefficient becomes
$-(3\alpha+1)/4$.
Consequently, choosing $\alpha=-1/3$,
the method using $\mathbf M^{\nabla^{(-1/3)}}$
and the associated second-order coordinate
retraction converges locally at least cubically in the $\nabla^*$-affine coordinates $\eta$:
$
\|\eta^{p_{k+1}}-\eta^{p^\star}\|
\leq K\|\eta^{p_k}-\eta^{p^\star}\|^3.
$
\end{remark}


\section{Experiments}
\label{chap6}

To verify the local convergence properties of the dual Riemannian Newton method,
we conduct numerical evaluations under three concrete optimization problems:
(1) minimization of the KL divergence with $\ell_2$ regularization for a log-linear model,
(2) $\alpha$-divergence minimization for Gaussian models,
and (3) MLE for Beta mixtures.
In addition, we numerically investigate the local convergence order
using a high-accuracy reference solution and the empirical estimate
$
q_k =
\frac{\log(e_{k+1}/e_k)}
{\log(e_k/e_{k-1})},
$
where $e_k$ denotes the error relative to the reference solution. For the numerical experiments, we adopt the
$\alpha$-connections defined on the probability
distribution family $p(x;\xi)$, as introduced
in Example~\ref{alpha_geometry}.
Moreover, we compare the dual Riemannian Newton method with
the natural gradient method and Adam~\cite{Kingma2014},
which is widely used for optimization of neural networks,
to examine how efficient the local quadratic convergence is
compared with other approaches.

All experiments are conducted on an Apple M3 Pro
with 36 GB of memory. The Newton systems in
Algorithm 2 are solved using dense direct
linear solvers without explicitly computing
matrix inverses. The reported runtimes
include derivative evaluations, construction
and solution of the Newton systems, Fisher
metric and connection computations, line
searches where applicable, and stopping tests.

\subsection{Experiment1: Log-linear Models with Higher-Order Interactions}
We first examine the performance of optimization methods in the scenario of \emph{interaction decomposition}, which is a central topic in multivariate discrete modeling: it hierarchically separates the main effects, pairwise effects, and higher-order interactions, clarifying which orders of interaction are necessary for reliable modeling and interpretation \cite{nakahara2006comparison}. Specifically, \emph{log-linear models}~\cite{Agresti2013} provide a natural representation of this decomposition within the exponential family framework, endowed with the geometric structure of dual coordinates \((\theta,\eta)\) and dual flatness under \((\nabla^{(e)},\nabla^{(m)})\). 
In practice, pairwise modeling, known as Boltzmann machines or Ising models, is widely used for learning and inference \cite{Ackley1985}, and their tensor formulations are known to be connected to structured higher–order interactions as well as scaling/balancing perspectives~\cite{sugiyama2017tensor}. 
In this context, we explicitly describe the log–linear modeling framework that realizes interaction decomposition for empirical distributions, and then formulate the optimization problems and algorithmic comparisons that follow.

More precisely, we use an $n$th-order tensor
$
\hat{\mathcal{P}} \in
\mathbb{R}_{>0}^{2 \times 2 \times \cdots \times 2}
$
to represent a strictly positive target distribution;
that is, we assume that
$\hat{\mathcal{P}} \in \mathcal{S}_{\Omega}$,
where $\mathcal{S}_{\Omega}$ denotes the set of all
strictly positive probability distributions on
$\{0,1\}^n$.

Its log-probability can be expressed as a log-linear
model with interaction terms indexed by $\Omega$.
Hence, $\mathcal{S}_{\Omega}$ is explicitly described as
\[
\mathcal{S}_{\Omega}
:=
\biggl\{
\mathcal{P}_{\theta}
\;\Big|\;
\log \mathcal{P}_{\theta}(x)
=
\sum_{A \in \Omega}
\theta^{A} \prod_{i \in A} x^i
-
\psi_{\Omega}(\theta),
\quad
\theta\in\mathbb{R}^{2^n-1}
\biggr\},
\]
where
$
\Omega
=
\bigcup_{k=1}^{n}
\left\{
i_1 i_2 \cdots i_k
\mid
1 \leq i_1 < i_2 < \cdots < i_k \leq n
\right\}.
$
Here, $\psi_{\Omega}(\theta)$ is the log-partition
function of the full exponential family, given by
$
\psi_{\Omega}(\theta)
=
\log
\sum_{x\in\{0,1\}^n}
\exp\left(
\sum_{A\in\Omega}
\theta^A\prod_{i\in A}x^i
\right).
$

Thus, $\log\mathcal{P}_{\theta}(x)$ can be
alternatively written as
\begin{equation}\label{eq:loglinear}
\log \mathcal{P}_{\theta}(x)
= \sum_i \theta^i x^i
+ \sum_{i<j} \theta^{ij} x^i x^j
+ \sum_{i<j<k} \theta^{ijk} x^i x^j x^k
+ \cdots
+ \theta^{1\ldots n} x^1 x^2 \cdots x^n
- \psi_{\Omega}(\theta).
\end{equation}
where
$x=(x^1,\ldots,x^n)\in\{0,1\}^n$
is an $n$-dimensional binary vector,
and
$\theta=(\theta^A)_{A\in\Omega}$
is the natural parameter vector of the full
exponential family.

The corresponding expectation parameters
$\eta=(\eta^A)_{A\in\Omega}$
represent the expected values of variable
combinations:
\[
\begin{aligned}
\eta^i
&=
\mathbb{E}[x^i]
=
\Pr(x^i=1),\\
\eta^{ij}
&=
\mathbb{E}[x^i x^j]
=
\Pr(x^i=x^j=1),
\quad \dots,\\
\eta^{1\ldots n}
&=
\mathbb{E}[x^1\cdots x^n]
=
\Pr(x^1=\cdots=x^n=1).
\end{aligned}
\]

The log-linear model belongs to the exponential
family, which forms a dually flat manifold equipped
with a pair of dual affine connections
$(\nabla^{(e)},\nabla^{(m)})$.
The $\alpha$-connections for
$\alpha\neq\pm1$ are generally not flat. The fully connected Boltzmann machine
$\mathcal{S}_{\mathcal{B}}$~\cite{Ackley1985}
is a special case of the log-linear model,
which considers interactions up to the second order.

Formally, let
$
\mathcal B
=
\{i\mid 1\leq i\leq n\}
\cup
\{ij\mid 1\leq i<j\leq n\}.
$
The fully connected Boltzmann family is defined by
\begin{equation}\label{eq:boltzmann_family}
\mathcal S_{\mathcal B}
:=
\biggl\{
\mathcal P_{\theta_{\mathcal B}}
\;\Big|\;
\log\mathcal P_{\theta_{\mathcal B}}(x)
=
\sum_i\theta^i x^i
+
\sum_{i<j}\theta^{ij}x^ix^j
-
\psi_{\mathcal B}(\theta_{\mathcal B})
\biggr\},
\end{equation}
where
$\theta_{\mathcal B}=(\theta^A)_{A\in\mathcal B}
\in\mathbb R^{n+\binom n2}$,
with $\theta^i$ and $\theta^{ij}$ representing
biases and interaction weights, respectively.
The family $\mathcal S_{\mathcal B}\subset\mathcal S_\Omega$
is obtained by setting $\theta^A=0$ for
$A\in\Omega\setminus\mathcal B$.
Thus,
$\psi_{\mathcal B}(\theta_{\mathcal B})
=\psi_\Omega(\theta_{\mathcal B},0)$,
where $(\theta_{\mathcal B},0)$ denotes
the zero-padded full natural parameter vector.
As an exponential family,
$\mathcal S_{\mathcal B}$ is intrinsically
dually flat, equipped with the Fisher metric
and its associated dual affine
connections~\cite{Amari2001}.

We aim to approximate a target distribution
$\hat{\mathcal P}\in\mathcal S_\Omega$,
which need not belong to $\mathcal S_{\mathcal B}$,
by solving
\begin{equation}\label{eq:klop}
\bar{\mathcal P}
=
\mathop{\mathrm{argmin}}_{\mathcal P\in\mathcal S_{\mathcal B}}
D_{\mathrm{KL}}(\hat{\mathcal P}\|\mathcal P).
\end{equation}

The dual potential of the full exponential family
is given by the Legendre transform of $\psi_\Omega$:
\[
\varphi_\Omega(\eta)
=
\sup_{\vartheta\in\mathbb R^{2^n-1}}
\left(
\sum_{A\in\Omega}\vartheta^A\eta^A
-
\psi_\Omega(\vartheta)
\right).
\]
Accordingly, the KL divergence can be expressed
as the $\nabla$-divergence on $\mathcal S_\Omega$:

\begin{equation}
D_{\mathrm{KL}}(\hat{\mathcal P}\|\mathcal P)
= D_\Omega(\mathcal P\|\hat{\mathcal P})
= \psi_{\mathcal B}(\theta_{\mathcal B})
+ \varphi_\Omega(\eta_{\hat{\mathcal P}})
- \langle\theta_{\mathcal B},\hat\eta_{\mathcal B}\rangle_E.
\end{equation}

where
$\hat\eta_{\mathcal B}
=(\eta_{\hat{\mathcal P}}^A)_{A\in\mathcal B}$
denotes the expectation vector of the
restricted sufficient statistics under
$\hat{\mathcal P}$.
The argument order follows Definition IV.2,
where $D_\Omega(p\|q)=D_{\mathrm{KL}}(q\|p)$.
Moreover, since $\varphi_\Omega$ is the
Legendre dual potential of the full
exponential family, it satisfies
$
\varphi_\Omega(\eta_{\hat{\mathcal P}})
=
\sum_{x\in\Omega}
\hat{\mathcal P}(x)
\log\hat{\mathcal P}(x),
$
which equals the negative Shannon entropy
of $\hat{\mathcal P}$.

To numerically verify the local convergence order,
we consider the unregularized KL minimization
problem~\eqref{eq:klop} with $n=16$ binary variables.
The restricted Boltzmann family
$\mathcal S_{\mathcal B}$ has
$d=n+\binom{n}{2}=136$ natural parameters,
consisting of $16$ biases and $120$ pairwise
interaction weights. We construct a target distribution
$\hat{\mathcal P}\in
\mathcal S_\Omega\setminus\mathcal S_{\mathcal B}$
by introducing a nonzero third-order interaction:
\begin{equation}
\hat{\mathcal P}(x)
=
\exp\left(
\sum_{A\in\mathcal B}
\theta_{\hat{\mathcal P}}^A
\prod_{i\in A}x^i
+0.1x^1x^2x^3
-\psi_\Omega(\theta_{\hat{\mathcal P}})
\right),
\end{equation}
where $\psi_\Omega$ is the log-partition function
of the full exponential family.
The bias parameters
$\theta_{\hat{\mathcal P}}^i$
are independently sampled from
$\mathcal U(-0.1,0.1)$, while the pairwise
parameters $\theta_{\hat{\mathcal P}}^{ij}$
are sampled from $\mathcal U(-0.02,0.02)$.
All higher-order interaction parameters
are set to zero except
$\theta_{\hat{\mathcal P}}^{123}=0.1$.
Consequently, $\hat{\mathcal P}$ does not
belong to the restricted Boltzmann family.

Let
$\bar{\mathcal P}
=\mathcal P_{\bar\theta_{\mathcal B}}
\in\mathcal S_{\mathcal B}$
denote the KL projection of
$\hat{\mathcal P}$ defined in~\eqref{eq:klop},
where $\bar\theta_{\mathcal B}\in\mathbb R^d$
is its natural parameter vector.
The optimality condition is given by
$
\nabla\psi_{\mathcal B}(\bar\theta_{\mathcal B})
=\hat\eta_{\mathcal B}.
$
We compute the reference solution
$\bar\theta_{\mathcal B}$ by solving the
moment-matching equation using a high-precision
Newton method, with the linear systems solved
by iterative refinement. For all tested $\alpha$-connections, we use
the same initial point
$
\theta_{\mathcal B,0}
=
\bar\theta_{\mathcal B}+0.2v,
 \|v\|_2=1,
$
where $v\in\mathbb R^{136}$ is obtained by
normalizing a vector with independent
standard Gaussian entries.
The random seed is fixed at $20261110$
for reproducibility.
We consider
$\alpha\in\{1,0.5,1/3,0,-0.5,-1\}$
and apply the corresponding second-order
$\alpha$-retractions with unit step sizes.
To measure the local convergence order,
we define
$e_{k,\alpha}
=\|\theta_{\mathcal B,k}^{(\alpha)}
-\bar\theta_{\mathcal B}\|_2$
and compute
\begin{equation}
q_{k,\alpha}
=
\frac{\log(e_{k+1,\alpha}/e_{k,\alpha})}
{\log(e_{k,\alpha}/e_{k-1,\alpha})}.
\end{equation}
Five Newton updates are performed to obtain
$q_{1,\alpha},\ldots,q_{4,\alpha}$.

To accurately evaluate the rapidly decreasing
errors near the minimizer, we employ
arbitrary-precision arithmetic.
The reference KL projection is computed
with 380 decimal digits of precision.
All derivative evaluations and Newton
linear systems are computed to the
corresponding precision using
high-precision iterative refinement.

In addition to the empirical convergence order,
we evaluate the Riemannian gradient norm
to illustrate the convergence behavior.
For a differentiable function $f$ on a Riemannian manifold
$(\mathcal M,g)$, the Riemannian gradient norm
at $p\in\mathcal M$ is defined as
\begin{equation}
\|\operatorname{grad}f(p)\|_g
=
\sqrt{
\left\langle
\operatorname{grad}f(p),
\operatorname{grad}f(p)
\right\rangle
}.
\end{equation}
This norm is coordinate-invariant and provides
a consistent measure of convergence
independent of the chosen coordinate system.

\begin{figure}[t]
\centering
\begin{subfigure}[t]{0.45\linewidth}
    \centering
    \includegraphics[width=\linewidth]{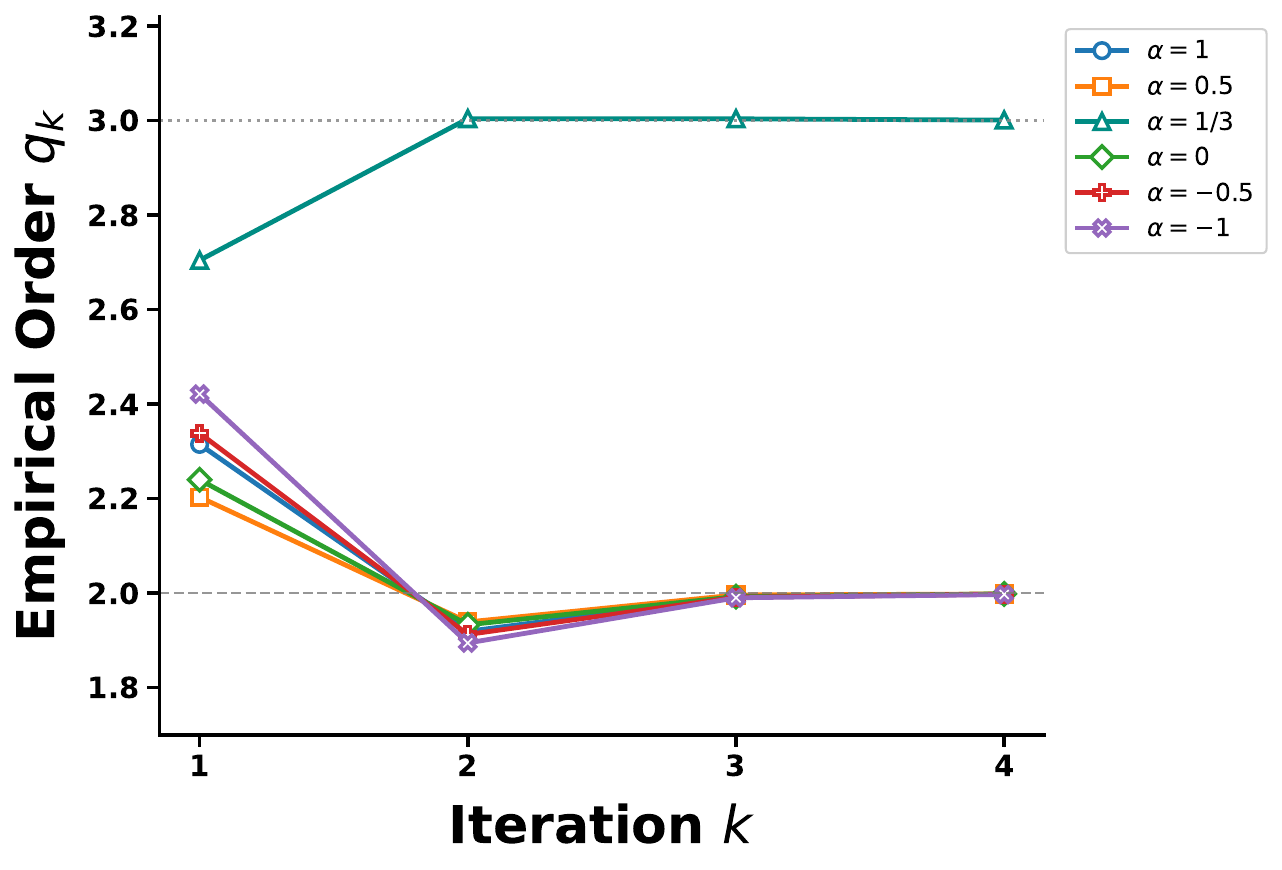}
    \caption{Empirical convergence order $q_{k,\alpha}$.}
    \label{fig:n16_qk}
\end{subfigure}
\hfill
\begin{subfigure}[t]{0.45\linewidth}
    \centering
    \includegraphics[width=\linewidth]{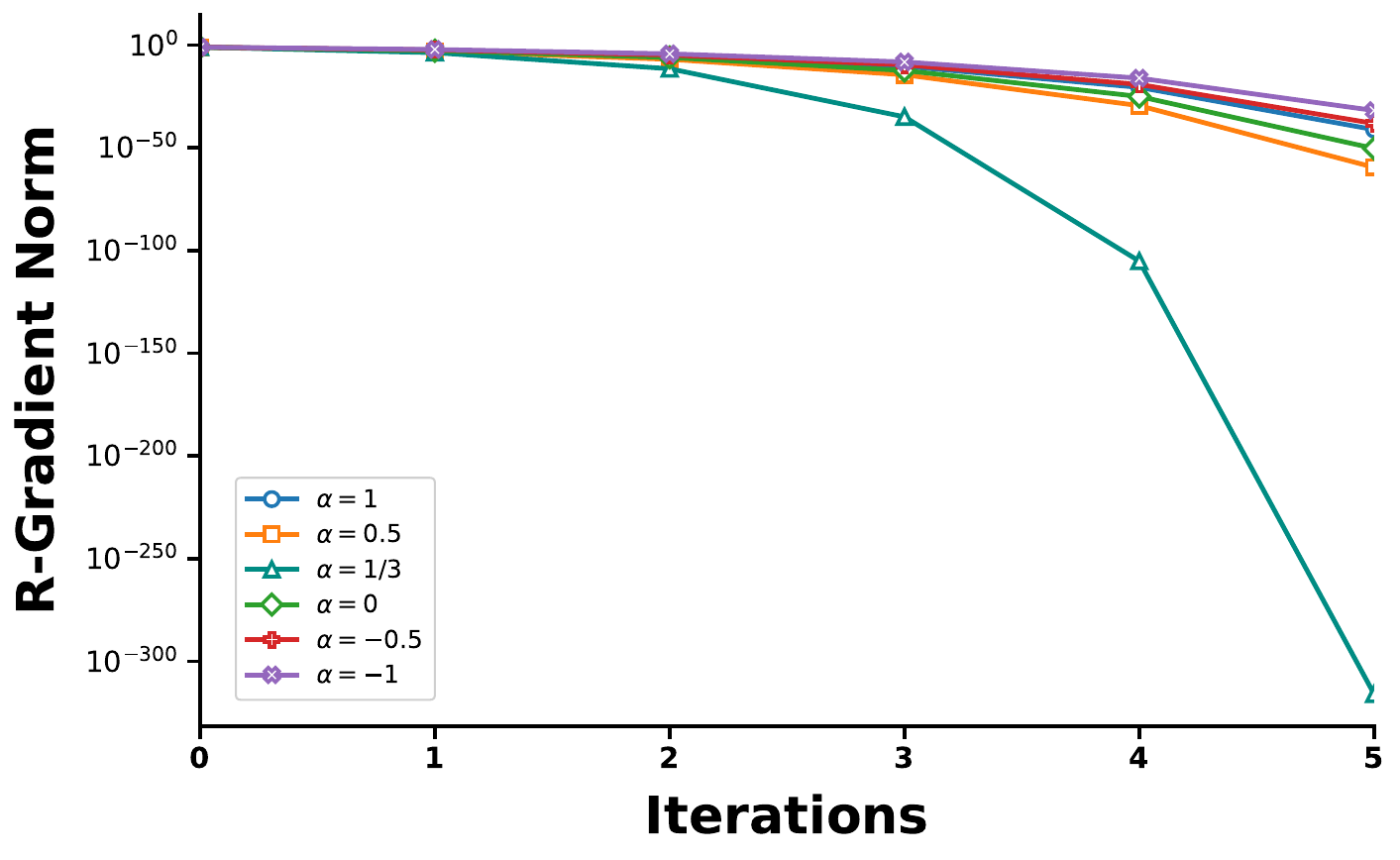}
    \caption{Riemannian gradient norm.}
    \label{fig:n16_rgrad}
\end{subfigure}
\caption{Local convergence behavior of the dual Riemannian Newton
methods for the $n=16$ Boltzmann model.
(a) Empirical convergence orders, where the dashed and dotted
horizontal lines indicate quadratic and cubic convergence,
respectively.
(b) Riemannian gradient norms versus iteration number.}
\label{fig:n16_convergence}
\end{figure}

Figures~\ref{fig:n16_qk} and~\ref{fig:n16_rgrad}
present the empirical convergence orders and
Riemannian gradient norms, respectively.
Figure~\ref{fig:n16_qk} shows that the empirical
order for $\alpha=1/3$ approaches $3$, with
$q_2=3.0040$, $q_3=3.0038$, and $q_4=3.0010$,
whereas the other tested connections exhibit
orders approaching $2$.
Figure~\ref{fig:n16_rgrad} further illustrates
the rapid decrease of the Riemannian gradient norm,
particularly for $\alpha=1/3$ in the later iterations.
These results provide numerical evidence of
cubic local convergence for $\alpha=1/3$
and quadratic local convergence for the
other tested connections.

Moreover, established in Theorem~\ref{natural_equal_newton}, 
the natural gradient method is equivalent to the dual Riemannian Newton method 
in the optimization problem~\eqref{eq:klop}. 
However, when a regularization term is added to the objective:
\begin{align*}
    \bar{\mathcal{P}} 
    &= \mathop{\mathrm{argmin}}_{\mathcal{P} \in \mathcal{S}_{\mathcal{B}}} 
    \Bigg(
    D_{\mathrm{KL}}\bigl(\hat{\mathcal{P}}\|\mathcal{P}\bigr) 
    + \lambda_1 \sum_i \bigl(\theta^{i}_{\mathcal{P}}\bigr)^2 
    + \lambda_2 \sum_{i<j} \bigl(\theta^{ij}_{\mathcal{P}}\bigr)^2
    \Bigg),
\end{align*}
with tunable hyperparameters \( \lambda_1 \) and \( \lambda_2 \) controlling 
the strength of regularization, the equivalence no longer holds. 
In this case, the natural gradient and dual Riemannian Newton methods 
follow different update dynamics due to the presence of the regularization term.

In this setting, the dual Riemannian Newton method 
with different $\alpha$ values (corresponding to different retractions $R^{(\alpha)}$) is expected to show faster convergence than the natural gradient method in both the $\theta$- and $\eta$-coordinate systems (the latter being equivalent to the mirror descent) with line search satisfying the Wolfe conditions~\cite{Boumal2023}, as well as the 
Euclidean optimizer \emph{Adam}, in the vicinity of a local optimum.
For Adam, we use a fixed learning rate of $0.01$,
moment decay coefficients $(\beta_1,\beta_2)=(0.9,0.999)$,
and a numerical stability constant $\epsilon=10^{-8}$.
The learning rate remains unchanged throughout the optimization, and no mini-batch sampling is employed.

To construct the target distribution \(\hat{\mathcal{P}}\), 
we first consider the exponential family \(\mathcal{S}_{\Omega}\) 
and assign random interaction parameters \(\theta^A\) for each interaction subset \( A \) with order \( r = |A| \). 
Specifically, the parameter \(\theta^A\) is sampled uniformly from the interval
$
\theta^A \sim \mathcal{U}\biggl[-\frac{\texttt{base\_scale}}{r},\, \frac{\texttt{base\_scale}}{r}\biggr],
$
where the base scaling factor is set to \(\texttt{base\_scale} = 1.0\). 
This construction yields larger parameter magnitudes for lower-order interactions 
and smaller ones for higher-order interactions, ensuring numerical stability. 
The distribution \(\hat{\mathcal{P}}\) is then obtained by evaluating 
and normalizing the unnormalized log-probabilities over all \(2^n\) binary states, 
resulting in an exact ground-truth distribution without sampling noise. 
This synthetic distribution is used to verify the efficiency of the proposed dual Riemannian Newton method.

We conducted numerical experiments under two configurations:  
(i) $(\lambda_1, \lambda_2, n) = (0.3, 0.8, 8)$ with initial $\theta$ values sampled from $\mathrm{Uniform}(-0.1, 0.0)$, and  
(ii) $(\lambda_1, \lambda_2, n) = (0.5, 0.5, 4)$ with initial $\theta$ values sampled from $\mathrm{Uniform}(-0.25, 0.2)$.  

The corresponding results are illustrated in Figure~\ref{fig:regular_term}. 
Moreover, Table~\ref{tab:runtime_n4_n8} reports the iteration counts,
mean wall-clock runtimes, average per-iteration times, and runtime
standard deviations for all methods. Each method was evaluated over 20 repeated
runs following two warm-up runs, using the same initialization.

In both experiments, the dual Riemannian Newton
(DN) methods require substantially fewer iterations
and achieve lower mean total runtimes than natural
gradient, mirror descent, and Adam.
Among the tested connections, $\alpha=1$ achieves
the lowest mean runtime in both configurations.
Moreover, the Levi--Civita connection ($\alpha=0$)
does not yield the fastest convergence in either
experiment, illustrating the influence of the
connection choice on numerical performance.
As shown in Theorem~\ref{euc_hession},
the case $\alpha=1$ reduces to the standard
Euclidean Newton method in the
$\theta$-affine coordinates, avoiding the
additional connection corrections.

Theorem~\ref{convergence} establishes that,
under the stated assumptions, the iterates remain
in a sufficiently small neighborhood of the
local minimizer and converge at least quadratically.
In particular, the proof yields
\begin{equation}
\|\xi^{p_{k+1}}-\xi^{p^\star}\|
\leq C_\alpha
\|\xi^{p_k}-\xi^{p^\star}\|^2,
\end{equation}
where
\begin{equation}
C_\alpha
=
2\kappa\left(
\gamma_{M,\alpha}+\frac{L_H}{2}
\right)+\gamma_{R,\alpha}.
\end{equation}
Here, $\kappa$ bounds the inverse Newton matrix
at the minimizer, $L_H$ is the local Lipschitz
constant of the coordinate Hessian,
$\gamma_{M,\alpha}$ controls the local variation
of the Newton matrix, and $\gamma_{R,\alpha}$
bounds the retraction error. The dependence of these constants on the chosen
connection provides a theoretical explanation
for why different $\alpha$ values may exhibit
different local error bounds.

\begin{figure}[t]
  \centering
  \includegraphics[width=0.40\textwidth]{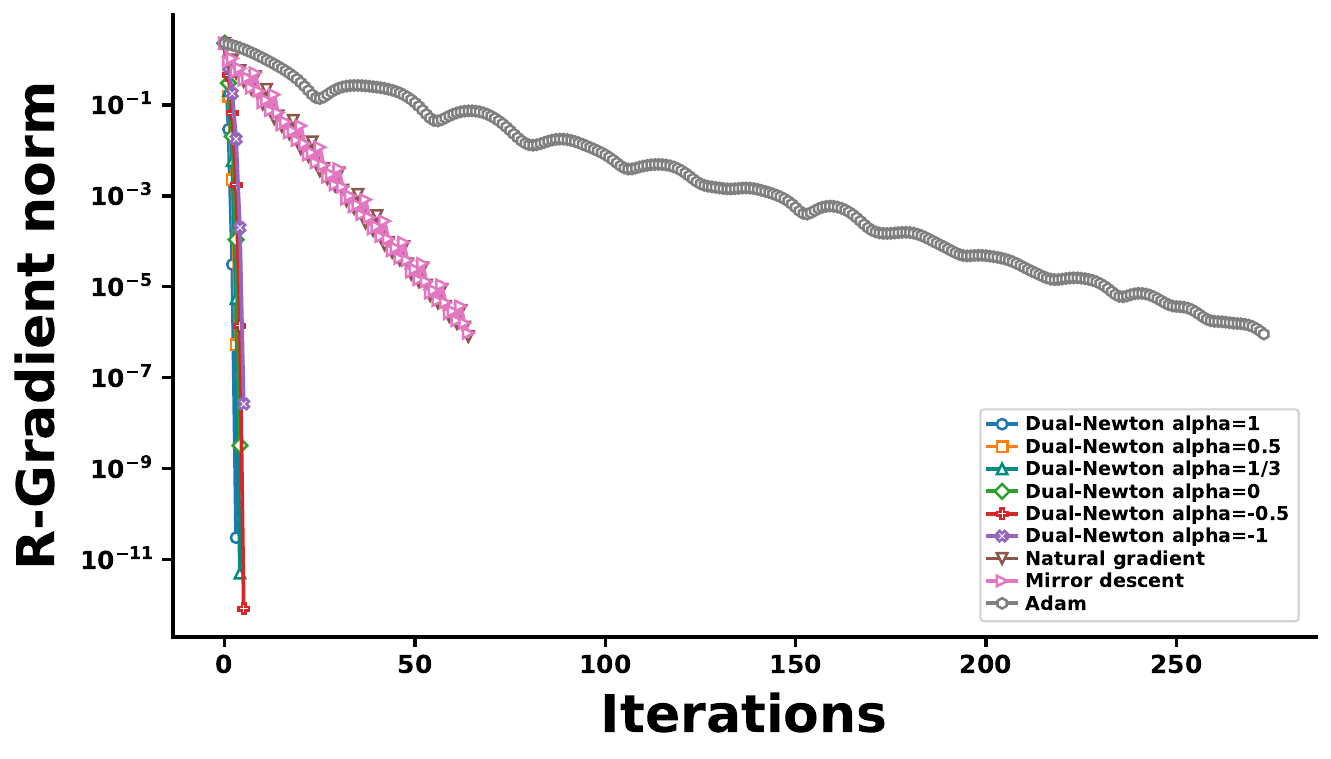}
  \includegraphics[width=0.40\textwidth]{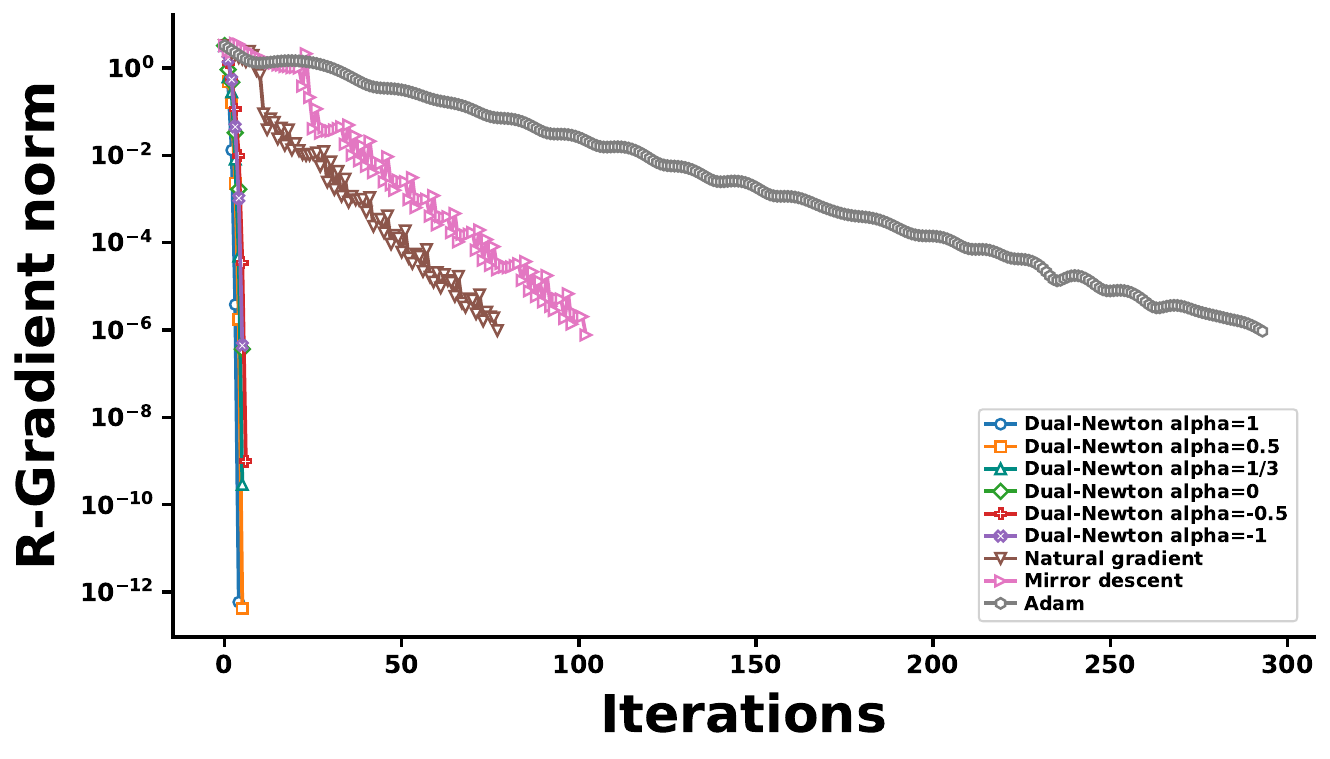}
  \includegraphics[width=0.40\textwidth]{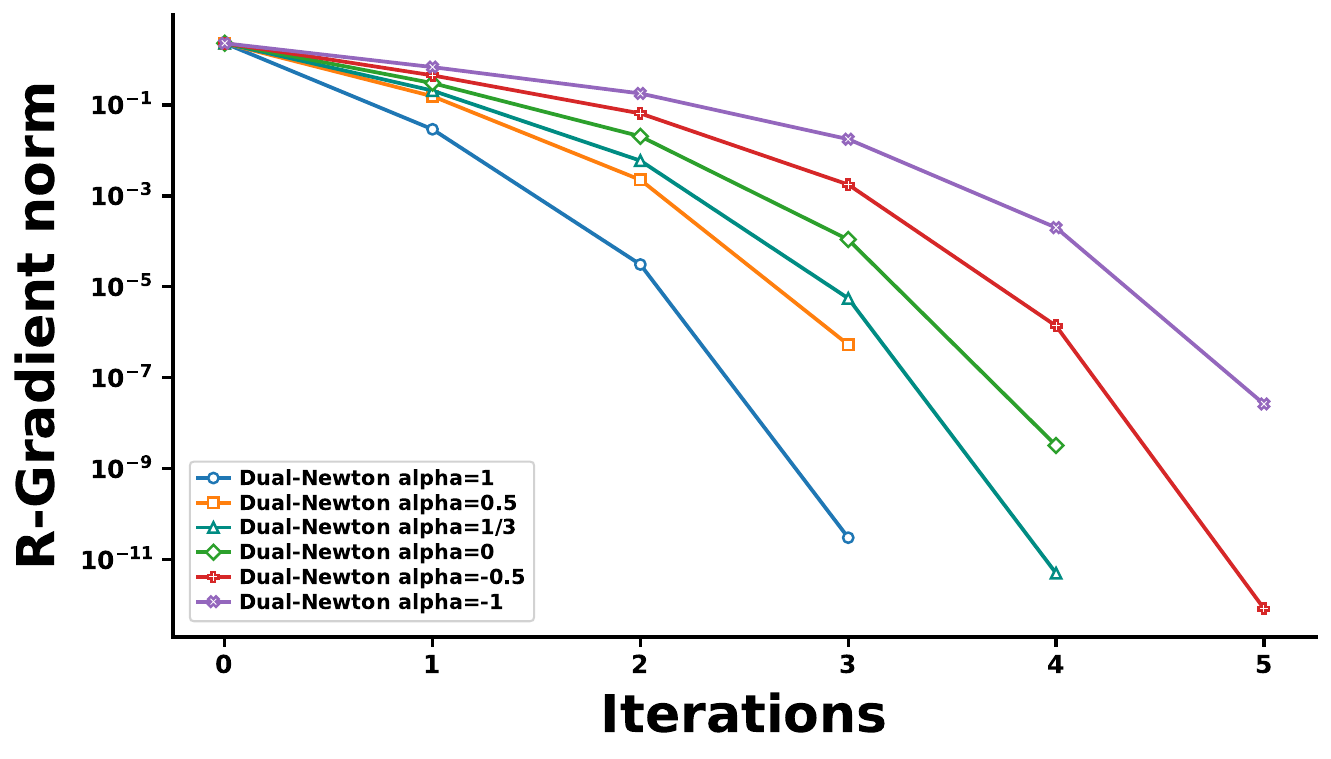}
  \includegraphics[width=0.40\textwidth]{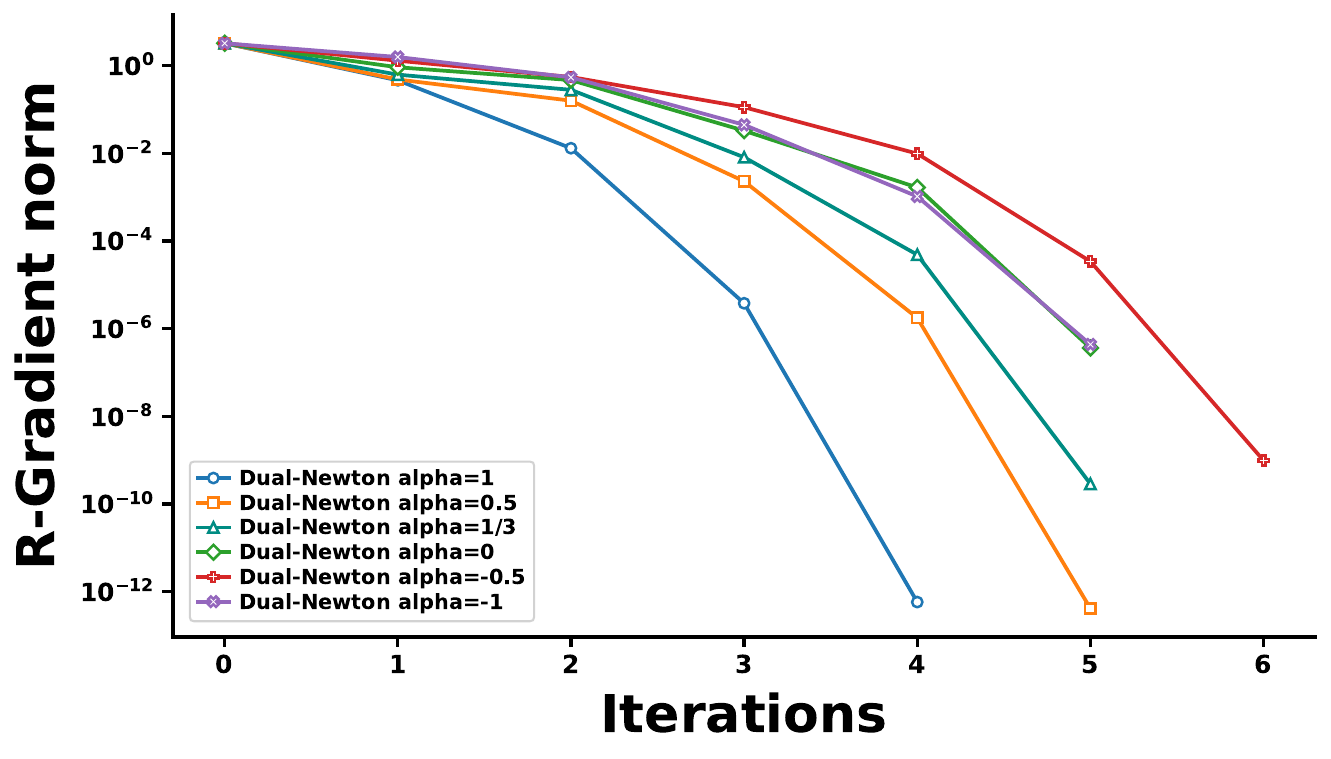}
  \caption{Number of iterations for each method. (Left): \((\lambda_1, \lambda_2, n) = (0.5, 0.5, 4)\). (Right): \((\lambda_1, \lambda_2, n) = (0.3, 0.8, 8)\). The bottom row enlarges the first five or six steps of the dual Riemannian Newton method for clarity.}
  \label{fig:regular_term}
\end{figure}



\begin{table}[t]
\centering
\caption{Iteration counts and runtimes in Experiment 1.
All methods terminate when the Riemannian gradient norm
is below $10^{-6}$. Total runtime is reported as
mean $\pm$ standard deviation over 20 runs after
two warm-up runs. All times are in milliseconds (ms).
DN denotes dual Riemannian Newton.
The lowest mean total runtime is underlined.}

\label{tab:runtime_n4_n8}
\footnotesize
\setlength{\tabcolsep}{3pt}
\renewcommand{\arraystretch}{0.85}
\begin{tabular}{@{}lcccccc@{}}
\toprule
& \multicolumn{3}{c}{$n=4$}
& \multicolumn{3}{c}{$n=8$} \\
\cmidrule(lr){2-4}
\cmidrule(lr){5-7}
Method & Iters & Total (ms) & Avg. (ms)
       & Iters & Total (ms) & Avg. (ms) \\
\midrule
DN $\alpha=1$
& 3 & \underline{$0.1691\pm0.0036$} & 0.0564
& 4 & \underline{$0.4077\pm0.0050$} & 0.1019 \\
DN $\alpha=0.5$
& 3 & $0.2125\pm0.0044$ & 0.0708
& 5 & $0.7123\pm0.0086$ & 0.1425 \\
DN $\alpha=1/3$
& 4 & $0.2824\pm0.0042$ & 0.0706
& 5 & $0.7057\pm0.0043$ & 0.1411 \\
DN $\alpha=0$
& 4 & $0.2807\pm0.0045$ & 0.0702
& 5 & $0.7052\pm0.0031$ & 0.1410 \\
DN $\alpha=-0.5$
& 5 & $0.3527\pm0.0063$ & 0.0705
& 6 & $1.1811\pm0.8829$ & 0.1968 \\
DN $\alpha=-1$
& 5 & $0.3470\pm0.0035$ & 0.0694
& 5 & $0.8132\pm0.1449$ & 0.1626 \\
\midrule
Natural Gradient
& 64 & $15.0811\pm0.6707$ & 0.2356
& 77 & $39.8771\pm12.2492$ & 0.5179 \\
Mirror Descent
& 64 & $45.6264\pm16.7333$ & 0.7129
& 102 & $249.0812\pm57.0318$ & 2.4420 \\
Adam
& 273 & $20.6821\pm9.5620$ & 0.0758
& 293 & $43.1395\pm14.6146$ & 0.1472 \\
\bottomrule
\end{tabular}
\end{table}

\subsection{Experiment 2: $\alpha$-divergence Minimization}
We next consider the projection problem based on the $\alpha$-divergence, 
\[
f = D_{\bar{\alpha}}(p\|q)
= \frac{4}{1-\bar{\alpha}^2}\left(
  1 - \int_{\mathcal{X}}
      p(x)^{\frac{1+\bar{\alpha}}{2}}
      q(x)^{\frac{1-\bar{\alpha}}{2}} \,\mathrm{d}x
\right).
\]
The use of $\alpha$-divergence other than the Kullback--Leibler divergence is advantageous because it allows flexible control over the projection behavior~\cite{HernandezLobato2016,Daudel2023,Ghalamkari2024} and the robustness~\cite{George2021} against the model misspecification.
In this experiment, $\bar{\alpha}$ is fixed at 3, and distributions $p$ and $q$ are given as
\[
p(x)
= \mathcal{N}\!\left(
x \,\middle|\,
\begin{bmatrix}\mu_1\\ \mu_2\end{bmatrix},
\begin{bmatrix}\sigma_1^{2} & 0\\[2pt] 0 & \sigma_2^{2}\end{bmatrix}
\right),
\qquad
q(x)
= \mathcal{N}\!\left(
x \,\middle|\,
\begin{bmatrix}\mu\\ \mu\end{bmatrix},
\begin{bmatrix}\sigma^{2} & 0\\[2pt] 0 & \sigma^{2}\end{bmatrix}
\right).
\]
We set $\mu_1=2$, $\mu_2=1.5$, $\sigma_1=1.3$, and $\sigma_2=0.7$ and estimate the parameters $\mu, \sigma$ for $q$ by minimizing the $\alpha$-divergence.
Here we introduce the $\alpha$-connection in $q(x)$.
The Fisher information matrix for this model is given by:
$
\mathbf{G}(\sigma) = 
\begin{pmatrix}
\frac{2}{\sigma^2} & 0 \\
0 & \frac{4}{\sigma^2}
\end{pmatrix}.
$
The corresponding Christoffel symbols of the second kind associated with the $\alpha$-connection are:
\[
\Gamma^{1(\alpha)}_{ij}
=-\frac{1+\alpha}{\sigma}
\begin{pmatrix}
0 & 1\\
1 & 0
\end{pmatrix},
\qquad
\Gamma^{2(\alpha)}_{ij}
=\frac{1}{\sigma}
\begin{pmatrix}
\frac{1-\alpha}{2} & 0\\
0 & -(1+2\alpha)
\end{pmatrix}.
\]
It is straightforward to verify that the manifold is not flat when \( \alpha \neq \pm1 \). Moreover, the $\alpha$-divergence admits the following closed-form expression~\cite{Ghosh2008}:
\[
\frac{4}{1-\bar{\alpha}^2} \left\{
1 -
\left(\sigma_1^2 \sigma_2^2\right)^{\frac{1-\bar{\alpha}}{4}}
\left(\sigma^4\right)^{\frac{1+\bar{\alpha}}{4}}
\left[
\left( \frac{1+\bar{\alpha}}{2} \sigma^2 + \frac{1-\bar{\alpha}}{2} \sigma_1^2 \right)
\left( \frac{1+\bar{\alpha}}{2} \sigma^2 + \frac{1-\bar{\alpha}}{2} \sigma_2^2 \right)
\right]^{-\frac12}
\right.
\]
\[
\left.
\quad \cdot \exp\left[
-\frac{1-\bar{\alpha}^2}{8}
\left(
\frac{(\mu_1 - \mu)^2}{
\frac{1+\bar{\alpha}}{2} \sigma^2 + \frac{1-\bar{\alpha}}{2} \sigma_1^2 }
+
\frac{(\mu_2 - \mu)^2}{
\frac{1+\bar{\alpha}}{2} \sigma^2 + \frac{1-\bar{\alpha}}{2} \sigma_2^2 }
\right)
\right]
\right\}.
\]

Note that the $\alpha$-divergence is finite only if
$
\frac{1+\bar{\alpha}}{2}\sigma^2
+
\frac{1-\bar{\alpha}}{2}\sigma_i^2>0, i=1,2.
$
For $\bar{\alpha}=3$, this condition reduces to
$
2\sigma^2>\max\{\sigma_1^2,\sigma_2^2\}.
$
Thus, with $\sigma_1=1.3$ and $\sigma_2=0.7$,
the domain of the objective function is
$
\mathcal D=
\left\{
(\mu,\sigma)\in\mathbb R\times(0,\infty):
\sigma>\frac{1.3}{\sqrt{2}}
\right\}.
$

In detail, in this experiment, we first set the divergence parameter to $\bar{\alpha} = 3$ in the $\alpha$-divergence, and choose $\alpha = -0.4$, $-0.2$, $0.0$, $0.2$, or $0.4$ to define the corresponding $\alpha$-connections.
The performance of the proposed dual Riemannian Newton method is then compared with that of the natural gradient and Adam methods.
For Adam, we use a fixed learning rate of $0.2$, with
$\beta_1=0.9$, $\beta_2=0.999$, and $\epsilon=10^{-8}$,
without a learning-rate schedule or mini-batch sampling.
 All methods are implemented in the \((\mu, \sigma)\) coordinate system, and the natural gradient method adopts a line search procedure satisfying the Wolfe conditions.

The corresponding convergence results, measured by
the Riemannian gradient norm, are illustrated in
Figure~\ref{fig:alpha_div}.
The optimization trajectories associated with
different retractions $R^{(\alpha)}$ are shown in
Figure~\ref{fig:path}.
As summarized in Table~\ref{tab:runtime_gaussion},
the proposed dual Riemannian Newton methods reach
the stopping criterion
$\|\operatorname{grad}f\|_g<10^{-6}$
within 9--12 iterations, whereas natural gradient
descent and Adam require 84 and 254 iterations,
respectively.
Based on 20 measured runs after two warm-up runs,
the mean total runtime of the dual Riemannian Newton
methods ranges from 1.3961 to 2.4043 ms, compared
with 24.1193 ms for natural gradient descent and
36.3106 ms for Adam.
This corresponds to approximately 10--17 times
and 15--26 times lower mean total runtimes,
respectively.

Furthermore, Figure~\ref{fig:path} shows that
the dual Riemannian Newton methods generally follow
more direct trajectories toward the minimizer
in the $(\mu,\sigma)$ parameter space.
All dual Riemannian Newton methods use unit step sizes,
and we verify that every generated iterate satisfies
the admissibility condition
$\sigma>1.3/\sqrt{2}$. In particular, the minimum recorded value of
$\sigma$ across all methods is $0.9985$,
which exceeds $1.3/\sqrt{2}\approx0.9192$.


\begin{figure}[t]
  \centering
  \includegraphics[width=0.40\textwidth]{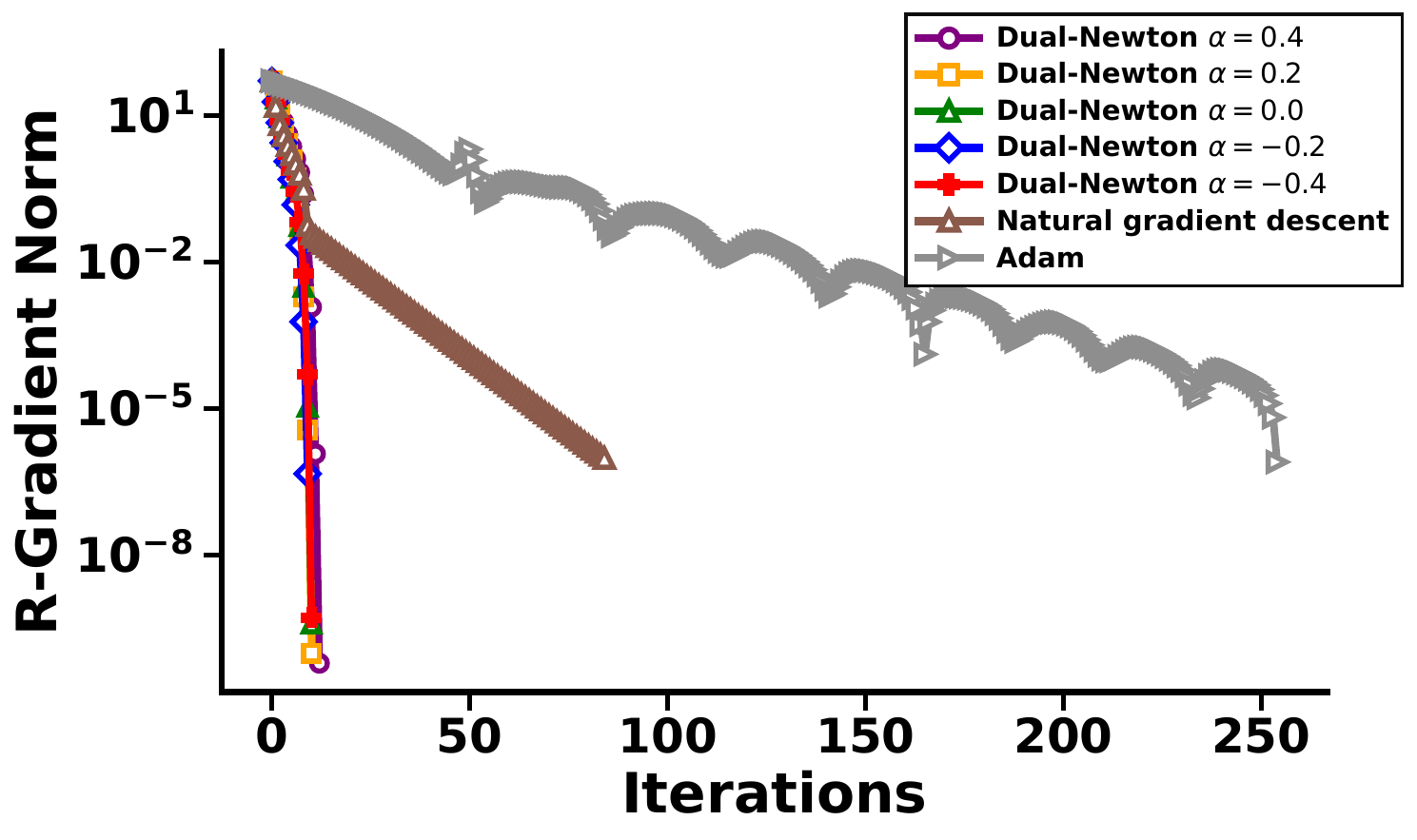}
  \includegraphics[width=0.40\textwidth]{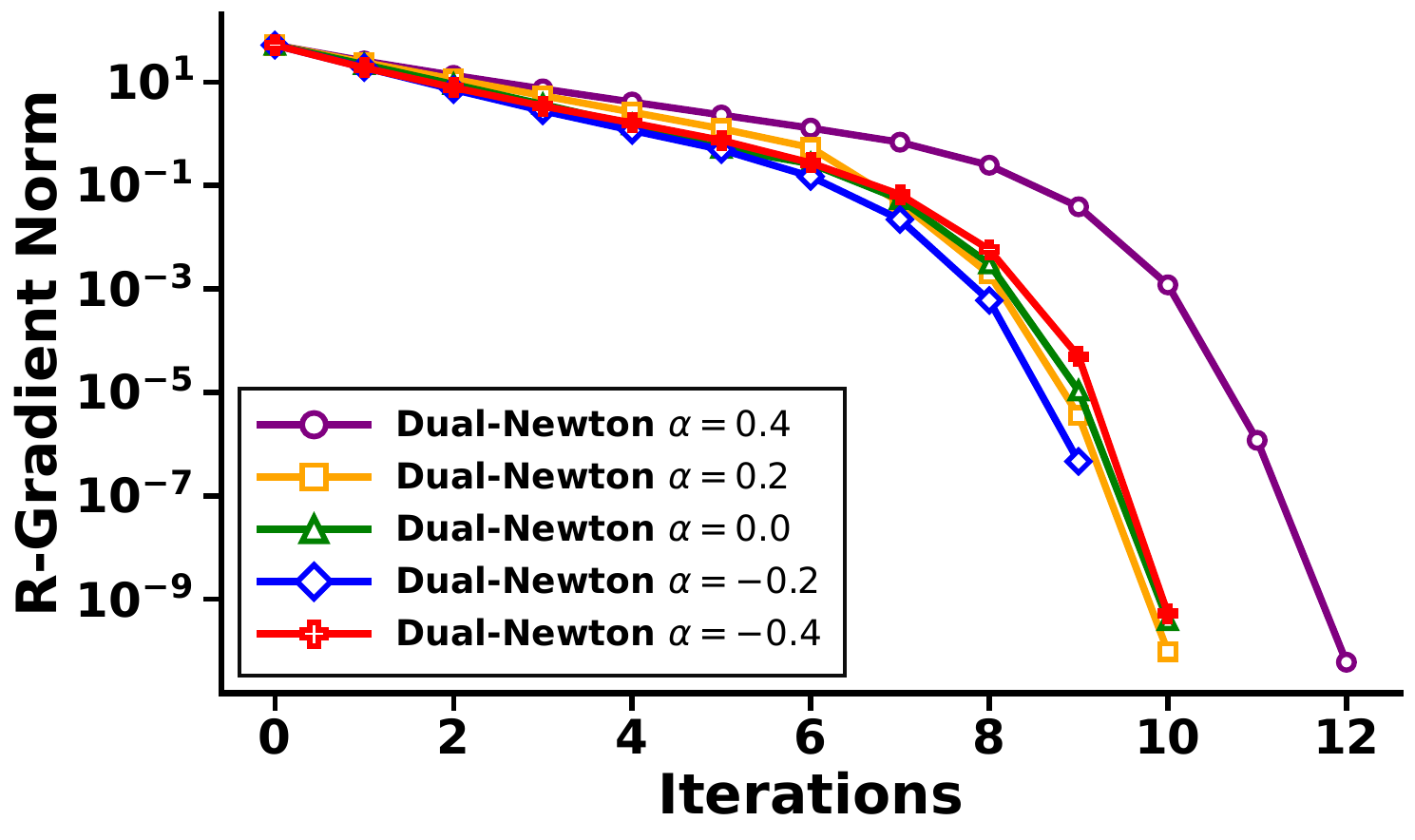}
  \caption{(Left): The evolution of the Riemannian gradient norm versus iteration for each method. (Right): A magnified view of the first 12 steps of the dual Riemannian Newton method from the left panel.}
  \label{fig:alpha_div}
\end{figure}


\begin{figure}[t]
  \centering

  \includegraphics[width=0.55\linewidth]{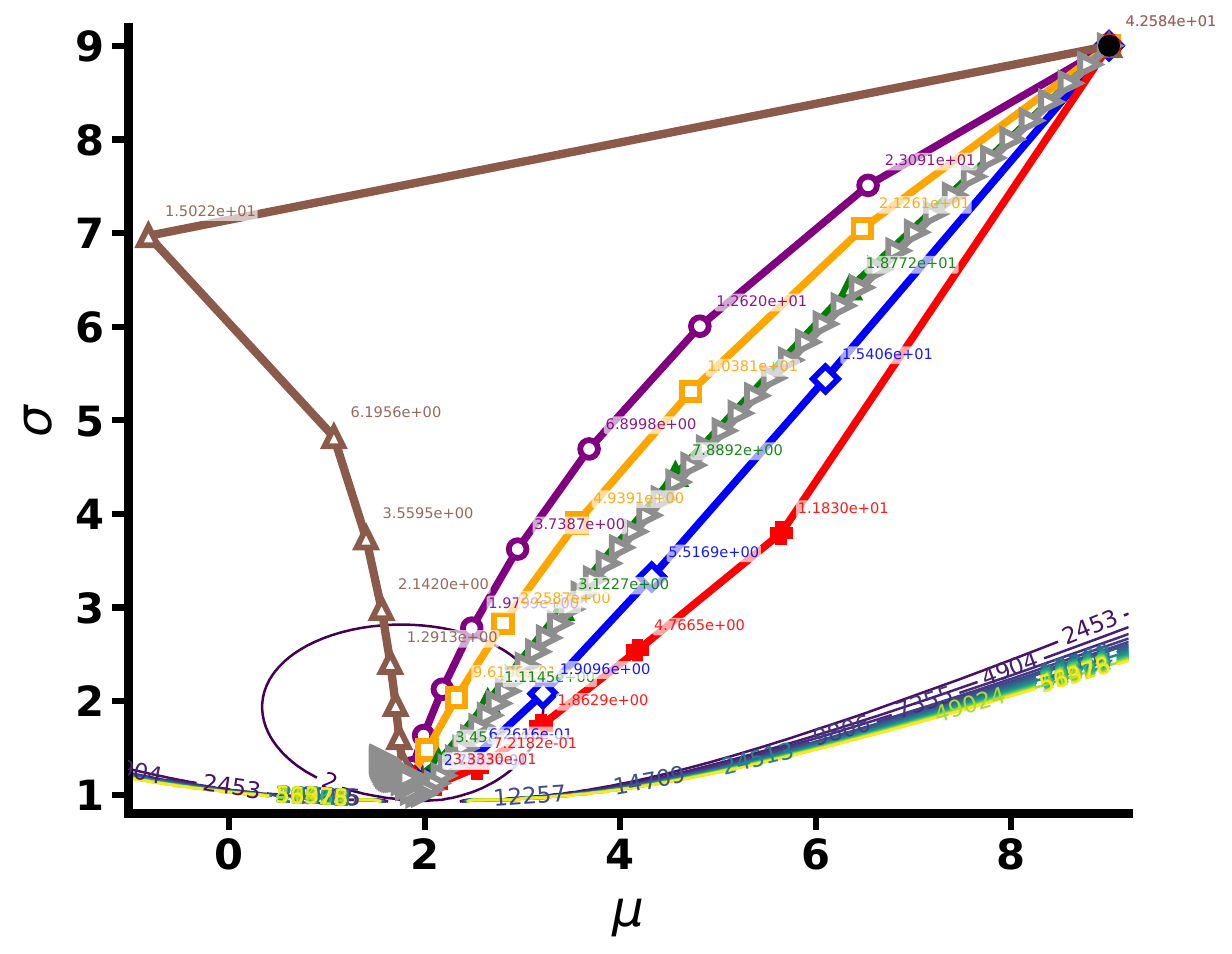}

  \vspace{4pt} 
  \includegraphics[width=0.80\linewidth]{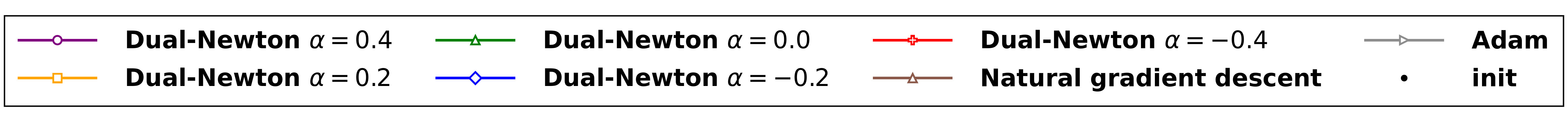}

  \caption{Optimization paths for each method. The numbers over steps denote the
  $\alpha$-divergence at that iteration (first five iterations shown). The bottom panel
  summarizes the legend for all methods.}
  \label{fig:path}
\end{figure}

\begin{table}[t]
\centering
\caption{Iteration counts and runtime statistics
for Experiment 2, measured over 20 runs after
two warm-up runs (mean $\pm$ sample standard deviation).
All methods stop when
$\|\operatorname{grad}f\|_g<10^{-6}$.
The lowest mean total runtime is underlined.}
\label{tab:runtime_gaussion}
\scriptsize
\setlength{\tabcolsep}{3pt}
\renewcommand{\arraystretch}{1.2}
\begin{tabular}{@{}lccc@{}}
\toprule
Method & Iterations & Total Time (ms) & Avg./Iter. (ms) \\
\midrule
Dual-Newton $\alpha=0.4$
& 12 & $2.4043 \pm 3.3678$ & $0.20036 \pm 0.28065$ \\
Dual-Newton $\alpha=0.2$
& 10 & $1.9282 \pm 3.7564$ & $0.19282 \pm 0.37564$ \\
Dual-Newton $\alpha=0$
& 10 & $2.0162 \pm 2.6258$ & $0.20162 \pm 0.26258$ \\
Dual-Newton $\alpha=-0.2$
& 9 & $\underline{1.3961} \pm 1.8867$
& $0.15512 \pm 0.20963$ \\
Dual-Newton $\alpha=-0.4$
& 10 & $1.8783 \pm 1.8711$ & $0.18783 \pm 0.18711$ \\
Natural Gradient
& 84 & $24.1193 \pm 8.7957$ & $0.28713 \pm 0.10471$ \\
Adam
& 254 & $36.3106 \pm 12.3364$ & $0.14296 \pm 0.04857$ \\
\bottomrule
\end{tabular}
\end{table}

\subsection{Experiment 3: Beta Mixture Model}
Finally, to demonstrate the efficiency of the proposed dual Riemannian Newton method for  \emph{arbitrary} probability distributions beyond the exponential family, 
we apply it in this experiment to the maximum likelihood estimation of a fixed-weight \emph{Beta mixture model}. 
The optimization problem is formulated as follows:
\begin{equation}
    \min_{\{\alpha_k,\, \beta_k\}_{k=1}^K} \!\mathcal{L}(\alpha, \beta),
    \quad
    \mathcal{L}(\alpha, \beta)
    = 
    - \sum_{i=1}^{N} 
    \ln \!\left[
        \sum_{k=1}^{K} 
        w_k \,
        \mathrm{Beta}\!\left(
            x_{i,1} \mid \alpha_k, \beta_k
        \right)
        \mathrm{Beta}\!\left(
            x_{i,2} \mid \alpha_k, \beta_k
        \right)
    \right],
    \label{eq:mle_beta_mixture}
\end{equation}
where $K$ denotes the number of mixture components, 
$\alpha_k$ and $\beta_k$ are the shape parameters of the $k$-th Beta component, 
and $w_k$ are fixed mixture weights satisfying $\sum_{k=1}^{K} w_k = 1$.
The objective $\mathcal{L}(\alpha, \beta)$ represents the negative log-likelihood of the observed data under the Beta mixture model.
The dataset of $N=5000$ samples is generated from a known mixture with parameters 
$w = [0.35,\, 0.4,\, 0.25]$, 
$\alpha_{\mathrm{true}} = [2.0,\; 3.0,\; 5.0]$, and 
$\beta_{\mathrm{true}} = [5.0,\; 2.0,\; 3.5]$.
All methods are initialized from the same starting point,
with $\boldsymbol{\alpha}_0=(1.8,2.5,4.0)$ and
$\boldsymbol{\beta}_0=(4.5,1.8,3.0)$.
All six Beta shape parameters were numerically verified
to remain strictly positive at every iteration.

Let $\xi=(\alpha_1,\ldots,\alpha_K,
\beta_1,\ldots,\beta_K)^\top$
denote the Beta shape parameters,
and let $\ell(\xi)=\ln p(x;\xi)$
be the log-density of a single observation
$x\in(0,1)^2$.
Then the objective function in
\eqref{eq:mle_beta_mixture} can be written as
$
f(\xi)=-\sum_{i=1}^N\ln p(x_i;\xi),
$
where $x_i=(x_{i,1},x_{i,2})^\top$
is the $i$-th observation.

For the Beta-mixture model in
\eqref{eq:mle_beta_mixture}, the Fisher
information metric is computed as
$
g_{ij}(\xi)=\mathbb E[\partial_i\ell(\xi)\partial_j\ell(\xi)],
$
as defined in Example II.1. Rather than explicitly constructing the full
$\alpha$-connection tensor, we directly compute
the two contractions in Equation \eqref{eq:connection_contractions}.
The first contraction is
\[
\sum_{k=1}^n\Gamma_{ij}^{(\alpha)k}(\xi)\partial_k f(\xi)
=\mathbb E\left[\left(\partial_{ij}\ell(\xi)+\frac{1-\alpha}{2}\partial_i\ell(\xi)\partial_j\ell(\xi)\right)\nabla\ell(\xi)^\top G(\xi)^{-1}\nabla f(\xi)\right].
\]

Similarly, the second contraction is
\[
\sum_{j,k=1}^n\Gamma_{jk}^{(\alpha)i}(\xi)\beta^j(\xi)\beta^k(\xi)
=\left[G(\xi)^{-1}\mathbb E\left[\nabla\ell(\xi)\left(\beta(\xi)^\top\nabla^2\ell(\xi)\beta(\xi)+\frac{1-\alpha}{2}\left(\nabla\ell(\xi)^\top\beta(\xi)\right)^2\right)\right]\right]_i.
\]

The expectations are approximated using
tensor-product Gauss--Legendre quadrature
with 24 nodes per coordinate over $(0,1)^2$,
without Monte Carlo sampling.

The first contraction is used to construct
the Newton coefficient matrix, while the
second is used to compute the second-order
retraction correction in Algorithm~\ref{alg:Computation dual riemannian newton}.
The Newton system is solved using a dense
direct linear solver. By evaluating these contractions directly,
we avoid constructing the full third-order
connection tensor and reduce the
connection-related quadrature cost from
$O(Qn^3)$ to $O(Qn^2)$, where $Q=24^2=576$
is the number of quadrature points and
$n=2K=6$ is the parameter dimension.

We compare the proposed dual Riemannian Newton method
(with $\alpha\in\{0,0.25,0.5,0.75,1\}$)
against the natural gradient method (with strong Wolfe
line search) and the Euclidean Adam optimizer.
For Adam, we use a fixed learning rate of $0.01$,
moment decay coefficients $(\beta_1,\beta_2)=(0.9,0.999)$,
and a numerical stability constant $\epsilon=10^{-8}$.
The learning rate remains unchanged throughout optimization,
and each gradient is computed using all $N=5000$
observations without mini-batch sampling.

The convergence speed is evaluated using the
Riemannian gradient norm at each iteration, as shown in Fig.~\ref{fig:beta_mix}. All methods are terminated when the Riemannian gradient norm falls below $10^{-8}$. Moreover, Table~\ref{tab:runtime_beta} reports
the iteration counts, total wall-clock runtimes,
and average per-iteration times, measured over
20 runs after two warm-ups
(mean $\pm$ sample standard deviation).
The dual Riemannian Newton methods reach the prescribed
tolerance within 5--7 iterations, compared with
39 iterations for the natural gradient method and
2178 iterations for Adam.
They also achieve lower total runtimes,
ranging from 6.22 to 8.22 ms,
compared with 16.15 ms for natural gradient
and 971.24 ms for Adam.

\begin{figure}[t]
  \centering
  \includegraphics[width=0.45\textwidth]{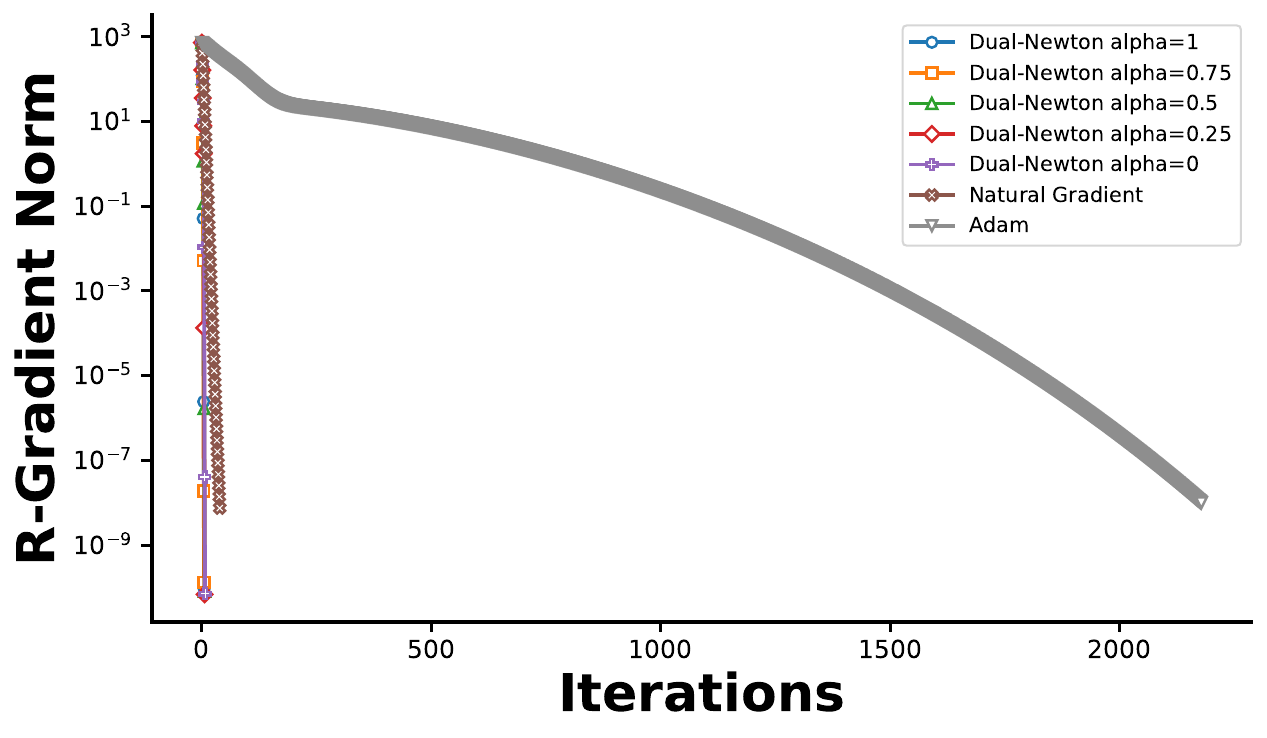}
  \includegraphics[width=0.45\textwidth]{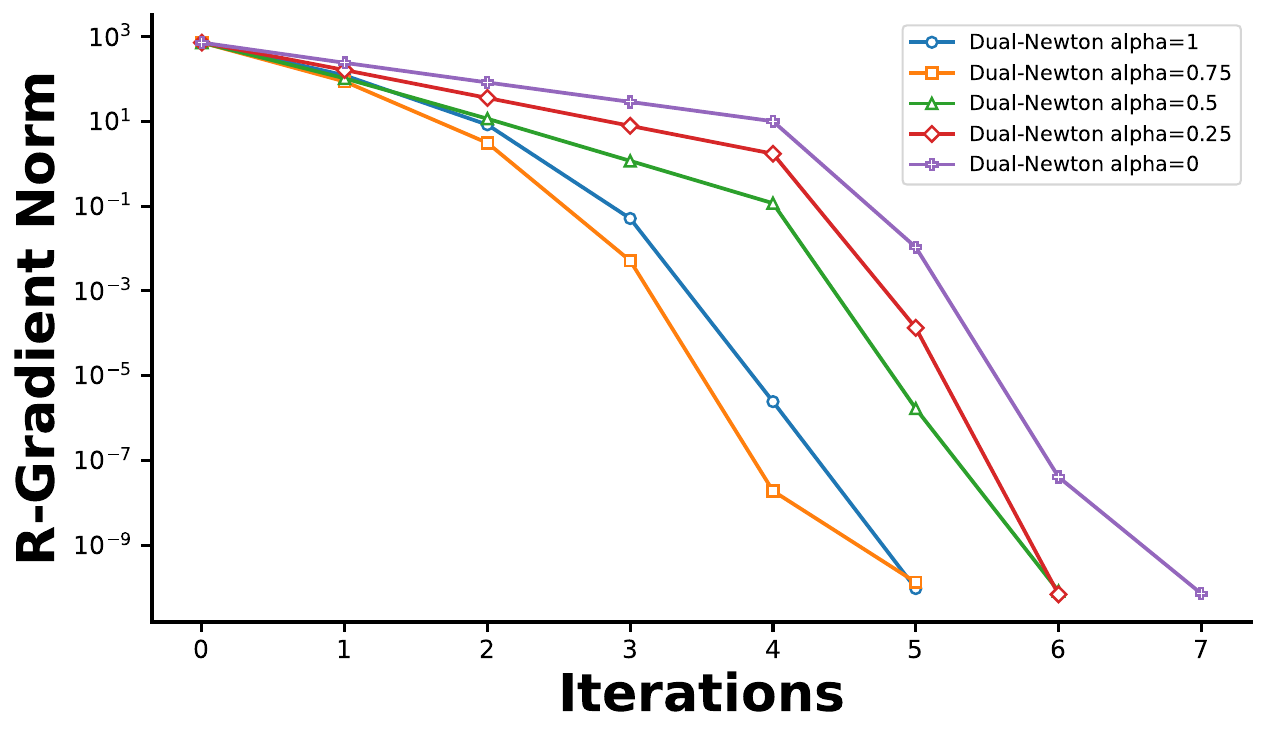}
  \caption{(Left): The evolution of the Riemannian gradient norm versus iteration for each method. (Right): A magnified view for first seven steps of the dual Riemannian Newton method from the left panel.}
  \label{fig:beta_mix}
\end{figure}

\begin{table}[t]
\centering
\caption{Iteration counts and runtime statistics for Experiment 3,
measured over 20 runs after two warm-ups
(mean $\pm$ sample standard deviation).
All methods stop when
$\|\operatorname{grad} f\|_g < 10^{-8}$.
The lowest mean total runtime is underlined.}
\label{tab:runtime_beta}

\resizebox{0.65\columnwidth}{!}{%
\begin{tabular}{lccc}
\toprule
Method
& Iterations
& Total Time (ms)
& Avg.\ per Iteration (ms) \\
\midrule
Dual-Newton $\alpha=1$
& 5
& $\underline{6.22}\pm0.10$
& $1.2440\pm0.0200$ \\

Dual-Newton $\alpha=0.75$
& 5
& $6.64\pm0.61$
& $1.3280\pm0.1220$ \\

Dual-Newton $\alpha=0.5$
& 6
& $7.16\pm0.11$
& $1.1933\pm0.0183$ \\

Dual-Newton $\alpha=0.25$
& 6
& $7.50\pm0.44$
& $1.2500\pm0.0733$ \\

Dual-Newton $\alpha=0$
& 7
& $8.22\pm0.09$
& $1.1743\pm0.0129$ \\

Natural Gradient
& 39
& $16.15\pm0.60$
& $0.4141\pm0.0154$ \\

Adam
& 2178
& $971.24\pm16.33$
& $0.4459\pm0.0075$ \\
\bottomrule
\end{tabular}%
}
\end{table}


\section{Discussion}
In this section, we clarify how the choice of
dual affine connections influences the descent
property of the Newton direction through its
relationship with geodesic strict convexity.
To this end, we consider $\alpha$-connections
and discuss the choice of $\alpha$.
The key fact is that the choice of $\alpha$
affects the positive definiteness of the
corresponding covariant Hessian, which
guarantees that the Newton direction is
well-defined and is a descent direction
at every nonstationary point.
To make this relationship precise, we first
introduce the definitions of descent direction
and geodesic strict convexity, and then
analyze the positive definiteness of the
covariant Hessian matrix
$\mathbf{M}^{\nabla^{(\alpha)}}(\xi)$.

\begin{definition}
Let $(\mathcal{M},g)$ be a Riemannian manifold and $f: \mathcal{M}\to\mathbb R$ be differentiable.
A nonzero tangent vector $X_p\in T_p\mathcal M$ is a descent direction for $f$ at $p\in \mathcal{M}$ if
\[
\big\langle \operatorname{grad} f(p),\, X_p \big\rangle \; < \; 0.
\]
Moreover, $X_p$ is a non-ascent direction if
$\big\langle \operatorname{grad} f(p),\, X_p \big\rangle \le 0$.
\end{definition}
Given a suitably chosen \( \alpha \), the Newton direction can be guaranteed to serve as a descent direction for a particular objective function \( f \) at each iteration. Nevertheless, such a choice of \( \alpha \) may not exist for any smooth functions \( f \). 
We choose the retraction $R^{(\alpha)}$ associated
with $\nabla^{(\alpha)}$ and define the Newton
direction by
$
X_p=-[\operatorname{Hess}^{(-\alpha)}f(p)]^{-1}
\operatorname{grad}f(p).
$
Under a local coordinate chart $(\varphi,\xi)$,
the corresponding Newton equation is
\[
\mathbf M^{\nabla^{(\alpha)}}(\xi)
\boldsymbol{\beta}(\xi)
=-\boldsymbol{\nabla}f(\xi),
\]
Consequently,
\begin{equation}
\langle\operatorname{grad}f(p),X_p\rangle
=-\boldsymbol{\nabla}f(\xi)^\top
[\mathbf M^{\nabla^{(\alpha)}}(\xi)]^{-1}
\boldsymbol{\nabla}f(\xi),
\end{equation}
where $\boldsymbol{\nabla}f(\xi)$ denotes
the coordinate gradient of $f$.

If $\mathbf M^{\nabla^{(\alpha)}}(\xi)$ is positive
definite, the Newton direction is a descent
direction whenever $\operatorname{grad}f(p)\neq0$.
Moreover, the positive definiteness of
$\mathbf M^{\nabla^{(\alpha)}}(\xi)$ characterizes
the positivity of the second derivative of $f$
along nonzero $\nabla^{(\alpha)}$-geodesic
directions. We next relate this condition
to $\alpha$-geodesic strict convexity.

\begin{definition}[$\alpha$-geodesic strict convexity]
Let $(\mathcal{M},g,\nabla^{(\alpha)},
\nabla^{(-\alpha)})$ be a statistical manifold,
and let $f\in C^2(\mathcal{M})$.
The function $f$ is said to be
$\alpha$-geodesically strictly convex
on an $\alpha$-geodesically convex
neighborhood $\mathcal{U}\subset\mathcal{M}$
if, for any distinct points $p,q\in\mathcal{U}$
and every $\nabla^{(\alpha)}$-geodesic
$\gamma^{(\alpha)}:[0,1]\to\mathcal{U}$
joining them,
\[
f(\gamma^{(\alpha)}(t))
<(1-t)f(p)+tf(q),
\quad t\in(0,1).
\]

A sufficient second-order condition for
$\alpha$-geodesic strict convexity on
$\mathcal{U}$ is
$
\left.
\frac{d^2}{dt^2}
f(\gamma^{(\alpha)}(t))
\right|_{t=0}>0,
$
for every $p\in\mathcal{U}$ and every
$\nabla^{(\alpha)}$-geodesic in
$\mathcal{U}$ satisfying
$\gamma^{(\alpha)}(0)=p$ and
$\dot{\gamma}^{(\alpha)}(0)=X_p\neq0$.

By the definition of the covariant Hessian,
this condition is equivalent to
\[
(\nabla^{(\alpha)}Df)_p(X_p,X_p)>0,
\quad
\forall p\in\mathcal{U},\quad
\forall X_p\in T_p\mathcal{M}\setminus\{0\}.
\]

Equivalently, in any local coordinate chart
$(\varphi,\xi)$, this condition can be written as
\[
\boldsymbol{\beta}(\xi^p)^\top
\mathbf M^{\nabla^{(\alpha)}}(\xi^p)
\boldsymbol{\beta}(\xi^p)>0,
\quad
\forall p\in\mathcal{U},\quad
\forall\boldsymbol{\beta}(\xi^p)\neq0,
\]
where $\xi^p=\varphi(p)$ and
$\boldsymbol{\beta}(\xi^p)$ denotes
the coordinate representation of $X_p$.
\end{definition}

Therefore, if we choose a value of $\alpha$
such that the covariant Hessian
$(\nabla^{(\alpha)}Df)_p$ is positive
definite for every $p\in\mathcal{U}$,
the Newton direction is well-defined
and is a descent direction at every
nonstationary iterate in $\mathcal{U}$.

\begin{figure}[t] 
  \centering
  \includegraphics[width=0.55\linewidth]{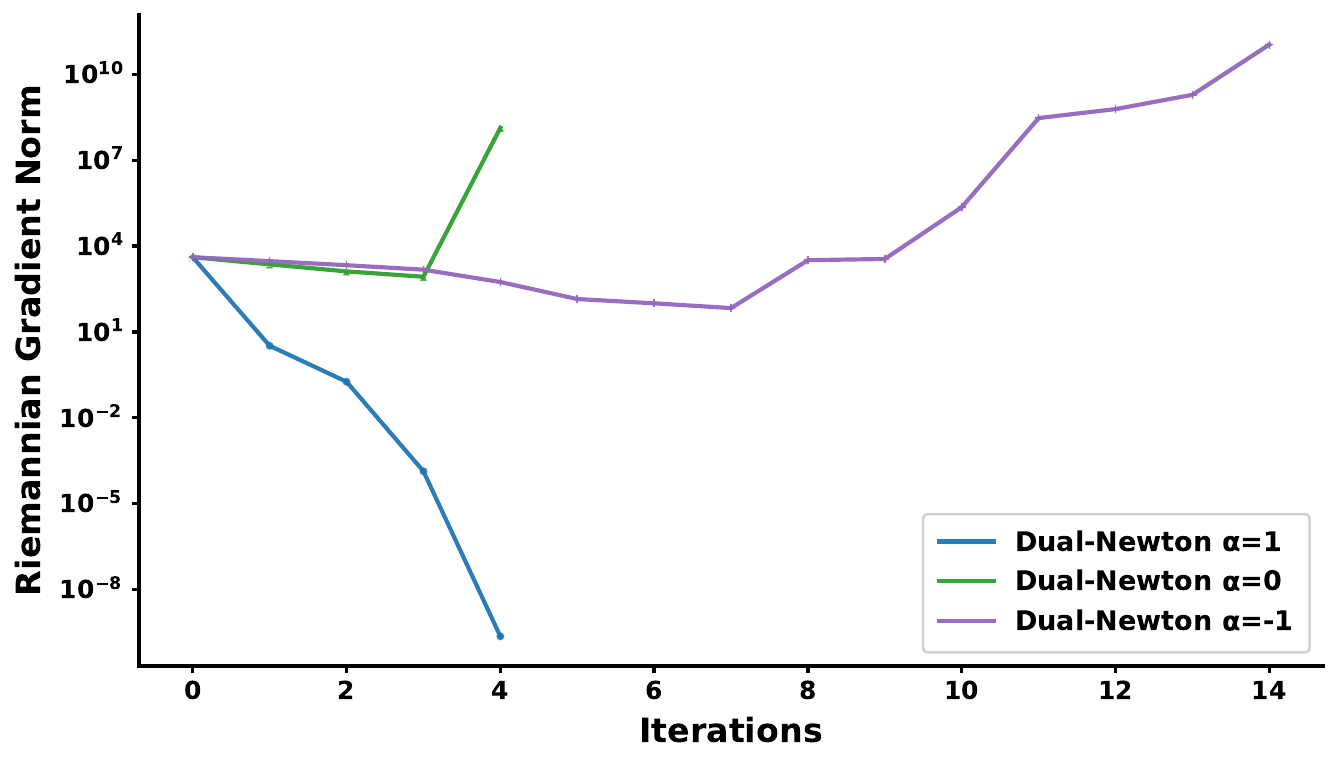} 
  \caption{The evolution of the Riemannian gradient norm versus iteration for each $\alpha$ connection. 
The results illustrate the convergence
behavior under different choices of
$\alpha$, consistent with the descent
property established for $\alpha=1$.
}
  \label{fig:non_convergence}
\end{figure}

For example, consider the objective function used in our first experiment:  
\[
f = D_{\mathrm{KL}}(\hat{\mathcal{P}}\|\mathcal{P})
+ \lambda_1 \sum_i (\theta^{i}_{\mathcal{P}})^2
+ \lambda_2 \sum_{i<j} (\theta^{ij}_{\mathcal{P}})^2.
\]
where \( \lambda_1 > 0 \) and \( \lambda_2 > 0 \).  

According to Theorem~\ref{euc_hession}, when
$\alpha=1$, the Newton coefficient matrix in
the $\nabla^{(1)}$-affine coordinates $\theta$
satisfies
$
\mathbf M^{\nabla^{(1)}}_{ij}(\theta)
=
\frac{\partial^2(f\circ\varphi_\theta^{-1})}
{\partial\theta_i\partial\theta_j}.
$
For the regularized KL divergence objective
considered above, the coordinate Hessian is
positive definite. Hence,
$
\mathbf M^{\nabla^{(1)}}(\theta)\succ0,
$
which implies that $f$ is
$1$-geodesically strictly convex and that
the Newton direction is a descent direction
at every nonstationary point. For $\alpha\neq1$, however, the matrix
$\mathbf M^{\nabla^{(\alpha)}}(\theta)$
is not necessarily positive definite,
and the corresponding Newton direction
may fail to be a descent direction.
For $\alpha=1$, a line-search strategy,
such as one satisfying the strong Wolfe
conditions, can be used to select an
appropriate step size.
In contrast, for $\alpha=0$ and $\alpha=-1$,
$\mathbf M^{\nabla^{(\alpha)}}(\theta)$
may become indefinite or even singular,
and the corresponding Newton direction
may fail to be a descent direction.
Figure~\ref{fig:non_convergence}
illustrates the different convergence behaviors
for $(\lambda_1,\lambda_2,n)=(0.3,0.8,3)$.



\section{Conclusion}
We have developed the \emph{dual Riemannian Newton method}
for optimization on statistical manifolds endowed with
a Riemannian metric and a pair of dual affine connections
$(\nabla,\nabla^*)$.
By deriving a second-order Taylor expansion
along a $\nabla$-based retraction, we obtain
a Newton equation involving the $\nabla^*$-Hessian
and establish its equivalence to the
affine-connection Newton equation associated with $\nabla$.
On dually flat manifolds, the proposed method reduces
to the Euclidean Newton method in affine coordinates.
For potential-based objectives, we further establish
its equivalence to natural gradient descent in primal
affine coordinates and, through Legendre duality, to
mirror descent in dual affine coordinates.
We prove local quadratic convergence under suitable
regularity conditions and establish local cubic convergence
for specific potential-based objectives under the
$\alpha$-connections with $\alpha=1/3$ in the
$\theta$-affine coordinates and $\alpha=-1/3$ in the
dual $\eta$-affine coordinates.
Numerical experiments on representative statistical models
support these theoretical findings and demonstrate
the computational effectiveness of the proposed method.
\section{Appendix}
\label{appendix}

In this appendix, we present the preliminaries of information geometry that were not included in the main text.

\begin{definition}[Affine Connection]\label{def:affine-connection}
Let \( \mathcal{M} \) be a smooth manifold and \( \mathfrak{X}(\mathcal{M}) \) be the set of all smooth vector fields on \( \mathcal{M} \).
An \emph{affine connection} on \( \mathcal{M} \) is a map
\[
\nabla: \mathfrak{X}(\mathcal{M}) \times \mathfrak{X}(\mathcal{M}) \to \mathfrak{X}(\mathcal{M}),
\quad (X, Y) \mapsto \nabla_X Y
\]
that satisfies the following properties for all \( X, Y, Z \in \mathfrak{X}(\mathcal{M}) \),  
\( f_1, f_2 \in C^\infty(\mathcal{M}) \), and \( a, b \in \mathbb{R} \):
\begin{enumerate}[label=\textnormal{(\arabic*)}, itemsep=0.4ex, topsep=0.6ex, leftmargin=*, align=left]
  \item $\nabla_{f_1 X + f_2 Y}\, Z = f_1\,\nabla_X Z + f_2\,\nabla_Y Z$,
  \item $\nabla_X ( a Y + b Z ) = a\,\nabla_X Y + b\,\nabla_X Z$,
  \item $\nabla_X (f_1\,Y) = (X f_1)\,Y + f_1\,\nabla_X Y$.
\end{enumerate}
\end{definition}

\begin{definition}[Differential of a Smooth Map Between Manifolds]
\label{def:manifold-map-differential}
Let $F:\mathcal{M}\to\mathcal{N}$ be a smooth
map between smooth manifolds.
The \emph{differential} of $F$ at
$p\in\mathcal{M}$, denoted by $DF(p)$,
is the linear map
$
DF(p):T_p\mathcal{M}\to T_{F(p)}\mathcal{N},
$
defined by
\[
DF(p)[X_p](f)=X_p(f\circ F),
\quad
\forall X_p\in T_p\mathcal{M},
\quad f\in C^\infty(\mathcal{N}).
\]
Here, $X_p(f\circ F)$ denotes the directional
derivative of $f\circ F$ at $p$ in the
direction $X_p$.
\end{definition}

\begin{definition}[Covariant Derivative]\label{def:covariant-derivative}
Let $\mathcal{M}$ be a smooth manifold with an affine connection $\nabla$, and let $\gamma: I \to \mathcal{M}$ be a smooth curve (with $I \subset \mathbb{R}$ an interval).  
An operator
$
\frac{D}{dt} : \mathfrak{X}(\gamma) \to \mathfrak{X}(\gamma)
$
is called the \emph{covariant derivative along $\gamma$} if, for any $Y,Z \in \mathfrak{X}(\gamma)$, $a,b \in \mathbb{R}$, and $f \in C^\infty(I)$, it satisfies all the following conditions:
\begin{enumerate}[itemsep=0.8ex, topsep=0.6ex, parsep=0pt]
  \item[\textnormal{(1)}]
  $\displaystyle \frac{D}{dt}(aY + bZ)=a\frac{D}{dt}Y + b\frac{D}{dt}Z$,
  \item[\textnormal{(2)}]
  $\displaystyle \frac{D}{dt}(f\,Z)=\frac{df}{dt}\,Z + f\frac{D}{dt}Z$,
  \item[\textnormal{(3)}]
  $\displaystyle \frac{D}{dt}(V\circ\gamma)(t)=\nabla_{\dot\gamma(t)}V
  \quad\text{for all } V \in \mathfrak{X}(\mathcal{M}).$
\end{enumerate}
\end{definition}





\section*{Acknowledgment}
This work was supported by JST, CREST Grant Number JPMJCR22D3,
and by Japan Society for the Promotion of Science (JSPS) Grant Number 23K11024.

\ifCLASSOPTIONcaptionsoff
  \newpage
\fi



%

\bibliographystyle{IEEEtran}
\bibliography{references}

\end{document}